\documentclass[11pt]{article} 
\usepackage[utf8]{inputenc} 

\usepackage[margin=1in]{geometry} 
\usepackage{graphicx} 
\usepackage[parfill]{parskip} 

\usepackage{parskip}
\usepackage{booktabs} 
\usepackage{array} 
\usepackage{paralist}
\usepackage{verbatim}  
\usepackage{titling}
\usepackage{enumitem}
\usepackage{amsmath}
\usepackage{amsthm}
\usepackage{amsfonts}
\usepackage{amssymb}
\usepackage{mathtools}
\usepackage{hhline}
\usepackage{breqn}
\usepackage{accents}
\usepackage{subcaption}
\usepackage{xfrac}
\usepackage{bbm}
\usepackage{physics}
\usepackage{csvsimple}
\usepackage{hyperref}
\usepackage{multirow}
\usepackage{float}
\usepackage{marvosym}
\usepackage[round]{natbib}
\usepackage{pdflscape}
\usepackage[scr=boondox]{mathalfa}
\usepackage{authblk}
\usepackage{xr}
\usepackage{setspace} 

\usepackage{xcolor}

\usepackage{sectsty}
\allsectionsfont{\sffamily\mdseries\upshape} 
\newcommand{\subtitle}[1]{%
  \posttitle{%
    \par\end{center}
    \begin{center}\large#1\end{center}
    \vskip0.5em}%
}

\usepackage[nottoc,notlof,notlot]{tocbibind} 
\usepackage[titles,subfigure]{tocloft} 

\newtheorem{lemma}{Lemma}
\newtheorem{theorem}{Theorem}
\newtheorem*{theorem*}{Theorem}
\newtheorem{corollary}{Corollary}

\theoremstyle{definition}
\newtheorem{definition}{Definition}
\newtheorem{remark}{Remark}

\newtheorem*{remark*}{Remark}
\newtheorem{example}{Example}
\newtheorem*{example*}{Example}
\newtheorem{assumption}{Assumption}

\newtheorem{condition2}{Condition}[section]

\newtheorem{conditionalt}{Condition}[assumption]

\newenvironment{conditionp}[1]{
  \renewcommand\theconditionalt{#1}
  \conditionalt
}{\endconditionalt}

\makeatletter
\def\thm@space@setup{%
  \thm@preskip=\parskip \thm@postskip=0pt
}
\makeatother

\newcommand{\R}{\mathbb{R}}

\newcommand{\E}{\mathbb{E}}

\title{On Gaussian Process Priors in Nonparametric Conditional Moment Restriction Models \thanks{Email address: \texttt{sid.kankanala@yale.edu}. I thank
	Xiaohong Chen for continuous guidance and support throughout this project. I thank  Donald Andrews, Yuichi Kitamura and seminar participants at Yale for valuable comments and suggestions that improved the quality of this paper. All errors are my own. }}
\author{Sid Kankanala}
\affil{Department of Economics, Yale University}
\date{\today}

\begin{document}
\maketitle

\begin{center}
\large

\end{center}

\bigskip

\begin{abstract}

	\noindent
	\normalsize
	This paper studies quasi-Bayesian estimation and uncertainty quantification for an
unknown function that is  identified by a nonparametric conditional moment restriction. We derive contraction rates for a class of Gaussian process priors. Furthermore, we provide conditions under which a Bernstein–von Mises theorem holds for the  quasi-posterior distribution. As a consequence, we show that  optimally-weighted quasi-Bayes credible sets have exact asymptotic frequentist coverage. This extends classical results on the frequentist validity of optimally weighted quasi-Bayes credible sets for parametric generalized method of moments (GMM) models.

	\bigskip


\end{abstract}

\newpage

\section{Introduction} \label{intro}
This paper considers a general conditional moment restriction model  where an unknown  function $h_0$ is identified from the restriction \begin{align}
    \label{intro-cmr} \E[\rho(Y,h_0(X))|W] = \mathbf{0}.
\end{align}
Here, $Y \in \R^{d_y}$ is a vector of observables,  $X \in \R^{d}$ is a vector of  regressors and $W \in \R^{d_w}$ is a vector of conditioning (or instrumental) variables. The vector $\rho(.) = [ \rho_{1}(.) , \dots , \rho_{d_{\rho}}(.)   ]' $ is a  known $d_{\rho}$ dimensional vector of generalized residual functions and $h_0(.)$ is an unknown function of interest. The model in (\ref{intro-cmr}) has a long history in econometrics and statistics (e.g. \citealp{liao2011posterior}; \citealp{chen2012estimation}). As a special case, it nests the nonparametric instrumental variable model (NPIV) studied in \citet{newey2003instrumental}; \citet{hall2005nonparametric}; \citet*{blundell2007semi} and the nonparametric quantile instrumental variable (NPQIV) model studied in \citet*{chernozhukov2007instrumental,horowitz2007nonparametric}. 

The conditional moment restriction model in  (\ref{intro-cmr}) represents a general class of  ill-posed inverse problems with an unknown and possibly nonlinear operator. In this setup, the (possibly nonlinear) operator $h \rightarrow m(W,h) = \E[\rho(Y,h(X))|W]$ smoothes out features of $h$. Moreover, the operator is not known as it depends on the true data generating process through the population conditional expectation operator $\E(.|W)$. One consequence of this, from a statistical perspective, is that feasible procedures which replace $m(W,h)$ with a finite sample analog $\widehat{m}(W,h)$ are sensitive to features of the generalized residual $\rho(.)$, that would otherwise be smoothed out in the case of a fully known operator. In particular, if $\rho(.)$ is pointwise nonsmooth or discontinuous in its arguments, as in the case of a NPQIV model, a feasible estimator $\widehat{m}(W,h)$ may inherit similar properties in finite samples.

In this paper, we study the quasi-Bayesian posterior distribution that arises from the conditional moment restriction in (\ref{intro-cmr}). If we denote the conditional mean of the residual at a function $h$ by $m(W,h) = \E[\rho(Y,h(X))|W]$, the unknown function is identified by the restriction $ \E \big( \| m(W,h_0) \|_{\ell^2}^2 \big) = 0 $. As such, estimation can be based on the finite sample objective function $h \rightarrow \E_n \big[    \widehat{m}(W,h) ' \widehat{\Sigma}(W)  \widehat{m}(W,h)            \big] $, where  $\widehat{m}(W,h)$ denotes a feasible estimator of $m(W,h)$ and $\widehat{\Sigma}(W)$ is a positive semi-definite weighting matrix. In a quasi-Bayes framework, we view this objective function as a pseudo-likelihood. When combined with a prior $\mu$, this leads to the quasi-posterior distribution
\begin{equation}
\label{posterior-intro}
\mu(.|\mathcal{Z}_n) =    \frac{\exp\big(    - \frac{n}{2}  \E_n \big[    \widehat{m}(W,.) ' \widehat{\Sigma}(W)  \widehat{m}(W,.)            \big]       \big) d \mu(.)  }{\int \exp\big(    - \frac{n}{2}  \E_n \big[    \widehat{m}(W,h) ' \widehat{\Sigma}(W)  \widehat{m}(W,h)            \big]       \big) d \mu(h) }         .
\end{equation} 
The quasi-posterior in (\ref{posterior-intro}) was proposed by \citet{liao2011posterior}. Using sieve based priors, they established posterior consistency. Posterior contraction rates, also using sieve based priors, were obtained in \citet{kato2013quasi} for the special case of a nonparametric instrumental variable model. The main results of this paper develop the limit theory for the quasi-Bayes posterior  when $\mu$ belongs to a class of sufficiently smooth Gaussian process priors. As a first step, our analysis extends the preceding work on contraction rates for the quasi-posterior  in that $(i)$ we establish contraction rates for the general conditional moment restriction model in (\ref{intro-cmr})  and $(ii)$ we do so using a class of  Gaussian process priors. Beyond contraction rates, we also provide conditions under which a Bernstein–von Mises  (BvM) theorem holds for the  quasi-posterior distribution. To be specific, if $h \rightarrow \mathbf{L}(h)$ is a sufficiently smooth linear functional, we show that the induced quasi-posterior distribution of $\mathbf{L}(h)$ can be well approximated (in the sense of weak convergence in probability) by a suitable Gaussian measure. We use this to provide frequentist guarantees for quasi-Bayesian credible sets that are centered around the posterior mean. In particular, we show that such credible sets have asymptotically exact frequentist coverage, provided that the quasi-Bayes posterior in (\ref{posterior-intro}) is optimally weighted. This extends classical results (e.g. \citealp{chernozhukov2003mcmc}) on the frequentist validity of optimally weighted quasi-Bayes credible sets for parametric generalized method of moments (GMM) models. 

Our approach to obtaining contraction rates is similar to the general strategy employed in  inverse problems (e.g. \citealp{knapik2018general}) in that we first
 obtain contraction rates for a suitable direct problem, i.e  $  d_{w}(h,h_0) = \| \Pi_n  \big[m(W,h) -   m(W,h_0)  \big] \|_{L^2(\mathbb{P})}  $ where $\Pi_n : L^2(\mathbb{P}) \rightarrow L^2(\mathbb{P})$ is a sample size dependent orthogonal  projection operator.  Similar to the approach in  \citet{giordano2020consistency}; \citet*{monard2021consistent,monard2021statistical}, we consider rescaled Gaussian process priors which concentrate (with high probability) on bounded subsets of a sufficiently smooth function class.  By combining this with contraction rates for the direct problem, we obtain contraction rates for stronger metrics such as $d(h,h_0) = \| h - h_0 \|_{L^2(\mathbb{P})}$. To obtain asymptotic Bernstein-von Mises Gaussian approximations, we  make use of posterior local fluctuation and change of parametrization arguments that are frequently used in the analysis of Gaussian process priors in density or regression frameworks (e.g. \citealp{castillo2015bernstein}; \citealp{monard2021statistical}). In the setup considered here, these arguments are suitably modified to account for $(i)$ a quasi-posterior as opposed to a traditional likelihood; $(ii)$ a general conditional moment restriction as in (\ref{intro-cmr}); $(iii)$ an unknown and possibly nonlinear operator $h \rightarrow m(W,h)$ that must be estimated as a first step and $(iv)$ a generalized residual function $\rho(.)$ that may be nonlinear and pointwise nonsmooth in its arguments.

This setting considered in this paper is closely related to several strands of literature that study statistical inverse problems in distinct setups. Rates of convergence and confidence sets for the model in (\ref{intro-cmr}) was established in \citet{chen2012estimation, chen2015sieve} under a frequentist sieve based framework. There is a large literature (e.g. \citealp*{knapik2011bayesian}; \citealp*{agapiou2013posterior}; \citealp{florens2016regularizing}; \citealp*{gugushvili2020bayesian}) that studies Gaussian priors within the context of linear inverse regression models with Gaussian noise and a known linear operator. In these cases, the model is conjugate with a known Gaussian posterior distribution and statistical properties can be analyzed directly. In the nonlinear case, \citet{monard2021statistical} consider Bayesian inference with Gaussian process priors on a class of nonlinear inverse regression models with Gaussian noise and a known nonlinear operator. Similar to their analysis, we consider additional posterior regularization obtained through scaling the Gaussian process. In our conditional moment setting, this scaling is crucial to $(i)$ control the nonlinearity and ill-posedness of the inverse problem and $(ii)$  obtain high probability guarantees on the behavior of the first stage estimator $\widehat{m}(W,h)$ used to approximate the unknown nonlinear operator $h \rightarrow m(W,h)$.

The paper is organized as follows. Section \ref{review} provides a brief review of Gaussian process priors. Section \ref{sec3} introduces the quasi-Bayes framework and main assumptions. Section \ref{sec4} develops the quasi-Bayes limit theory and main theoretical results. Section \ref{sec6} contains proofs and auxiliary results for all the statements in the main text.

\subsection{Notation} \label{notation}
Given positive sequences $(x_n)_{n=1}^{\infty} $and $ (y_n)_{n=1}^{\infty}$, we write $x_n \lessapprox y_n$  if $\limsup_{n \rightarrow \infty} x_n/y_n < \infty $ and $x_n \asymp y_n$ if $x_n \lessapprox y_n \lessapprox x_n$.  Given a positive definite matrix $\Sigma \in \R^{k \times k }$, the induced inner product and norm on $\R^k$ is denoted by $  \langle ., . \rangle_{\Sigma} $    and  $ \| . \|_{\Sigma}$,  respectively. That is, $ \langle u , v \rangle_{\Sigma} = u' \Sigma v$. The  Euclidean norm $(\Sigma = I_k)$ is denoted by $\| .\|_{\ell^2}$. Let $\| . \|_{\infty}$ denote the usual supremum norm on functions and vectors. We use $\E$ and $\mathbb{P}$ to denote the usual expectation and probability operators. Let $\E_n$ and $\mathbb{P}_n$ denote the empirical analog of $\E$ and $\mathbb{P}$, respectively. Given a random vector $Z$, let $L^2(Z)$ denote the usual $L^2$ space of $\R$ valued functions that are measurable with respect to the $\sigma$ algebra generated by $Z$. Similarly, the  $L^2$ space corresponding to $\R^d$ valued functions is denoted by $L^2(Z,\R^d)$. In both cases, the $L^2$ norm  is denoted by $\| f \|_{L^2(\mathbb{P})}^2 = \int \| f(Z) \|_{\ell^2}^2 d  \mathbb{P}$. Given a cube $\mathcal{X} \subset \R^d$, we use $ \mathbf{H}^{p} = ( \mathbf{H}^p(\mathcal{X}), \| . \|_{\mathbf{H}^p}) $ to denote the usual $p$-Sobolev space of functions on $\mathcal{X}$. The $L^2$ space with respect to the Lebesgue measure on $\mathcal{X}$ is denoted by $L^2(\mathcal{X})$. We denote a $p$-Sobolev ball of radius $M > 0 $ by $\mathbf{H}^p(M) = \{ h \in \mathbf{H}^p : \| h \|_{\mathbf{H}^p} \leq M   \}$.

\section{Review} \label{review}

In this section, we briefly review Gaussian random elements and the related topic of generating covariance operators and Hilbert scales through self-adjoint operators. For further details on Gaussian process priors, we refer to \citet{ghosal2017fundamentals}.

\begin{definition}[Gaussian random elements] \label{gre}
Given a probability space $\Omega$ and a separable Banach space $(\mathbb{B}, \| . \|_{\mathbb{B}})$, we say $G : \Omega \rightarrow \mathbb{B}$ is a Gaussian random element if it is a Borel measurable mapping and the random variable $L(G)$ is normally distributed for every $L$ in the dual space $\mathbb{B}^*$ of $\mathbb{B}$. In the special case where $G$ can be viewed as a map into a separable Hilbert space $(\mathcal{H}, \langle . \: ,  . \: \rangle_{\mathcal{H}})$, we refer to its mean as the unique $\mu \in \mathcal{H}$ that satisfies $  \E[\langle G  , h \rangle ] = \langle \mu , h \rangle  $ for every $h \in \mathcal{H}$. The covariance operator of a mean-zero Gaussian random element on $(\mathcal{H}, \langle . \: ,  . \: \rangle_{\mathcal{H}})$ is the continuous, linear, compact self-adjoint operator $ \Lambda : \mathcal{H} \rightarrow \mathcal{H}$ that satisfies \begin{align}
    \label{covariance} \E[\langle G , h_1 \rangle \langle G , h_2 \rangle ] = \langle h_1 , \Lambda h_2 \rangle \; \; \; \; \; \; \forall \: h_1,h_2 \in \mathcal{H}.
\end{align}
\end{definition}

A Gaussian process over an index set $T$ is a stochastic process $\{ G_t : t \in  T \}$ such that the vector $(G_{t_1} , \dots , G_{t_k})$ is multivariate normal distributed, for every $t_1,\dots,t_k \in T$ and $k \in \mathbb{N}$. Furthermore, if $\{G_t :t \in T \}$ is a Borel measurable random element with sample paths in a separable subset of $\ell^{\infty}(T)$, then $G$ is a Gaussian random element in this space. The covariance of a Gaussian process $\{ G_t : t \in T \}$ may either refer to the operator $\Lambda$ in Definition \ref{gre} (when the sample paths lie in a Hilbert space) or to the covariance function $C(s,t) = \text{Cov}(G_s,G_t)  $.

Consider a mean-zero Gaussian process $G$ with realizations in a separable Hilbert space  $\mathcal{H}$ with covariance operator $\Lambda$. By the spectral theorem, there exists an orthonormal basis of eigenfunctions $(e_i)_{i=1}^{\infty} \subset \mathcal{H}$ that diagonalizes the operator $\Lambda$. Furthermore, if $\lambda_i $ denotes the non-negative eigenvalue associated to $e_i$, the sequence is of trace class $(\sum_{i=1}^{\infty} \lambda_i < \infty)$ and $G$ admits an expansion  of the form \begin{align}
    \label{gexpand} G \stackrel{d}{=}  \sum_{i=1}^{\infty} \sqrt{\lambda_i} Z_i e_i \;, \; \; \; \; \; \; \; \; Z_i \stackrel{i.i.d}{\sim} N(0,1).
\end{align}
Conversely, given any non-negative sequence $(\lambda_i)_{i=1}^{\infty}$ with $\sum_{i=1}^{\infty} \lambda_i < \infty$, the representation in (\ref{gexpand}) defines a Gaussian random element on $(\mathcal{H}, \langle . \: ,  . \: \rangle_{\mathcal{H}})$. As the following definition illustrates, this leads to a family of Gaussian random elements that differ only by a smoothness scale.
\begin{definition}[Sobolev norms and Gaussian Series] \label{gseries}
Suppose $(e_i)_{i=1}^{\infty} $ is an orthonormal sequence of a Hilbert space $(\mathcal{H}, \langle . \: ,  . \: \rangle_{\mathcal{H}})$ of functions over a domain $\mathcal{X} \subseteq \R^d$. Fix any $\beta   \in \R $. Given any function $f \in \mathcal{H}$ and its unique basis expansion $f = \sum_{i=1}^{\infty} \langle f , e_i \rangle_{\mathcal{H}} e_i$, we denote the $\beta$ Sobolev norm with respect to $(e_i)_{i=1}^{\infty}$ by \begin{align}
    \label{sobolev-norm} \| f \|_{\mathcal{H}^{\beta}}^2 = \sum_{i=1}^{\infty}  \left| \langle f , e_i \rangle_{\mathcal{H}} \right|^2 i^{2 \beta/d}.
\end{align}
For $\beta > 0 $, the norm $\| . \|_{\mathcal{H}^{\beta}}$ measures smoothness or regularity of $f$ with respect to the basis $(e_i)_{i=1}^{\infty} $, with higher values of $\beta$ leading to greater regularity. For example, if $\mathcal{H} = L^2[0,1]^d $ and $(e_i)_{i=1}^{\infty}$ are the usual Fourier basis, this coincides with the usual Sobolev norm. The subset of $\mathcal{H}$ for which the norm is finite is denoted by $\mathcal{H}^{\beta} = \{ f \in \mathcal{H} :  \| f \|_{\mathcal{H}^\beta} < \infty   \}$.\footnote{For $\beta > 0$, we define $\mathcal{H}_{-\beta}$ as the dual space of $ \mathcal{H}^{\beta}$. As any $f \in \mathcal{H}$ defines a continuous linear functional on $\mathcal{H}^{\beta}$, we can identify $\mathcal{H} \subset \mathcal{H}^{-\beta}$. Moreover, the operator norm of such an $f \in \mathcal{H}$ agrees with the definition $\| . \|_{\mathcal{H}^{- \beta}}$ in (\ref{sobolev-norm}).} For $\beta > 0$ and $f \in \mathcal{H}$, the norm $\| . \|_{\mathcal{H}^{-\beta}}$ is always finite. In particular, $\| . \|_{\mathcal{H}^{-\beta}}$  is a weaker norm on $\mathcal{H}$.\

Given any $\alpha > 0$, we can define a Gaussian random element on $\mathcal{H}$ via \begin{align}
    \label{sobolev-gauss} G = \sum_{i=1}^{\infty} i^{-(1/2 + \alpha/d)} Z_i e_i \;, \; \; \; \; \; \; \; \; Z_i \stackrel{i.i.d}{\sim} N(0,1).
\end{align}
For every $\beta < \alpha$, we have $\E \|  G  \|_{\mathcal{H}^\beta}^2 = \sum_{i=1}^{\infty} i^{-1 + 2(\beta - \alpha)/d } < \infty $. In particular, $ \mathbb{P}( \| G \|_{\mathcal{H}^\beta} < \infty)  = 1 $ and the sample realizations of the Gaussian random element can be viewed as being almost $\alpha$ regular with respect to $(e_i)_{i=1}^{\infty}$. The Reproducing Kernel Hilbert Space (RKHS) $\mathbb{H}$ of the Gaussian Process $G$ in (\ref{sobolev-gauss}) is the set of exactly $\alpha$ regular functions, in the sense that  \begin{align}
\label{RKHS-0}  \mathbb{H}  =  \bigg \{  g \in \mathcal{H}  : \|g \|_{\mathbb{H}}^2 =   \sum_{i=1}^{\infty} i^{ 1 +   2\alpha /d}  \left| \langle g ,  e_i \rangle_{\mathcal{H}}      \right|^2 < \infty               \bigg \} .
\end{align}
Intuitively, $\mathbb{H}$ determines the support and small ball concentration properties of $G$.
\end{definition}
The family $\{ \mathcal{H}^{\beta} : \beta \in \R \}$ in Definition \ref{gseries} is an example of a  Hilbert or smoothness scale. More generally, Hilbert scales can be generated through a densely defined self-adjoint operator $L$. The following definition clarifies this connection.

\begin{definition}[Hilbert Scale] \label{hil-scale}
Let $(\mathcal{H}, \langle . \: ,  . \: \rangle_{\mathcal{H}})$ be a Hilbert space and $ D(L) \subseteq \mathcal{H}$ an open dense subset of $\mathcal{H}$. Suppose $L : D(L) \subseteq \mathcal{H} \rightarrow  \mathcal{H}$ is an unbounded self-adjoint operator that is coercive, i.e $\langle L(x) , x \rangle_{\mathcal{H}} \geq  \gamma \| x \|^2$ for some $\gamma > 0 $ and all $x \in D(L)$. For $ k \in \mathbb{N}$, denote the domain of $L^k$ by $D(L^k) \subseteq \mathcal{H}$. Then the family $\{ L^k : k \in \mathbb{N}\}$ is defined on the dense subset $\mathcal{H}_{\infty} = \bigcap_{k=1}^{\infty} D(L^k)$. By spectral theory, $L^{s}$ can be defined as an operator on $\mathcal{H}_{\infty}$ for every $s \in \R$.\footnote{For more details on the functional calculus of self-adjoint operators, see e.g. \cite{reed2012methods}.} We can define an inner product and norm on $\mathcal{H}_{\infty} $ by  \begin{align}
    \label{hs-norm} \langle h,g \rangle_{\mathcal{H}^s} = \langle L^s h  \; ,  \; L^s g \rangle_{\mathcal{H}} \;, \; \| h \|_{\mathcal{H}^s} = \|  L^s h \|_{\mathcal{H}} \; \; \; \; \; \; \; \forall \: h,g \in \mathcal{H}_{\infty}.
\end{align}
Denote the completion of $\mathcal{H}_{\infty}$ with respect to $\| . \|_{\mathcal{H}^s}$ by $\mathcal{H}^s$. The family $\{ \mathcal{H}^s : s \in \R \}$ is referred to as (see e.g. \citealp{mair1996statistical}; \citealp{mathe2001optimal})  the  Hilbert scale generated by the operator $L$.
\end{definition}
If the operator $L$ in the preceding definition admits a compact self-adjoint inverse $L^{-1}$, we can define a Gaussian Process by viewing $L^{-s}$ for $s > 0$ as a covariance operator. In this case, the Gaussian Process can be expressed in a  similar form to (\ref{sobolev-gauss}), where $(e_i)_{i=1}^{\infty}$ are the eigenfunctions of $L^{-1}$. The following example illustrates the essential idea for a commonly used class of priors.
\begin{example*}[Mat\'ern Gaussian Priors]  For a cube $\mathcal{X} \subset \R^d$, the Mat\'ern covariance function is defined by \begin{align}
    C(x,y) =  \frac{2^{1-\alpha}}{\Gamma(\alpha)} \bigg(  \sqrt{2 \alpha} \frac{\| x-y \|}{l} \bigg)^{\alpha} B_{\alpha}\bigg( \sqrt{2 \alpha}  \frac{\| x-y \|}{l} \bigg) \; \; \; \; \forall \; x,y \in \mathcal{X} \; ,
\end{align}
where $ l > 0$ is a correlation length scale parameter, $\alpha > 0$ is a smoothness parameter, $\Gamma(.)$ is the Gamma function and $B_{\alpha}$ is the modified Bessel function of the second kind. The Mat\'ern Gaussian process $G_{\alpha}$ satisfies (see e.g. \citealp{borovitskiy2020matern}) the stochastic partial differential equation $$ \bigg( \frac{2 \alpha}{l^2} - \Delta  \bigg)^{ \frac{\alpha}{2} + \frac{d}{4}} G_{\alpha} =   \mathcal{W} \:, $$
where $\Delta$ is the Laplacian and $\mathcal{W}$ is Gaussian white noise, renormalized by a fixed constant.\footnote{To ensure invertability of the Laplacian on $\mathcal{X}$, one typically restricts the functions to satisfy certain Dirichlet or Neumann boundary conditions.} In particular, the covariance operator $\Lambda$ of $G_{\alpha}$ diagonalizes in the same eigenbasis as the Laplacian. As the eigenvalues $ (\kappa_i)_{i=1}^{\infty}  $ of the Laplacian  scale at rate $\kappa_i \asymp i ^{2/d}$, it follows that the eigenvalues $(\lambda_i)_{i=1}^{\infty}$ of the covariance operator $\Lambda$ scale at rate $ \lambda_i \asymp i^{-(1 +2 \alpha/d)}$. Suppose $(e_i)_{i=1}^{\infty}$ are eigenfunctions of the Laplacian (on $\mathcal{X} = [0,1]^d$ they coincide with the usual Fourier basis) and $ \{\mathcal{H}^{\beta} : \beta \in \R \}$ are the associated smoothness scales (as in Definition \ref{gseries}). It follows from the representation in (\ref{gexpand}) and Definition \ref{gseries} that the same paths of $G_{\alpha}$ are $\beta$ regular for every $\beta < \alpha$. In particular, a larger smoothness parameter $\alpha $  induces a more regular process.
\end{example*}

\begin{figure}[htbp]
  \centering

  \begin{subfigure}[b]{0.48\textwidth}
    \centering
    \includegraphics[width=\textwidth, height=4.5cm]{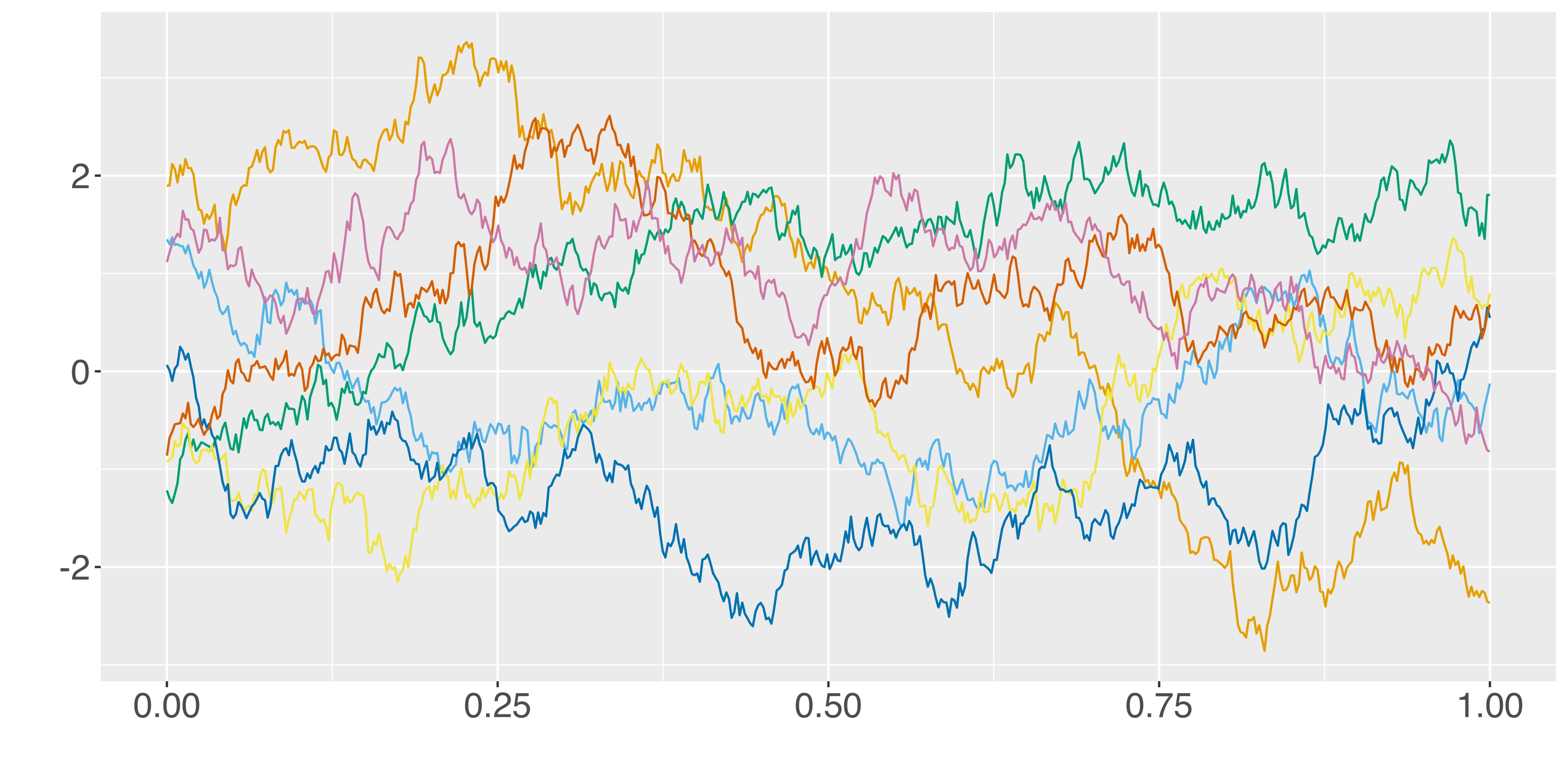}
    \caption{$\alpha=0.5$}
    \label{fig:figure1}
  \end{subfigure}
  \hfill
  \begin{subfigure}[b]{0.48\textwidth}
    \centering
    \includegraphics[width=\textwidth, height=4.5cm]{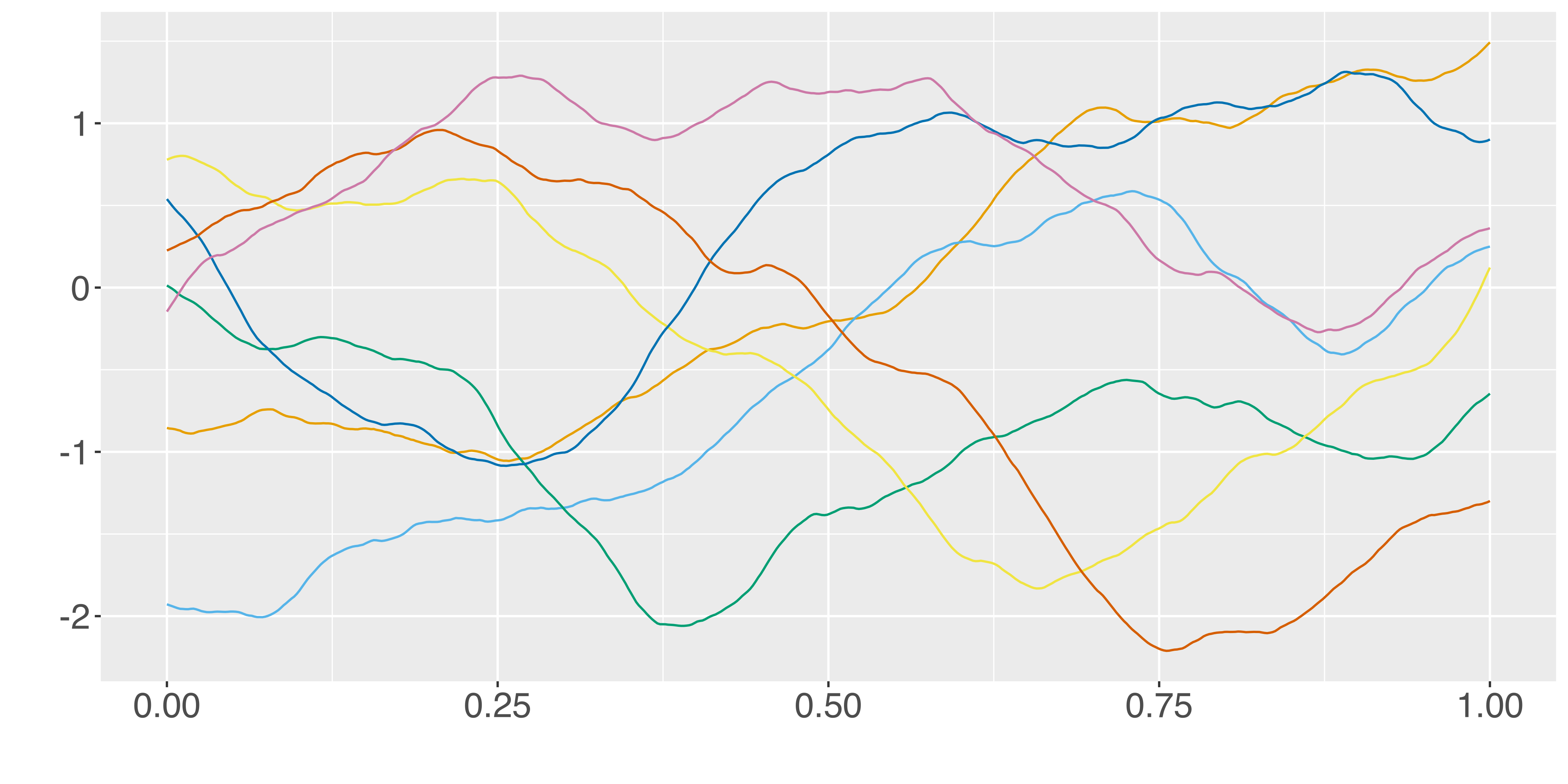}
    \caption{$\alpha=1.5$}
    \label{fig:figure2}
  \end{subfigure}
  
  \vspace{0.5cm} 
  
  \begin{subfigure}[b]{0.48\textwidth}
    \centering
    \includegraphics[width=\textwidth, height=4.5cm]{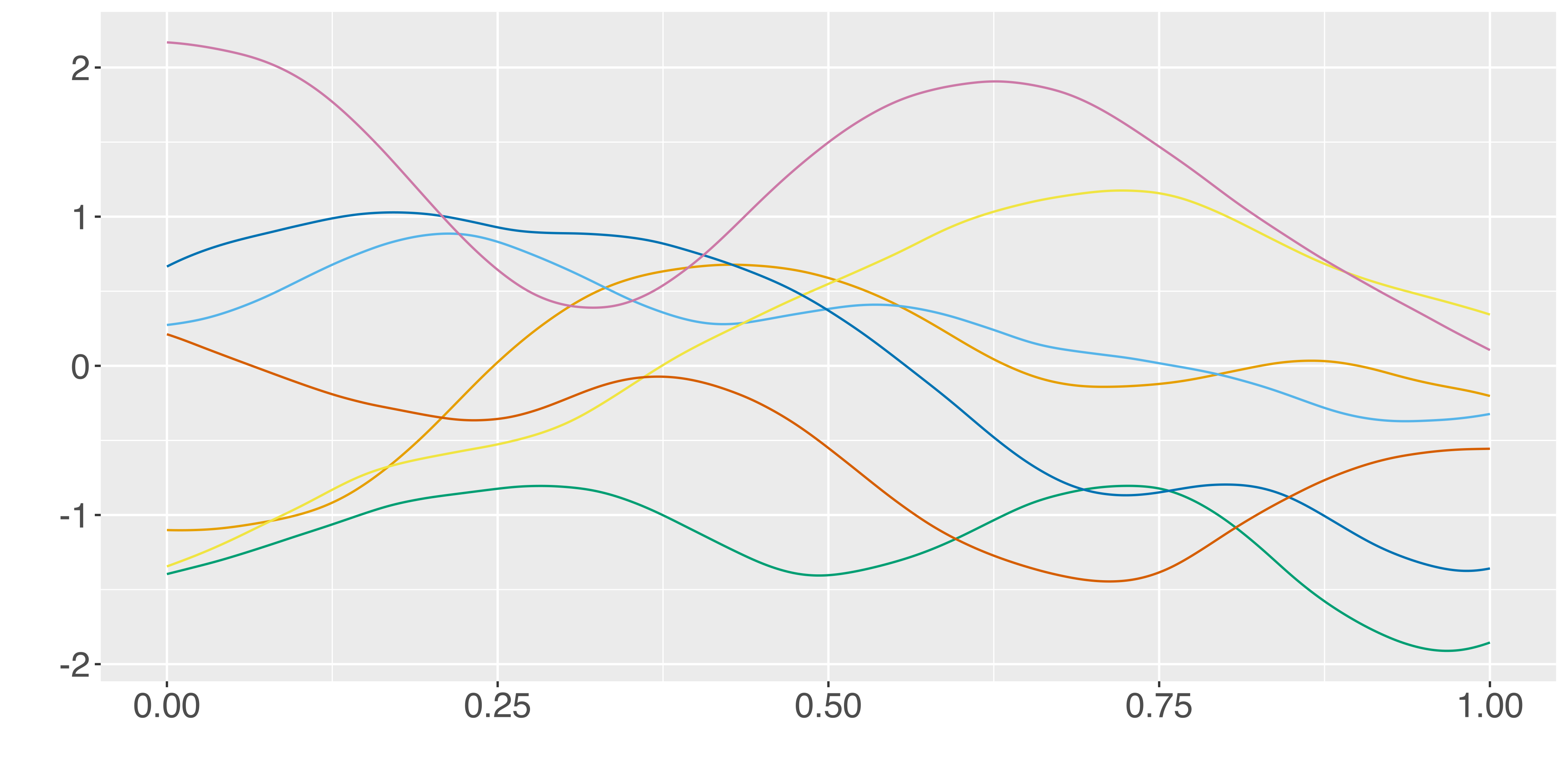}
    \caption{$\alpha=3$}
    \label{fig:figure3}
  \end{subfigure}
  \hfill
  \begin{subfigure}[b]{0.48\textwidth}
    \centering
    \includegraphics[width=\textwidth, height=4.5cm]{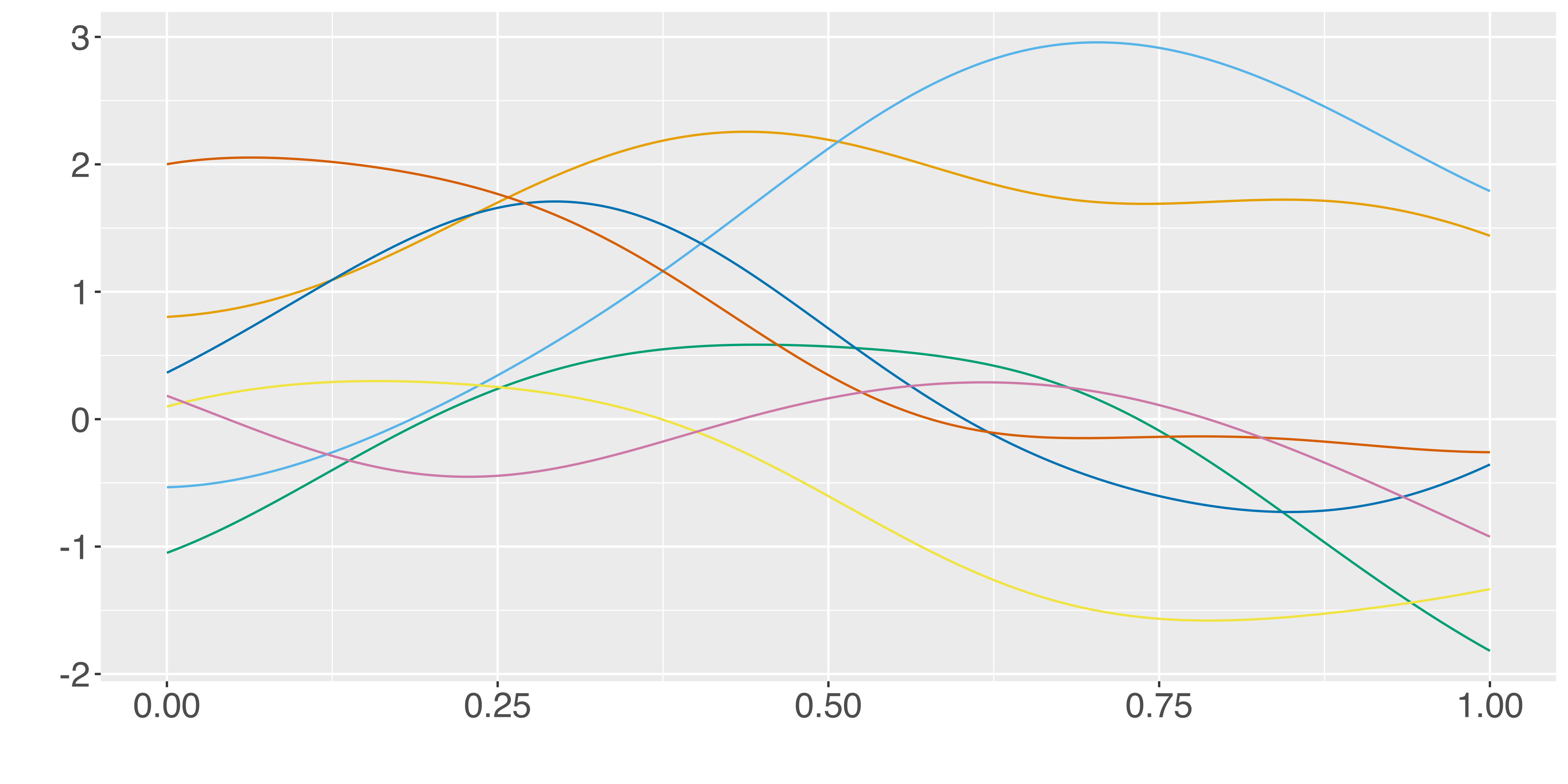}
    \caption{$\alpha=10$}
    \label{fig:figure4}
  \end{subfigure}

  \caption{Draws of a Mat\'ern Gaussian Process on $\mathcal{X} = [0,1]$ with varying regularity $\alpha$.}
  \label{fig1}
\end{figure}

\section{Framework and Assumptions} \label{sec3}
In this section, we introduce the conditional moment restriction (CMR) framework and state our main assumptions on the model.
\subsection{Model} \label{sec2model}
Suppose we observe a random sample $\mathcal{Z}_n = \{ (X_i,Y_i,W_i)  \}_{i=1}^n $, where $Y \in \R^{d_y}$ is a vector of observable variables,  $X \in \R^{d}$ is a vector of  regressors with support $\mathcal{X} \subset \R^d$ and $W \in \R^{d_w}$ is a vector of conditioning (or instrumental) variables with support $\mathcal{W} \subset \R^{d_w}$. We are interested in a structural function $h_0(X) : \mathcal{X} \rightarrow \R$ that is assumed to satisfy the conditional moment restriction:  \begin{equation}
\label{cmr}    \E \big[  \rho(Y,h_0(X)) \big| W      \big] = \mathbf{0} .
\end{equation}
Here,  $\rho(.) = [ \rho_{1}(.) , \dots , \rho_{d_{\rho}}(.)   ] $ is a vector of generalized residual functions with functional forms that are assumed known up to the structural function $h_0(.)$.

\begin{example} \label{ex1} [Nonparametric Instrumental Variables] The observed data consists of a real valued scalar $Y$, a vector $X$ of endogenous inputs and a vector $W$ of instrumental variables. The structural function of interest $h_0$ is identified by the conditional moment restriction \begin{align}
    \label{npiv} \E[Y - h_0(X)|W] = \mathbf{0}.
\end{align}
The generalized residual function is $\rho(Y,h(X)) = Y  - h(X)$. As a special case with $W=X$, the function of interest is the conditional mean $h_0(X) = \E[Y|X]$.
\end{example}

\begin{example}  \label{ex2} [Nonparametric Quantile IV] The observed data is as in Example \ref{ex1}. Given a quantile $\gamma \in (0,1)$, the structural function of interest $h_0$ is identified by the conditional moment restriction \begin{align}
    \label{npqiv} \E \big[ \mathbbm{1} \{ Y - h_0(X) \leq 0  \}    ] = \gamma.
\end{align}
The generalized residual function is $\rho(Y,h(X)) = \mathbbm{1} \{ Y - h(X) \leq 0 \} - \gamma$.
\end{example}
Many commonly used statistical and econometric models can be reformulated as a conditional moment restriction of the form in $(\ref{cmr})$. In Example \ref{ex1}, the generalized residual function is a linear function of $h_0$, while in Example \ref{ex2}, the residual function is  nonlinear and nonsmooth in $h_0$. Thus, Examples \ref{ex1} and \ref{ex2} may be seen as representative of two distinct classes of statistical models, characterized by regularity of the generalized residual function.

Given a function $h: \mathcal{X} \rightarrow \R$, we denote the conditional mean of the residual at $h$ by \begin{equation}
\label{condm} m(W,h) =   \E \big[  \rho(Y,h(X)) \big| W      \big].
\end{equation}
As the true function is assumed to satisfy $m(W,h_0) = \mathbf{0}$, the function can be identified as a minimizer of the objective function \begin{align} Q(h) =  \E[ m(W,h)' \Sigma(W) m(W,h) ] \:,            \label{cmr-true-obj} \end{align}
where $\Sigma(W) \in \R^{d_{\rho} \times d_{\rho}} $ is a suitable positive definite matrix. As the true data generating process is not assumed to be known, it is infeasible to work with $Q(h)$ directly. Denote by $\widehat{m}(W,h)$ any feasible estimator of $m(W,h)$ in  $ (\ref{condm}) $. The usual conditional moment restriction objective function is given by \begin{align}
    \label{cmr-obj} Q_n(h) =    \E_n \big[    \widehat{m}(W,h) ' \widehat{\Sigma}(W)  \widehat{m}(W,h)            \big] \; ,
\end{align}
where $\widehat{\Sigma}(W) = \widehat{\Sigma}_n(W)   $ is a (possibly data dependent) positive definite weighting matrix. In a quasi-Bayes framework, (\ref{cmr-obj}) is viewed as a psuedo-likelihood for the model. Given a prior probability measure $\mu$, the quasi-Bayes posterior induced from $\mu$ and (\ref{cmr-obj}) is denoted by \begin{equation}
\label{posterior-0}
\mu(.|\mathcal{Z}_n) =    \frac{\exp\big(    - \frac{n}{2}  \E_n \big[    \widehat{m}(W,.) ' \widehat{\Sigma}(W)  \widehat{m}(W,.)            \big]       \big) d \mu(.)  }{\int \exp\big(    - \frac{n}{2}  \E_n \big[    \widehat{m}(W,h) ' \widehat{\Sigma}(W)  \widehat{m}(W,h)            \big]       \big) d \mu(h) }         .
\end{equation}

\subsection{Assumptions} \label{assump}

In this section, we state our main assumptions on the model and data generating process.

\begin{assumption} \label{data}
$(i)$ $X$ has  support on a cube $\mathcal{X} \subset \R^d$ and the density of $X$ with respect to the Lebesgue measure is bounded away from $0$ and $\infty$ on $\mathcal{X}$. (ii) $W$ has support on a cube $\mathcal{W} \subset \R^{d_w} $  and the density of $W$  with respect to the Lebesgue measure is  bounded away from $0$ and $\infty$ on $\mathcal W$. 
\end{assumption}

Assumption \ref{data} is standard. We can always ensure $X,W$ are contained in cubes by transforming them through  suitable homeomorphisms $ \tilde{X} =  \Phi(X) $ and $ \tilde{W} = \Psi(W)$. The conditional moment  restriction in (\ref{intro-cmr}) then holds with  $\tilde{X} , \tilde{W}$ and $ \tilde{h}_0 = h_0 \circ \Phi^{-1}  $. The bounded density condition facilitates the analysis as it implies, among other things, that   $  \| f \|_{L^2(\mathbb{P})}^2$ and $f \rightarrow \int_{\mathcal{X}} \left| f(x) \right|^2 dx $ are equivalent metrics over $L^2(X)$.

\begin{assumption}\label{residuals}$(i)$  $  \| \rho(Y,h(X)) \|_{L^2(\mathbb{P})}  < \infty  $ for every $h \in L^2(X  )$.  $(ii)$ For some  $ \kappa  \in (0,1]$, $ t \geq d/ \kappa  $ and any $M < \infty$, there exists $C_1 = C_1(M) < \infty$ such that
\begin{align*}
 & \sup_{w \in \mathcal{W}} \E_{} \bigg(  \sup_{h \in  \mathbf{H}^t(M) :  \| h' - h \|_{ \infty } \leq \xi}  \|   \rho(Y, h(X) ) - \rho(Y,h'(X))            \|_{\ell^2}^2 \big|W=w   \bigg) \leq C_1^2  \xi^{2 \kappa}  , \\ & \sup_{h \in  \mathbf{H}^t(M) :  \| h' - h \|_{ L^2(\mathbb{P}) } \leq \xi} \;  \sup_{w \in \mathcal{W}} \E \bigg(   \|   \rho(Y, h(X) ) - \rho(Y,h'(X))            \|_{\ell^2}^2 \big|W=w    \bigg) \leq C_1^2 \xi^{2 \kappa}
\end{align*}
holds for all $h' \in \mathbf{H}^{t}(M)$ and $\xi > 0$ small enough.
\end{assumption}
Assumption \ref{residuals} is  similar to conditions frequently imposed in the literature (e.g. \citealp*{chen2003estimation}) to facilitate analysis involving non smooth objective functions. In particular, it allows for a pointwise discontinuous residual function $\rho(.)$. However, it requires that the residual function be  uniformly continuous in $L^2(\mathbb{P})$ expectation. The parameter $\kappa$  determines the modulus of continuity. It holds with  $\kappa = 1$ for the NPIV model (Example \ref{ex1}) and $\kappa=1/2$ for the NPQIV model (Example \ref{ex2}).

Before stating the remaining assumptions, we fix any sufficiently large  $t \geq d/\kappa$ that satisfies Assumption  \ref{residuals}. 

\begin{assumption}\label{residuals2} 

There exists $\epsilon, \delta > 0 $  such that for any $M > 0 $, there exists finite constants $C_2(M) , C_3(M), C_4(M) < \infty$ that satisfy

\begin{align*} &  (i) \; \; \; \;  \sup_{w \in \mathcal{W}} \E \bigg(  \sup_{h \in \mathbf{H}^{t}(M)}  \|  \rho(Y,h(X))  \|_{\ell^2}^2 \big| W=w        \bigg) \leq C_2^2 \;  , \\ & (ii) \; \; \; \; \E \bigg(  \sup_{h \in \mathbf{H}^{t}(M)}  \|  \rho(Y,h(X))  \|_{\ell^2}^{2+ \epsilon}       \bigg) \leq C_3^2 \; , \\ & (iii)  \; \; \; \; \mathbb{P} \bigg(  \sup_{h,h' \in \mathbf{H}^{t}(M) : \|h - h' \|_{L^2(\mathbb{P})} \leq \delta  }  \|  \rho(Y,h(X)) - \rho(Y,h'(X))  \|_{\ell^2} \leq C_4 \bigg) = 1.
\end{align*}
\end{assumption}
Assumption $\ref{residuals2}$ imposes weak moment bounds on the residual function. The assumption is trivially satisfied with bounded residual functions, such as in a NPQIV model. For more general cases, observe that if $t > d/2$, the Sobolev inequality \citep[5.6.3]{evans2022partial} implies that $\mathbf{H}^t$ embeds into a H\"{o}lder space. In particular, the functions in $\mathbf{H}^t(M)$ are bounded in $\| . \|_{\infty}$ norm. In most cases, this can be used to verify  Assumption $\ref{residuals2}$ directly.

\begin{assumption}\label{identif}
$(i)$ There exists a unique $h_0 \in \mathbf{H}^t$ that satisfies the conditional moment restriction $\E_{}[m^2(W,h_0)] = 0$. $(ii)$ For every $M > 0$, there exists a constant $C_5(M) < \infty$ such that $\| m(W,h) - m(W,h') \|_{L^2(\mathbb{P})} \leq C_5  \| h - h' \|_{L^2(\mathbb{P})}$ holds for every $h,h' \in \mathbf{H}^t(M)$.
\end{assumption}

Assumption \ref{identif}$(i)$  is a standard identification condition for the conditional moment restriction model. We consider a relaxation of this assumption in Section \ref{sect-cons}. Assumption \ref{identif}$(ii)$ imposes that the conditional mean function $h \rightarrow m(W,h)$ is Lipschitz over any fixed Sobolev ball $\mathbf{H}^t(M)$ of radius $M$. This is made for convenience and can be relaxed further.

\begin{assumption} \label{svd}
(i) The conditional mean function  is Fr\'echet differentiable as a map $m(W, \, . ) : (L^2(X ), \| . \|_{L^2(\mathbb{P})} ) \rightarrow(L^2(W,\R^{d_{\rho}} ), \| . \|_{L^2(\mathbb{P})} )    $ at $h_0$.  (ii) The Fr\'echet derivative at $h_0$ is a compact injective operator $D_{h_0} : (L^2(X ), \| . \|_{L^2(\mathbb{P})} ) \rightarrow(L^2(W,\R^{d_{\rho}} ), \| . \|_{L^2(\mathbb{P})} ) $. 
\end{assumption}
Assumption \ref{svd} is a mild differentiability restriction. As we illustrate in Section \ref{sec4}, this allows us to study the behavior of the map $h \rightarrow m(W,h)$ in a local neighborhood around $h_0$ through properties of its  linearization $h \rightarrow D_{h_0}[h]$.

\section{Main Results} \label{sec4}

\subsection{Setup}
We denote by $b^K(W) =  ( b_1(W), \dots , b_K(W)   )'$ a vector of first stage basis functions. Here, $K \in \mathbb{N}$ denotes the dimension of the basis. Denote by $\mathcal{V}_K$, the linear subspace of $L^2(W)$ spanned by the basis. Let $\Pi_K : L^2(W) \rightarrow  \mathcal{V}_K  $ denote the $L^2(\mathbb{P})$ orthogonal projection onto $\mathcal{V}_K$. Given a function $h: \mathcal{X} \rightarrow \R$, we estimate the conditional mean function $m(W,h) = \E[\rho(Y,h(X))|W]$ using the empirical analog of $\Pi_K$. Given the observed  data $\mathcal{Z}_n = \{ (X_i,Y_i,W_i)  \}_{i=1}^n $, we estimate the conditional expectation using a least squares projection  onto $\mathcal{V}_K$. That is, \begin{align} \label{cmean}  &  \widehat{m}(w,h)  =      \E_n \big[   \rho(Y,h_{}(X)) \big(b^K(W) \big)'    \big]               [ \widehat{G}_{b,K} ]^{-1}   b^K(w)     \; , \\ &  \text{where} \; \; \;    \widehat{G}_{b,K} = \E_n \big[
 \big(b^K(W) \big) \big( b^K(W) \big)'    \big]      .        \nonumber
\end{align}
This choice is not crucial towards obtaining contraction rates but, as a closed form expression, greatly facilitates the analysis in obtaining inferential results.

As a first step towards defining our prior probability measure $\mu$, we consider a family of Gaussian process priors $G_{\alpha} = \{ G_{\alpha}(x) : x \in \mathcal{X}  \}$ that are indexed by a regularity hyperparameter $\alpha \in \mathcal{L} \subset \R_+$. As in Section \ref{review}, we assume that the regularity hyperparameter $\alpha$ influences the family through the coefficients on a series expansion with respect to particular basis or through the exponent of an unbounded self-adjoint operator (e.g. the Laplacian) on $\mathcal{X}$. In either case, without loss of generality\footnote{If the mapping $\alpha \rightarrow \lambda_{i,\alpha} $ influences the exponent in a different way, the results can be stated in terms of the induced exponent $ s(\alpha) $, i.e $\lambda_{i,\alpha} \asymp i^{-s(\alpha)}$.}, we can express $G_{\alpha}$ as
 \begin{align}
  \label{g-series-4}  G_{\alpha} = \sum_{i=1}^{\infty}  \sqrt{\lambda_{i,\alpha}}  Z_i e_i 
\end{align}
where  $\lambda_{i,\alpha} \asymp i^{-(1+ 2 \alpha/d)} $,  $    Z_i \stackrel{i.i.d}{\sim} N(0,1)$ and $(e_i)_{i=1}^{\infty}$ is an orthonormal basis of $L^2(\mathcal{X}).$

While we do not impose any restrictions on the eigenbasis $(e_i)_{i=1}^{\infty}$ directly, we will require that the sample paths of the Gaussian process $G_{\alpha}$ (for $\alpha \in \mathcal{L}$) belong to a separable linear subspace of the Sobolev space $\mathbf{H}^t$ for some $t > d/2$.\footnote{By the Sobolev inequality \citep[5.6.3]{evans2022partial}, $t > d/2$ implies that same paths of $G$ belong, almost surely, to a H\"{o}lder space. In particular, this ensures that the sample paths are continuous. As the  H\"{o}lder space is separable with respect to the $\| . \|_{\infty}$ norm, $G$ can be viewed as a Gaussian random element on $(C(\mathcal{X}), \|. \|_{\infty})$.} In fact, due to the possible non smoothness of the residual function $\rho$, we will generally require more stringent conditions (see Assumption \ref{residuals}) on the minimal such $t$. In most cases, this can be viewed as a restriction on the hyperparameter set by considering $ \mathcal{L}  \subseteq [\underline{\alpha}, \infty)$ for some minimum regularity $\underline{\alpha} > 0$. Given a Gaussian process $ G_{\alpha}$ of regularity $\alpha$ and first stage sieve dimension $K$, we consider the prior distribution \begin{align} \label{prior1} \mu(.|\alpha,K) \sim \frac{G_{\alpha}}{  \sqrt{\log n}  \sqrt{  K}} . \end{align}
We scale the Gaussian process by the first stage sieve dimension (up to a log term) to provide additional regularization. Let $\widehat{\Sigma}(W) $ denote  a (possibly data dependent) positive semidefinite weighting matrix. The quasi-Bayes posterior induced from $(\mu, \widehat{\Sigma} )$ and the conditional moment restriction in (\ref{cmr}) is given by
 \begin{equation}
\label{posterior}
\mu(.|\alpha,K,\mathcal{Z}_n) =    \frac{\exp\big(    - \frac{n}{2}  \E_n \big[    \widehat{m}(W,.) ' \widehat{\Sigma}(W)  \widehat{m}(W)            \big]       \big) d \mu(.|\alpha,K)  }{\int_{} \exp\big(    - \frac{n}{2}  \E_n \big[    \widehat{m}(W,h) ' \widehat{\Sigma}(W)  \widehat{m}(W,h)            \big]       \big) d \mu(h|\alpha,K) }         .
\end{equation}

In the analysis that follows, we will frequently measure regularity with respect to
the orthonormal basis  $(e_i)_{i=1}^{\infty}$ in (\ref{g-series-4}). To that end, we define the p-Sobolev space and p-Sobolev ball relative to the orthonormal basis $(e_i)_{i=1}^{\infty}$ by 
\begin{align} \label{sob} &  \mathcal{H}^p = \bigg\{  h \in L^2(\mathcal{X}) : h  = \sum_{i=1}^{\infty}  c_i   e_i   \; , \;   \| h \|_{\mathcal{H}^p}^2 =    \sum_{i=1}^{\infty} i^{2p/d}  c_i^2  < \infty     \bigg \} \\ &  \mathcal{H}^p(M) = \{ h \in \mathcal{H}^p : \| h \|_{\mathcal{H}^p} \leq M   \} .  \nonumber \end{align}

\subsection{Consistency} \label{sect-cons}
In this section, we establish consistency of the quasi-Bayes posterior in (\ref{posterior}). From the prior representation in (\ref{prior1}), it is expected that posterior limit theory depends on some interplay between sample path realiziations of the Gaussian process $G_{\alpha}$ and the quasi-Bayes objective function. Intuitively, given a function class $\mathcal{D}$, if the conditional moment restriction in (\ref{cmr})   arises from a nonlinear and/or nonsmooth residual function $\rho(.)$, further (smoothness) restrictions on $\mathcal{D} $ are necessary to quantify the uniform sampling uncertainty of the map $ \mathcal{D} \ni h \rightarrow \E_n\big[ \widehat{m}(W,h)'  \widehat{\Sigma}(W) \widehat{m}(W,h) \big]$. In our setting, it suffices for the Gaussian process to be contained (with high probability) in a sufficiently regular function class $\mathcal{D}$. This is formalized in the following condition.

\begin{condition2} \label{gp-structure}
The Gaussian process $\{G_{\alpha}(x) : x \in \mathcal{X}  \}$ in (\ref{g-series-4})  is a Gaussian random element (in the sense of Definition \ref{gre}) on a separable subspace of the Sobolev space $\mathbf{H}^t$, where $t$ is as in Assumption \ref{residuals}.

Beyond sampling uncertainty, consistency also depends on the weighting matrix $\widehat{\Sigma}(.)$ and the  first stage basis functions $b^K(W) = (b_1(W),\dots,b_K(W))'$ used to construct the  estimator $\widehat{m}$. At a minimum, we will impose the following structure.
\end{condition2}

\begin{condition2} \label{fsbasis} $(i)$ The matrix  $G_{b,K} = \E \big( [b^K(W)] [ b^K(W) ]'    \big)$ is positive definite for every $K$ and  $\zeta_{b,K} =  \sup_{w \in \mathcal{W} }  \|  G_{b,K}^{-1/2} b^K(w) \|_{\ell^2}  \lessapprox \sqrt{K}$. $(ii)$ The eigenvalues of $\widehat{\Sigma}(W)$ are asymptotically bounded away from $0$ and $\infty$: $\mathbb{P}\big(  c \leq   \lambda_{\min}(  \widehat{\Sigma}(W) ) \leq    \lambda_{\max}(  \widehat{\Sigma}(W) )   \leq C   \big)  \rightarrow 1$ for some $0 < c \leq C < \infty$. $(iii)$ For any fixed $ M > 0 $, we have  $ \sup_{h \in \mathbf{H}^t(M)  } \|(\Pi_K - I) m(W,h)   \|_{L^2(\mathbb{P})} \rightarrow 0 $ as  $K \rightarrow \infty$.  
\end{condition2}
If Assumption \ref{data}$(i)$ holds, Condition \ref{fsbasis}$(i)$ is satisfied by splines, Cohen–Daubechies–Vial (CDV) wavelets and Fourier series (see e.g. \citealp{chen2015optimal}; \citealp{belloni2015some}).

\begin{theorem}[Consistency] \label{t1}  Suppose Assumptions \ref{data}-\ref{identif} and Conditions \ref{gp-structure}, \ref{fsbasis} hold. Let $(K_n)_{n=1}^{\infty}$ denote any sequence that satisfies $ n^{d/2(\alpha + d)}  \lessapprox K_n  $ and $ \log(n) K_n = o(n)$. If $h_0 \in \mathcal{H}^p$ for some $p \geq \alpha + d/2$, the quasi-Bayes posterior is consistent at $h_0$. That is, \begin{align}
    \label{const-eq} \mu(h : \| h - h_0  \|_{L^2(\mathbb{P})} > \epsilon \:\big| \: \alpha,K_n,\mathcal{Z}_n) = o_{\mathbb{P}}(1) \; \; \; \; \forall \: \epsilon > 0.
\end{align} 
\end{theorem}
Theorem \ref{t1} shows that the quasi-Bayes posterior is consistent provided that the regularity of the true function exceeds that of the Gaussian process by a factor of $d/2$. The upper bound constraint on $K_n$ is very weak, it ensures that the least squares estimator $\widehat{m}(w,h)$ is well defined and that it uniformly approximates its population analog $\Pi_K m(w,h)$. On the other end, the theorem does enforce a strict lower bound on how slowly the first stage basis can grow.\footnote{The restriction on $(K_n)_{n=1}^{\infty}$ can be weakened even further in settings where the conditional mean function $m(W,h) = \E[\rho(Y,h(X))|W]$ is known to smooth out features of $h$ in a neighborhood of $h_0$. In particular, Theorem \ref{t1} requires no conditions (except the weak contraction requirement of Assumption \ref{identif}$(ii)$) on the smoothing properties (or ill-posedness) of the mapping $h \rightarrow m(W,h)$ in a local neighborhood around $h_0$.} Intuitively, large values of $K_n$ induce greater sampling uncertainty but also act as a form of regularization by shrinking the Gaussian process in (\ref{prior1}). This regularization is crucial to control the ill-posedness in the model.

Theorem \ref{t1} can be extended in several ways. One possibility is to consider a continuously updated version of the quasi-Bayes objective function. In this case, the data dependent weighting matrix $\widehat{\Sigma}$ may depend pointwise  on both $W$ and the prior realization $h$, i.e $ \widehat{\Sigma} = \widehat{\Sigma}(W,h)$.\footnote{The usual continuously updated objective function takes $\widehat{\Sigma}(W,h)$ to be a suitable estimator of $$ \Sigma(W,h) = \{ \E[ \rho(Y,h(X)) \rho(Y,h(X))'|W  ] \}^{-1} .$$} In this case, the associated quasi-Bayes posterior is \begin{align}
    \label{posteriorcu}  \mu^{CU}(.|\alpha,K,\mathcal{Z}_n) =    \frac{\exp\big(    - \frac{n}{2}  \E_n \big[    \widehat{m}(W,.) ' \widehat{\Sigma}(W,.)  \widehat{m}(W,.)            \big]       \big) d \mu(.|\alpha,K)  }{\int_{} \exp\big(    - \frac{n}{2}  \E_n \big[    \widehat{m}(W,h) ' \widehat{\Sigma}(W,h)  \widehat{m}(W,h)            \big]       \big) d \mu(h|\alpha,K) }  .
\end{align} Another possible avenue, as in \citet{liao2011posterior}, is  to generalize the contraction in Theorem \ref{t1} to settings where the unknown function $h_0$ is not uniquely identified from the data. In this setting, the identified set is given by $\Theta_0 = \{ h    : \| m(W,h) \|_{L^2(\mathbb{P})} = 0\}$. Intuitively, regardless of point identification, draws from the quasi-Bayes posterior should concentrate in areas where the quasi-Bayes objective function is minimized, i.e around the identified set $\Theta_0$. Below, we state a version of Theorem \ref{t1} that accommodates both of the extensions discussed above. To that end, we impose the following analog of Condition \ref{fsbasis}.

\begin{conditionp}{\ref*{fsbasis}$^*$} \label{fsbasis2}
$(i)$ The matrix  $G_{b,K} = \E \big( [b^K(W)] [ b^K(W) ]'    \big)$ is positive definite for every $K$ and  $\zeta_{b,K} =  \sup_{w \in \mathcal{W} }  \|  G_{b,K}^{-1/2} b^K(w) \|_{\ell^2}  \lessapprox \sqrt{K}$. $(ii)$ Over any Sobolev ball, the eigenvalues of $  \widehat{\Sigma}(W,h) $ are asymptotically bounded away from $0$ and $\infty$:  For every $ M > 0$, there exists constants $c,C > 0$ such that $ \mathbb{P} \big(  c \leq \inf_{h \in \mathbf{H}^t(M)} \lambda_{\min}( \widehat{\Sigma}(W,h)) \leq \sup_{h \in \mathbf{H}^t(M)} \lambda_{\max}(\widehat{\Sigma}(W,h)) \leq C    \big) \rightarrow 1$. $(iii)$ For any fixed $ M > 0 $, we have  $ \sup_{h \in \mathbf{H}^t(M)  } \|(\Pi_K - I) m(W,h)   \|_{L^2(\mathbb{P})} \rightarrow 0 $ as  $K \rightarrow \infty$.  
\end{conditionp}

\begin{theorem}[Identified Set Consistency]  \label{t1-2} Let $ \Theta_0 = \{ h \in L^2(\mathcal{X})  : \| m(W,h) \|_{L^2(\mathbb{P})} = 0  \} $ denote the identified set. Suppose Assumptions \ref{data}-\ref{residuals2} and Condition \ref{gp-structure},  \ref{fsbasis2} holds. Let $(K_n)_{n=1}^{\infty}$ denote any sequence that satisfies $ n^{d/2(\alpha + d)}  \lessapprox K_n  $ and $ \log(n) K_n = o(n)$. If there exists some $h_0 \in \Theta_0 \cap \mathcal{H}^{p}$ for $p \geq \alpha+d/2$ that satisfies Assumption \ref{identif}$(ii)$, then the  continuously updated quasi-Bayes posterior $\mu^{CU}(.)$ in (\ref{posteriorcu}) is consistent for the identified set. That is, \begin{align}
    \label{const-eq-2} \mu^{CU}(h :  d(h,\Theta_0)  > \epsilon \:\big| \: \alpha,K_n,\mathcal{Z}_n) = o_{\mathbb{P}}(1) \; \; \; \; \forall \: \epsilon > 0
\end{align} 
where $d(h,\Theta_0) = \inf_{h' \in \Theta_0} \| h-h' \|_{L^2(\mathbb{P})} $.
\end{theorem}
Theorem \ref{t1-2} shows that a continuously updated quasi-Bayes posterior is consistent, provided that at least one element of the identified set has sufficient regularity relative to the sample paths of the Gaussian process.

\begin{remark}
Towards verifying Condition \ref{fsbasis2}$(ii)$, consider the usual case where $\widehat{\Sigma}(w,h)$ is  uniformly (over $\mathbf{H}^t(M)$ and $w \in \mathcal{W}$) consistent  for $ \Sigma(w,h) = \{ \E[ \rho(Y,h(X)) \rho(Y,h(X))'|W =w ] \}^{-1}$. In Example \ref{ex1} (NPIV), we have $\Sigma^{-1}(W,h) = \E[u^2|W] +  \E[ \big( h(X)- h_0(X) \big)^2|W  ] $. As functions in $\mathbf{H}^t(M)$ are bounded in $\| . \|_{\infty}$ norm (for $t > d/2$), Condition \ref{fsbasis2}$(ii)$ holds if the conditional variance $ \sigma^2(w) =  \E[u^2|W=w]$ is bounded above and below. In Example \ref{ex2} (NPQIV), we have $\Sigma^{-1}(W) = \mathbb{P} (u \leq h_0(X) - h(X) |W   ) $. In this case, Condition \ref{fsbasis2}$(ii)$ holds if the conditional distribution $u|W$ has full support on $\R$.
\end{remark}
For the remainder of Section \ref{sec4}, we focus on the setting where the true structural function $h_0$ is uniquely identified. Furthermore, unless otherwise specified, all further limit theory is developed with the quasi-Bayes posterior in (\ref{posterior}).\footnote{Extensions to the continuously updated version in (\ref{posteriorcu}) can be handled in a similar manner to Theorem \ref{t1-2}. }

\subsection{Contraction Rates} \label{rates}

In this section, we further develop the quasi-Bayes limit theory. From the concluding statements in the proof of Theorem \ref{t1}, we can deduce that the quasi-Bayes posterior concentrates on bounded sample paths that lie in shrinking local neighborhoods around $h_0$. To be specific, there exists a sequence $\delta_n \rightarrow 0$ and $M > 0$ sufficiently large such that
\begin{align}
   &     \mu( h \notin \Omega_n \: | \alpha,K_n,\mathcal{Z}_n) = o_{\mathbb{P}}(1) \; , \; \;  \Omega_n(M)  = \{  h \in \mathbf{H}^t(M)  : \| h - h_0 \|_{L^2(\mathbb{P}))} \leq \delta_n  \}. \label{nhd}
\end{align}
As the focus in the preceding section was on general consistency, the analysis did not lead to an explicit form for $\delta_n$. In this section, we improve the preceding consistency results by explicitly quantifying the posterior contraction rate $\delta_n$. As we illustrate below, the posterior contraction rate depends on the interplay between $(i)$ Sample path realizations of the Gaussian process prior, $(ii)$ The local curvature of the pseudo-likelihood that defines the quasi-Bayes posterior, $(iii)$ the precise smoothing properties of the operator obtained from linearizing $h \rightarrow m(W,h)$  in a sufficiently regular local neighborhood around the true structural function $h_0$ and $(iv)$ the sequence of basis functions $b^K(W) = (b_1(W) , \dots , b_K(W))'$ used to compute a feasible estimate of the residual conditional mean function $m(W,h) = \E[ \rho(Y,h(X))|W]$.

Given the consistency result in (\ref{nhd}), to determine the rate of convergence, it suffices to restrict our attention to sample path realizations that are sufficiently regular and lie in a local neighborhood around $h_0$. Locally around $h_0$, we approximate the behavior of the map $h \rightarrow m(W,h)$ through a linearized version of it. Depending on the model and assumed hypothesis on the data generating process for $ \mathcal{Z} = (Y,X,W)$, there may be several distinct maps that serve as a reasonable linearization.\footnote{As the existence of such a map  primarily serves as a proof technique, the precise choice is not crucial.} If the map $h \rightarrow m(W,h)$ is sufficiently regular around $h_0$, the natural linearization to consider is the Fr\'echet differential at $h_0$. This is the unique continuous linear operator $D_{h_0} : L^2(X) \rightarrow L^2(W)$ that satisfies 
\begin{align}
    \label{dh0}   \| m(W,h_0+h) - m(W,h_0) - D_{h_0}[h] \|_{L^2(\mathbb{P})} = o(\| h\|_{L^2(\mathbb{P})}) \; \; \; \; \text{as} \; \; \; \; \|h \|_{L^2(\mathbb{P})} \rightarrow 0.
\end{align}
Intuitively, if $D_{h_0}[h]$ closely approximated $m(W,h)$ in local neighborhood around $h_0$, the smoothing properties of the map $h \rightarrow m(W,h)$ can be analyzed through the  simpler linear map $h \rightarrow D_{h_0}[h]$. In the analysis that follows, we relate the smoothing property of $h \rightarrow D_{h_0}[h]$ to a change in regularity relative to the  orthonormal basis $(e_i)_{i=1}^{\infty}$ that defines the Gaussian process in (\ref{g-series-4}). As the smoothness of $h_0$ is also defined relative to this basis through membership in the Sobolev ball (\ref{sob}), this will allow us to study the action of $D_{h_0}$ on $(G_{\alpha},h_0)$  under a  common regularity scale.

It will be convenient in our analysis to define a family of weak norms on $L^2(\mathcal{X})$, all obtained by shrinking the Fourier coefficients of a general function with respect to the basis $(e_i)_{i=1}^{\infty}$. To that end, we employ the following definition.

\begin{definition}[Weak Norms] \label{weak-n}
Let $ \sigma = ( \sigma_i)_{i=1}^{\infty}$ denote a non-negative sequence with $\sigma_i \rightarrow 0$. Given any function $h \in L^2(\mathcal{X})$ with basis expansion $h = \sum_{i=1}^{\infty} \langle h , e_i \rangle e_i$, we define the weak norm
\begin{align}
   h \rightarrow  \| h \|_{w,\sigma}^2 = \sum_{i=1}^{\infty}   \sigma_i^2 \left|  \langle  h , e_i  \rangle \right|^2 .     \label{w-norm}
\end{align}
\end{definition}

The following two conditions quantify the smoothing action of the map $h \rightarrow m(W,h)$ in a local neighborhood around $h_0$ by relating it to a weak norm in (\ref{w-norm}). For $\alpha \in \mathcal{L}$ and $\gamma \in (0,\alpha)$ sufficiently small, denote a smooth local neighborhood around $h_0$ by \begin{align}
    \label{nhd-smooth} \Omega(M,\epsilon,\gamma) =  \{ h \in \mathbf{H}^t(M) \cap \mathcal{H}^{\alpha - \gamma}(M)  : \| h - h_0 \|_{L^2(\mathbb{P})} \leq \epsilon    \}.
\end{align}
where $\mathcal{H}^{p}$ is defined as in (\ref{sob}) for every $ p \geq 0$. The $\gamma$ in (\ref{nhd-smooth}) is used to account for the fact that a Gaussian process $G= G_{\alpha}$ as in (\ref{g-series-4}) does not possess exact regularity $\alpha$. The process has nearly exact regularity $\alpha$ in the sense that $ \E \big (\| G_{\alpha}  \|_{\mathcal{H}^{\alpha}} \big) = \infty$  but $  \E \big (\| G_{\alpha}  \|_{\mathcal{H}^{\beta}} \big) < \infty  $ for every $\beta < \alpha$.

\begin{condition2}[Smoothing Link Condition] \label{gp-link}
 $(i)$ There exists $\epsilon \in (0,1), \gamma \in (0, \alpha)$  and a non-negative sequence of constants  $ \sigma = (\sigma_i)_{i=1}^{\infty}$ such that for any $ M > 0 $, there exists constants $C_1(M),C_2(M) < \infty$ that satisfy $ \|  D_{h_0}[h- h_0]  \|_{L^2(\mathbb{P})} \leq C_1 \| h - h_0  \|_{w,\sigma} $ and $ \| h - h_0  \|_{w,\sigma}  \leq C_2 \|  D_{h_0}[h- h_0]  \|_{L^2(\mathbb{P})}   $ for every $h \in \Omega(M,\epsilon,\gamma) $. $(ii)$ The model is either mildly or severely ill-posed in the sense that $$ \sigma_{i}  \asymp \begin{cases}    i^{- \zeta/d} &  \text{mildly ill-posed} \\  \exp( - R i ^{\zeta/d}) & \text{severely ill-posed}   \end{cases}  $$
for some $R, \zeta \geq 0$.
\end{condition2}

\begin{condition2}[Local Curvature] \label{lcurv} 
 There exists $\epsilon \in (0,1), \gamma \in (0, \alpha)$ such that for any $ M > 0$, there exists a constant  $ B = B(M) < \infty $  that satisfies  $\|  m(W,h)  \|_{L^2(\mathbb{P})} \leq   B \| D_{h_0}[h- h_0] \|_{L^2(\mathbb{P})} $  and $ \|  D_{h_0}[h-h_0]  \|_{L^2(\mathbb{P})} \leq    B \|  m(W,h)  \|_{L^2(\mathbb{P})}  $ for every $h \in \Omega(M,\epsilon,\gamma)$.
\end{condition2}
Condition \ref{gp-link} and \ref{lcurv} are similar to Assumption $4.1$ and $5.2$ in \cite{chen2012estimation}. Condition $\ref{lcurv}$ is trivially satisfied if $h \rightarrow m(W,h)$ is a linear map, as in NPIV (Example \ref{ex1}). If $D_{h_0}^*$ denotes the adjoint of $D_{h_0}$, Condition \ref{gp-link} is  satisfied if the self-adjoint operator $D_{h_0}^* D_{h_0}$ diagonalizes in the same eigenbasis $(e_i)_{i=1}^{\infty}$ that defines the Gaussian process in (\ref{g-series-4}). In this special case, Condition \ref{gp-link} coincides with assumptions frequently used in the literature (e.g. \citealp*{knapik2011bayesian}, Assumption 3.1) for linear inverse problems. A stronger version of Condition \ref{gp-link} is also used in \citet*{gugushvili2020bayesian} within the context of a white noise model with a known linear operator.

\begin{remark}[On Variations of Local Curvature Conditions]
   Variations of of Condition \ref{lcurv} can be used without requiring any significant changes to our analysis. For example, (\citealp*{chernozhukov2023constrained}, Remark A.2.3) and (\citealp*{dunker2014iterative}, Theorem 2) imposes (in our notation) a local curvature condition between the quantity $\| \Pi_K m(W,h) \|_{L^2(\mathbb{P})}$ and $\| \Pi_K D_{h_0}[h-h_0]  \|_{L^2(\mathbb{P})}$ for all sufficiently large $K$. In this setting, our Condition \ref{basis-approx} stated below would instead impose an order for the local linear bias  $ \gamma(K) =   \sup_{h \in \mathcal{H}^{\alpha}(M) \cap \mathbf{H}^t(M) : \| h - h_0 \|_{L^2(\mathbb{P})} \leq \epsilon   } \| (\Pi_K - I) D_{h_0}[h-h_0]  \| _{L^2(\mathbb{P})}$. 
\end{remark}

The quasi-Bayes objective function uses the estimator $\widehat{m}(w,h)$ in (\ref{cmean}) as a feasible analog for the true conditional mean function $m(W,h) = \E[ \rho(Y,h(X))|W]$. The discrepancy between $\widehat{m}(W,h)$ and $m(W,h)$ depends on two factors: $(i)$ the stochastic error between $\widehat{m}(W,h)$ and the population projection  $\Pi_K m(W,h)$ and $(ii)$ the distance between the projection $\Pi_K m(W,h)$ and the true $m(W,h)$. While  the effect of $(i)$ can be quantified\footnote{By Lemma \ref{emp-c}, under Condition \ref{fsbasis}$(i)$, the stochastic error has order $\sqrt{K}/ \sqrt{n}$ uniformly over any Sobolev ball $\mathbf{H}^t(M)$.} under very weak conditions for a large set of basis functions, the effect of $(ii)$ is intrinsic to the choice of first stage basis functions $b^K(W) = (b_1(W) , \dots , b_K(W))'$ used in the projection. Our final condition quantifies the approximation properties of the basis in a smooth local neighborhood around $h_0$.

\begin{condition2}[Basis Approximation] \label{basis-approx} Let $\sigma_K$ be as in Condition \ref{gp-link}. There exists $\epsilon \in (0,1)$ and a sequence of (possibly diverging) constants $ \lambda_K    < \infty$  such that for any $ M > 0 $, there exists  $ B = B(M) < \infty $ that satisfies \begin{align}  \sup_{h \in \mathcal{H}^{\alpha}(M) \cap \mathbf{H}^t(M) : \| h - h_0 \|_{L^2(\mathbb{P})} \leq \epsilon   }   \|  (\Pi_{K} - I) m(W,h)  \|_{L^2(\mathbb{P})}  \leq \lambda_K   B \sigma_{K+1}  K^{- \alpha /d}   \label{biascondition}  \end{align} for all $K $ sufficiently large and  $\alpha \in \mathcal{L}$. 
\end{condition2}
Condition \ref{basis-approx} is similar to assumptions commonly imposed  in the literature.\footnote{For example, Corollary 5.1 of \citep{chen2012estimation} imposes that the bias in  the left side of (\ref{biascondition}) is upper bounded by the stochastic projection error, having order $\sqrt{K}/ \sqrt{n}$. As the optimal $K$ there is chosen to balance the stochastic projection error and the projection bias  $\sigma_{K+1} K^{-\alpha/d}$, this is equivalent to assuming $\lambda_K \lessapprox 1$.} Locally around $h_0$, the map $h \rightarrow m(W,h)$ acts as a smoothing operator that is similar to $D_{h_0}[h-h_0]$. The quantity $\sigma_{K+1} K^{-\alpha/d}$ represents the projection bias of the smoothed function.\footnote{This is the bias of $\| (\Pi_K - I) D_{h_0}[h-h_0]  \| _{L^2(\mathbb{P})}$ if  $b^K(W) = (b_1(W) , \dots , b_K(W))'$ 
 are chosen to be the eigenfunctions generated by the singular value decomposition of $D_{h_0}$.} The sequence $\lambda_K$ in Condition \ref{basis-approx}  represents a slack factor to account for rates that may be slightly larger than the usual local linear bias $\| (\Pi_K - I) D_{h_0}[h-h_0]  \| _{L^2(\mathbb{P})}$, possibily due to a non optimal choice of first stage basis  functions.

Given any Gaussian process $G_{\alpha}$ of regularity $\alpha$, let  $ K_n (\alpha)$ denote a sequence of sieve dimensions defined by \begin{align}
\label{kn} K_n = \sup  \bigg \{ K \in \mathbb{N} :  \frac{  \sqrt{K}}{\sqrt{n}}  \leq   \sigma_K K^{-\alpha/d}       \bigg      \}.
\end{align}
 Intuitively, $(K_n)_{n=1}^{\infty}$ denotes the optimal sieve dimension that balances the stochastic projection error and the local (around $h_0$) projection bias. In particular, if the model is mildly ill-posed, we have $  K_n \asymp n^{d/[2(\alpha + \zeta) + d]}  $. If the model is severely ill-posed, we have $K_n \asymp (\log n)^{d/\zeta}$.

The following result provides contraction rates for the quasi-Bayes posterior in (\ref{posterior}).

\begin{theorem}[Contraction Rates] \label{rate}
Suppose Assumptions \ref{data}-\ref{svd} and Conditions \ref{fsbasis}-\ref{basis-approx} hold. Denote by $K_n$, the sequence defined in (\ref{kn}). Let $\nu_{n} = \lambda_{K_n}$ where $\lambda_K$ is as in Condition \ref{basis-approx}. Suppose $h_0 \in \mathcal{H}^p$ for some $p \geq \alpha + d/2$ and $\nu_{n} \lessapprox n^{c}$ for some $c < 1/2$.
\begin{enumerate}
\item[$(i)$] If the model is mildly ill-posed, there exists a universal constant $ D > 0 $ such that \begin{align} \label{mild-c-rate}
   \mu \bigg( \| h -  h_0 \|_{L^2}   >  D \nu_{n} n^{\frac{-\alpha}{2[\alpha +\zeta] +d}}    \sqrt{\log n}    \: \big| \: \alpha,K_n,  \mathcal{Z}_n  \bigg)   = o_{\mathbb{P}}(1).
\end{align}
\item[$(ii)$] If the model is severely ill-posed, there exists a universal constant $ D > 0 $ such that
\begin{align} \label{sev-c-rate}
 \mu \bigg( \| h -  h_0 \|_{L^2}   >  D   (\log n)^{- \alpha / \zeta} \sqrt{\log \log n}     \: \big| \: \alpha,K_n,  \mathcal{Z}_n  \bigg)   = o_{\mathbb{P}}(1).
\end{align} 
\end{enumerate}
\end{theorem}
As a point estimator for $h_0$, we consider the posterior mean \begin{align}
\label{posmean}   \E \big[ h|\mathcal{Z}_n \big] = \int  h d \mu(h|\alpha,K_n,\mathcal{Z}_n). 
\end{align}

Given the posterior contraction rate in Theorem \ref{rate}, we expect that the posterior mean, as a point estimator, achieves a similar rate of convergence. Intuitively, this follows from the preceding result if the posterior tail probabilities appearing in (\ref{mild-c-rate}) and (\ref{sev-c-rate}) decay sufficiently fast. The following result verifies this.

\begin{corollary}
\label{posgp} 
Suppose the hypothesis of Theorem \ref{rate} holds.
\begin{enumerate}
\item[$(i)$] If the model is mildly ill-posed, there exists a universal constant $ D > 0 $ such that \begin{align*}
    \mathbb{P} \bigg(  \|  h_0 -  \E \big[ h|\mathcal{Z}_n \big]  \|_{L^2}  >    D \nu_{n}  n^{\frac{-\alpha}{2[\alpha +\zeta] +d}}    \sqrt{\log n}       \bigg)  \rightarrow 0.
\end{align*}
\item[$(ii)$] If the model is severely ill-posed, there exists a universal constant $ D > 0 $ such that
\begin{align*}
   \mathbb{P} \bigg(  \|  h_0 -  \E \big[ h|\mathcal{Z}_n \big]   \|_{L^2}  >    D   (\log n)^{- \alpha / \zeta} \sqrt{\log \log n}     \bigg)  \rightarrow 0.
\end{align*}
\end{enumerate}
\end{corollary}
The contraction rates for the mildly ill-posed case can be improved further if one includes the slack factor $\lambda_K$ from Condition \ref{basis-approx} as part of the projection bias in the definition of $K_n$ in (\ref{kn}). If $ \nu_{n} = \lambda_{K_n}$ is asymptotically
negligible relative to the other factors, the resulting rates will not change significantly from those stated (without $\nu_{n}$) in Theorem \ref{rate} and Corollary \ref{posgp}. For ease of exposition and notation, in the remaining sections we focus on the case where $\nu_{n} \lessapprox 1$.\footnote{It is straightforward to incorporate $\nu_{n} \uparrow \infty$ into the analysis, as in the preceding results.}

\subsection{Inference}
In this section, we establish the limiting quasi-posterior distribution for a class of linear functionals of $h$. We verify that, under suitable regularity conditions, the quasi-posterior distribution of a fixed linear functional is asymptotically Gaussian. Moreover, quasi-Bayesian credible sets  have asymptotically exact frequentist coverage, provided that the quasi-Bayes posterior in (\ref{posterior}) is optimally weighted.

Given the posterior contraction rate in Theorem \ref{rate}, to derive the distributional limit theory, it suffices to restrict our analysis to a quasi-Bayes posterior with support contained on shrinking local neighborhoods around $h_0$. Let $\delta_n$ denote the posterior contraction rate in Theorem \ref{rate}. On route to deriving the posterior contraction rate, the previous results also provide contraction rates with respect to weaker metrics such as  $ d_{w}(h,h_0) =  \|   m(W,h_{}) - m(W,h) \big]   \|_{L^2(\mathbb{P})}  $. It will be convenient in our analysis to emphasize this when defining the posterior support. To that end, define \begin{align*}\xi_n =    \begin{cases}      n^{- \frac{ \alpha + \zeta }{2[\alpha + \zeta] + d}}  \sqrt{\log n}    & \text{mildly ill-posed} \; , \\  \sqrt{\log n} (\log n)^{d/(2 \zeta)} n^{-1/2}   & \text{severely ill-posed}.    \end{cases}
\end{align*}
Fix any $\gamma \in (0,\alpha)$ sufficiently small.\footnote{It is fine to let $\gamma   \rightarrow 0$ slowly here. In fact, the posterior concentrates on subsets of $ \mathbf{H}^t(M) \cap \mathcal{H}^{\alpha - \gamma_n} (M) $, where $$  \gamma_n \asymp \begin{cases} (\log n)^{-1} & \text{mildly ill-posed} \\ (\log \log n)^{-1} & \text{severely ill-posed}  
    \end{cases} $$} For constants $M,D  > 0$ sufficiently large, we denote the localized posterior support and its image under the map $h \rightarrow m(W,h)$ by \begin{align}
\label{thetan-support}   \Theta_n =  \bigg \{ h  \in \mathbf{H}^t(M) \cap \mathcal{H}^{\alpha - \gamma} (M)  & :  \|    m(W,h_{})  - m(W,h_0)   \|_{L^2(\mathbb{P})}  \leq   D  \xi_n ,        \| h- h_0 \|_{L^2(\mathbb{P})} \leq D \delta_n       \bigg \} \; ,  \\   \nonumber &  \mathcal{M}_n = \{  m(W,h) : h \in \Theta_n\}.
\end{align}
We denote the entropy integral of the image of the localized posterior support by \begin{align}
    \label{entropy-int}  \mathcal{J}(c) =  \int_{0}^{c}  \sqrt{\log N( \mathcal{M}_n , \| . \|_{L^2(\mathbb{P})}, \tau D \xi_n  )  } d \tau \; \; \; \; \; \; \; \; \forall  \; c > 0 .
\end{align}

To connect with the usual linear distributional theory, we quantify the discrepancy between  $m(W,h)$ and its linear approximation $D_{h_0}[h-h_0]$ locally around $h_0$. To that end, given any function $h : \mathcal{X} \rightarrow \R$, we denote the remainder obtained from linearizing the map $h \rightarrow m(W,h)$  locally around $h_0$ by  \begin{align}
    \label{linear-remainder} R_{h_0}(h,W) =  m(W,h) - m(W,h_0) - D_{h_0}[h-h_0] .
\end{align}
For linear problems such as NPIV (Example \ref{ex1}), we have $R_{h_0}(h,W) = 0$ for every $h$. As such, including (\ref{linear-remainder}) in the analysis is only relevant for nonlinear models. Analogous to the Euclidean case, the remainder vanishes as $\| h- h_0 \|_{L^2(\mathbb{P})} \rightarrow 0$. The precise rate at which this occurs depends on (among other factors) $(i)$ the ill-posedness in the model, the regularity of $h$ and $(iii)$ the convergence rate of $\| h- h_0 \|_{L^2(\mathbb{P})}$. Our main conditions on the localized support and remainder are as follows.

\begin{condition2}\label{misc1}  Let $\kappa,t$ be as in Assumption \ref{residuals}. Suppose that \begin{align}
 &(i) \; \; \; \; n^{-1/2} K_n^2 \log(n) \mathcal{J}(K_n^{-1/2}) \xrightarrow[n \rightarrow \infty]{} 0. \\ & (ii) \; \; \; \; \sqrt{\log n} \max \bigg \{ \frac{K_n^2 \log(K_n)}{\sqrt{n}} , \frac{K_n \delta_n^{-d/t}}{\sqrt{n}}, K_n \sqrt{\log K_n} \delta_n^{\kappa} , \sqrt{K_n} \delta_n^{\kappa - d/2t}   \bigg \} \xrightarrow[n \rightarrow \infty]{} 0. \\ & (iii) \; \; \; \;  \sqrt{n} \sqrt{K_n} \sqrt{\log n} \sup_{h \in \Theta_n}  \| \Pi_{K_n} R_{h_0}(h,W)  \|_{L^2(\mathbb{P})} \xrightarrow[n \rightarrow \infty]{} 0. 
\end{align}
\end{condition2}
Condition \ref{misc1}$(i-ii)$ arise as a consequence of empirical process tools that are used to control the deviation of  $  \E_n \big[    \widehat{m}(W,h) ' \Sigma(W)  \widehat{m}(W)            \big]$   from its population analog $ \E[ \Pi_{K_n} m(W,h)' \Sigma(W) \Pi_{K_n} m(W,h)]  $, uniformly over $h \in \Theta_n$. The dependence on $\kappa$ arises because the generalized residual function $\rho(.)$ in (\ref{cmr}) may be nonlinear and pointwise discontinuous in the input function $h$. As such, our argument instead utilizes the weaker $L^2(\mathbb{P})$  uniform continuity  in Assumption \ref{residuals}.

\begin{remark}[On the Remainder Order]
    Condition \ref{misc1}$(iii)$ imposes that the nonlinear remainder vanishes sufficiently fast on local shrinking neighborhoods around $h_0$. Under weak regularity conditions, the remainder term has order bounded above by a quadratic distance to $h_0$, i.e \begin{align} \label{quadrem} \| \Pi_{K_n} R_{h_0} (h,W) \|_{L^2(\mathbb{P})} \leq \|  R_{h_0} (h,W) \|_{L^2(\mathbb{P})} \leq C \| h - h_0 \|_{L^2(\mathbb{P})}^2 \; \; \; \; \; \; \forall \; \; h \in \Theta_n. \end{align}   For mildly ill-posed models, Condition \ref{misc1}$(iii)$ is  satisfied if $\delta_n^2 \sqrt{K_n} \sqrt{\log n} = o(n^{-1/2})$.  From the definition of $K_n$ in (\ref{kn}), this reduces to the smoothness requirement $\alpha > \zeta + d$. This is similar to the assumption used in \cite{chen2009efficient}.\footnote{More specifically, Condition 5.7 in \cite{chen2009efficient}.} Unfortunately, smoothness restrictions do not suffice for severely ill-posed models as the contraction rate $\delta_n$ has at most logarithmic order. As pointed out in the literature (e.g. \citealp{hanke1995convergence}), quadratic bounds as in (\ref{quadrem}) are generally too weak to be informative when the model is highly ill-posed. In such settings, $\| \Pi_{K_n} D_{h_0}[h-h_0] \|_{L^2(\mathbb{P})}$  can be significantly smaller than $\| h- h_0 \|_{L^2(\mathbb{P})}^2$, in which case (\ref{quadrem}) is too conservative to be informative. A more informative variant which often appears in the literature is the so called tangential cone condition\footnote{This is expression \textbf{(1.8)} in \cite{hanke1995convergence} with $\phi(t) = t$. For uses and proofs of tangential cone conditions in a variety of settings, see e.g. (\citealp{kaltenbacher2009iterative}; \citealp{de2012local}; \citealp{chen2014local}; \citealp{dunker2014iterative}; \citealp{breunig2020specification}; \citealp{kaltenbacher2021tangential}).} which, in our notation, states \begin{align}
        \label{tang-cone}  \| 
  R_{h_0}(h) \|_{L^2(\mathbb{P})} \leq \phi \big(   \|  h- h_0   \|_{L^2(\mathbb{P})}  \big) \|   m(W,h) - m(W,h_0)    \|_{L^2(\mathbb{P})} \; \; \;\; \;  \forall \;  h \in  \Theta_n
    \end{align}
for some function $\phi: \R_+ \rightarrow \R_+  $ with $\phi(0) = 0 $ and continuous at zero. For example, if $\phi(t) = t$, (\ref{tang-cone}) implies that Condition \ref{misc1}$(iii)$ holds for severely ill-posed models when $\alpha > \zeta + d$ and for mildly ill-posed models when $\alpha > d$. It is worth noting that the inequality in (\ref{tang-cone}) only needs to hold with $\Pi_K R_{h_0}$. As $\Pi_K : L^2(W) \rightarrow \mathcal{V}_K$ is a norm decreasing projection, this results in a slightly weaker condition. Other variants can be incorporated as well such as a direct bound between $ \| \Pi_K R_{h_0} \|_{L^2(\mathbb{P})}$ and $ \|\Pi_K \big[m(W,h) - m(W,h_0)] \|_{L^2(\mathbb{P})} $.
\end{remark}

The Reproducing Kernel Hilbert Space (RKHS) of the Gaussian Process $G_{\alpha}$ in (\ref{g-series-4}) can be expressed as \begin{align}
\label{RKHS}  \mathbb{H}  =  \bigg \{  g \in L^2(\mathcal{X})  : \|g \|_{\mathbb{H}}^2 =   \sum_{i=1}^{\infty} i^{ 1 +   2\alpha /d}  \left| \langle g ,  e_i \rangle      \right|^2 < \infty               \bigg \} .
\end{align}
Let $\mathbf{L} : L^2(X) \rightarrow \R$ denote a linear functional of interest. By duality (Riesz representation) there exists a function $\Phi(X) \in L^2(X)$ such that the linear functional can be expressed as \begin{align} \label{rieszrep}
\mathbf{L}(h) = \E[ \Phi(X) h(X) ] = \langle h , \Phi \rangle_{L^2(\mathbb{P})} \; \; \; \; \; \; \; \; \; \forall \; h \in L^2(X).
\end{align}
As such, to derive the limit theory with $\mathbf{L}(.)$, it suffices to work with  $\Phi(.)$ directly. 

In the preceding sections, the choice of weighting matrix $\widehat{\Sigma}(.)$ in the quasi-Bayes posterior (\ref{posterior}) did not influence the limit theory, provided that the eigenvalues of $\widehat{\Sigma}(.)$ are asymptotically bounded away from $0$ and $\infty$. Intuitively, if such a condition holds, rates of convergence can be determined by studying a quasi-Bayes posterior based on the simpler  objective function $  h \rightarrow \E_n \big( \| \widehat{m}(W,h)  \|_{\ell^2}^2    \big) $. However, to determine finer aspects of the posterior (such as a precise limiting distribution for functionals), it will be necessary to include the limiting behavior of $\widehat{\Sigma}(.)$ in the analysis. The following condition imposes that $\widehat\Sigma(W)$ converges   to  a limiting positive definite  matrix $\Sigma(W)$. Properties of the quasi-posterior can then be analyzed through the deterministic matrix $\Sigma(.)$.

\begin{condition2}\label{misc2-pre}
$(i)$There exists a positive definite matrix $\Sigma(.) $ such that $ \sup_{w \in \mathcal{W}}  \|  \widehat{\Sigma}(W) - \Sigma(W)  \|_{op} = O_{\mathbb{P}}(\gamma_n)  $ for some sequence $(\gamma_n)_{n=1}^{\infty}$ satisfying  $\gamma_n \log(n) K_n \rightarrow 0$. $(ii)$ The eigenvalues of $\Sigma(W)$ are bounded away from $0$ and $\infty$: $\mathbb{P}(c \leq \lambda_{\min} (\Sigma(W)) \leq \lambda_{\max}(\Sigma(W)) \leq C ) = 1$ for  some constants $ 0 < c \leq C < \infty$.
\end{condition2}
Given a positive definite matrix $\Sigma(.)$ as in Condition \ref{misc2-pre}, we can endow $L^2(W,\R^{d_{\rho}})$ with a norm (and inner product) by \begin{align}
    \label{innerprod-induce} m  \rightarrow  \sqrt{ \E \big[ m(W)' \Sigma(W) m(W) \big] }  \; \; \; \; \forall \; m  \in L^2(W,\R^{d_{\rho}}).
\end{align}
This is the weighted norm induced by $\Sigma(.)$. As the quasi-Bayes objective function is based on a feasible version of (\ref{innerprod-induce}), it will be convenient in our analysis to view  (\ref{innerprod-induce}) as the natural norm for $m(W,h)$.\footnote{Observe that, under Condition \ref{misc2-pre}$(ii)$, this norm is equivalent to the usual unweighted norm $\|m \|_{L^2(\mathbb{P})} = \E[ \|  m(W)  \|_{\ell^2}^2]$. However, since certain operators such as adjoints are defined based on the specific form of the inner product, it is convenient to view $L^2(W,\R^{d_{\rho}})$ as a Hilbert space with natural norm as in (\ref{innerprod-induce}) directly.} Unless otherwise stated, for the remainder of this section, we view $L^2(W,\R^{d_{\rho}})$ as a Hilbert space with norm as in $(\ref{innerprod-induce})$. Denote the adjoint of $D_{h_0} : L^2(X) \rightarrow L^2(W,\R^{d_{\rho}})$ by $D_{h_0}^*$. We impose the following condition on the Riesz representer  $\Phi(.)$ in (\ref{rieszrep}) that determines the linear functional $\mathbf{L}(.)$.

\begin{condition2}\label{misc2}
$(i)$ There exists $  \tilde{\Phi} \in \mathbb{H}   $ such that $  \Phi = D_{h_0}^* D_{h_0} \tilde{\Phi}$. $(ii)$ $ D_{h_0}[\tilde{\Phi}]$ and $\Sigma(.)$ satisfy $ \sqrt{K_n} \sqrt{\log n}  \|(\Pi_{K_n} - I)D_{h_0}[\tilde{\Phi}] \|_{L^2(\mathbb{P})} \rightarrow 0$ and  $ \sqrt{K_n} \sqrt{\log n}  \|(\Pi_{K_n} - I) \Sigma(W) D_{h_0}[\tilde{\Phi}] \|_{L^2(\mathbb{P})} \rightarrow 0$.
\end{condition2}
Condition \ref{misc2}$(i)$ is a source type condition on $\Phi$. This is similar to Condition 3.3 in \citet{monard2021statistical}. We note that the condition for $\Phi$ to be in the range of the adjoint operator is a well known necessary condition for $\sqrt{n}$ estimability of linear functionals.\footnote{ For classical semiparametric models, this follows from \citet{van1991differentiable}. For specific applications to NPIV and NPQIV setups, see \citet{severini2012efficiency} and \citet*{chen2019penalized}.} Condition \ref{misc2}$(ii)$ imposes weak smoothness requirements on $D_{h_0}[\tilde{\Phi}]$ and $\Sigma(.) D_{h_0}[\tilde{\Phi}] $.

Let $ \rho $ denote any metric which metrizes weak convergence of probability measures on $\R$. If $(\mu_{n})_{n=1}^{\infty}$ is a random (data dependent) sequence of measures on $\R$, we say that $\mu_{n} \overset{\mathbb{P}}{\rightsquigarrow} \nu  $ for some non-random limit measure $\nu$ if $\rho(\mu_n , \nu) \xrightarrow{\mathbb{P}} 0$. The following theorem shows that the quasi-posterior distribution of a sufficiently smooth (in the sense of Condition \ref{misc2})  linear functional $\mathbf{L}(.)$  can be well approximated by a suitable Gaussian measure.

\begin{theorem}[Bernstein-von Mises] \label{bvm}
Suppose $h_0 \in \mathcal{H}^p$ for some $p \geq \alpha + d/2$, Assumptions \ref{data}-\ref{svd} and Conditions \ref{fsbasis}-\ref{misc2} hold. Then \begin{align*}
   & (i) \; \; \; \; \sqrt{n} \langle  h - \E \big[h| \mathcal{Z}_n \big] , \Phi       \rangle_{L^2(\mathbb{P})} \big| \mathcal{Z}_n    \overset{\mathbb{P}}{\rightsquigarrow} N \big( 0 ,  \E \big[ (D_{h_0} \tilde{\Phi})' \Sigma (D_{h_0} \tilde{\Phi})      \big]             \big) \:, \\ & (ii)  \; \; \;  \sqrt{n} \langle  h_0 - \E \big[  h|\mathcal{Z}_n  \big] , \Phi       \rangle_{L^2(\mathbb{P})}  \rightsquigarrow    N(0, \E \big[ (D_{h_0} \tilde{\Phi})' \Sigma  \rho \rho'  \Sigma  (D_{h_0} \tilde{\Phi}) \big]   ) .
\end{align*}
\end{theorem}
Observe that the two variances in Theorem \ref{bvm} agree if and only if the limiting weighting matrix of the quasi-Bayes objective is the optimal weighting matrix. That is, when \begin{align} \label{optimal-sigma} \Sigma(W) = \{ \E[ \rho(Y,h_0(X)) \rho(Y,h_0(X))'|W  ] \}^{-1} . \end{align}  
In this case, the common limiting variance is given by \begin{align}
    \label{sigma-optimal} \sigma_{\Phi}^2 =  \E \big[ (D_{h_0} \tilde{\Phi} )' \{ \E[ \rho(Y,h_0(X)) \rho(Y,h_0(X))'|W  ] \}^{-1} (D_{h_0} \tilde{\Phi} )     \big]      .
\end{align}
One implication of an optimally weighted quasi-Bayes objective function, or equivalently the equality of variances in Theorem \ref{bvm}, is that quasi-Bayesian credible bands centered around the posterior mean have asymptotically exact frequentist coverage. To be specific, given a linear functional $\mathbf{L}(.)$ and significance level $\gamma \in (0,1)$, define \begin{align}
    \label{sig-alpha} c_{1- \gamma} = (1- \gamma) \; \:  \text{quantile of} \:   \;  \left| \mathbf{L}(h) - \mathbf{L}\big( \E [h| \mathcal{Z}_n ]   \big)  \right| \;  ,  \; h \sim \mu(.|\alpha,K,\mathcal{Z}_n).
\end{align}
Denote the quasi-Bayesian credible band at significance level $\alpha$ by \begin{align} \label{cn-alpha}
    C_n(\gamma) = \{ t \in \R : \big|  t - \mathbf{L}\big( \E [h| \mathcal{Z}_n ] \big)   \big|  \leq c_{1- \gamma}   \}.
\end{align}
An immediate consequence of Theorem \ref{bvm} is that optimally-weighted quasi-Bayes credible sets have asymptotically exact frequentist coverage. 
\begin{corollary}
\label{bvm-col} 
Suppose the hypothesis of Theorem \ref{bvm} holds and the quasi-Bayes objective is asymptotically optimally weighted, i.e $\Sigma(W) = \{ \E[ \rho(Y,h_0(X)) \rho(Y,h_0(X))'|W  ] \}^{-1} $. Then \begin{align}
    \label{bvm-final} \lim_{n \rightarrow \infty} \mathbb{P} \big(   \mathbf{L}(h_0)  \in C_n(\gamma)    \big) = 1- \gamma.
\end{align}

\end{corollary}

\section{Proofs} \label{sec6}
In this section, we provide proofs for all the main results. We first state some notation that appears throughout the proofs. We denote by $\widehat{G}_{b,K}^{o}$ the matrix \begin{align}
\label{gbko}  \widehat{G}_{b,K}^{o} =  G_{b,K}^{-1/2} \widehat{G}_{b,K} G_{b,K}^{-1/2} \; \end{align}
where $G_{b,K} = \E \big( [b^K(W)] [ b^K(W) ]'    \big)$.  We denote by $\zeta_{b,K}$ the quantity $ \zeta_{b,K} =  \sup_{w \in \mathcal{W} }  \|  G_{b,K}^{-1/2} b^K(w) \|_{\ell^2}  $. By Condition \ref{fsbasis}, we have $\zeta_{b,K} \lessapprox \sqrt{K}$. For ease of notation, in several of the limit results, we often suppress the dependence of $K=K_n$ on $n$.

\begin{lemma}
\label{aux1}
Suppose Assumption \ref{fsbasis}(i) holds. Then, for every sieve dimension $K$ and $t > 0 $, we have that
 $$  \mathbb P \left( \|\widehat{G}_{b,K}^{o}- I_{K}  \|_{op} > t \right)  \leq  2 K \exp \bigg( - \frac{t^2/2}{\zeta^2_{b,K}/n + 2 \zeta^2_{b,K} t/(3n) }      \bigg).   $$
 
\end{lemma}

\begin{proof}[Proof of Lemma \ref{aux1}]
Observe that
$$ \widehat{G}_{b,K}^o - I_K  = n^{-1} \sum_{i=1}^n G_{b,K}^{- 1/2} \big \{ b^K (W_i) b^K(W_i)' - \E [b^K(W) b^K(W)' ]       \big \} G_{b,K}^{- 1/2} = \sum_{i=1}^n \Xi_{i}  \; , $$
where  $(\Xi_i)_{i=1}^n$ are i.i.d matrices of dimension $K \times K$. Furthermore, we have that  \begin{align*}
& \| \Xi_i  \|_{op} \leq 2 n^{-1} \zeta_{b,K}^2 \; , \\ & \| \E[\Xi_i \Xi_i']   \|_{op} \leq n^{-2} \| \E [ G_{b,K}^{- 1/2} b^K(W)  b^K(W)' G_{b,K}^{- 1/2}    ]    \|_{op} = n^{-2} \| I_K \|_{op} = n^{-2} \; , \\ & \| \E[\Xi_i' \Xi_i]   \|_{op} \leq n^{-2} \big | \E[ b^K(W)' G_{b,K}^{-1} b^K(W)  ]    \big| \leq n^{-2} \zeta_{b,K}^2 .
\end{align*}
The claim follows from using these bounds in an application of \citep[Theorem 1.6]{tropp}.
\end{proof}

\begin{lemma}
\label{aux2}
 Suppose Assumption \ref{fsbasis}(i) holds. Let $\bar{K}_{\max}  = \bar{K}_{\max,n} $ denote a sequence  that satisfies $\bar{K}_{\max} \uparrow \infty$ and $\bar{K}_{\max}  \log( \bar{K}_{\max} ) /n \downarrow 0$. Then, there exists a universal constant $D < \infty$ such that $$  \mathbb{P} \bigg( \sup_{K \in \mathbb{N} :K \leq \bar{K}_{\max}} \| \widehat{G}_{b,K}^o- I_K \|_{op} \leq    D  \frac{\sqrt{  \bar{K}_{\max} } \sqrt{\log  \bar{K}_{\max} }}{\sqrt{n}}    \bigg) \rightarrow 1 . $$

\end{lemma}
\begin{proof}[Proof of Lemma \ref{aux2}]
Lemma \ref{aux1} and a union bound yields \begin{align*}   \mathbb{P} \bigg(\sup_{K \in \mathbb{N} :K \leq  \bar{K}_{\max} } \|\widehat  G_{b,K}^o - I_K  \|_{op} > t   \bigg) & \leq \sum_{K \in \mathbb{N} :K \leq  \bar{K}_{\max}}   \mathbb{P} \left(  \|\widehat  G_{b,K}^o - I\|_{op} > t  \right) \\ & \leq   2 \sum_{K \in \mathbb{N} :K \leq  \bar{K}_{\max}}  K \exp\left\{ - \frac{t^2 /2}{\zeta_{b,K}^2(1+ 2 t  /3)n^{-1}} \right\}  .
\end{align*}

Let $L > 0 $ be such that $\zeta^2_{b,K} \leq L K $ for all $K$ and fix any $D  > \sqrt{8 L} $. Define $t = t_n = D  \sqrt{  \bar{K}_{\max}   \log \bar{K}_{\max}  }/ \sqrt{n}$. Since $t_n \downarrow 0$, there exists $N \in \mathbb{N}$ such that  $2 t_n /3 \leq 1   $ for all $  n > N$. For $n > N$, it follows that \begin{align*}   \sum_{K \in \mathbb{N} :K \leq  \bar{K}_{\max}}  K \exp\left\{ - \frac{t_n^2 /2}{\zeta_{b,K}^2(1+ 2 t_n  /3)/n} \right\} & \leq \bar{K}_{\max}^2 \exp \bigg \{  -  \frac{D^2  \log (\bar{K}_{\max}) }{4L}    \bigg \} \\ & = \exp \bigg \{  \bigg(2- \frac{D^2}{4L} \bigg)  \log(\bar{K}_{\max})  \bigg \}  \\ & \rightarrow 0.
\end{align*}
\end{proof}
\begin{lemma}
\label{aux3}
Suppose $G$ is a Gaussian random element on a separable Banach space $(B,\| . \|_{B})$.  Let $B^*$ denote the dual space of $B$. Then  
  \begin{align*}
\mathbb{P}  \big( \| G    \|_{B} >  \E \big(  \| G \|_{B}    \big) + u              \big) \leq 2 \exp\bigg( - \frac{u^2}{2 \sigma^2}       \bigg) \; \; \;  \; \; \; \; \; \; \;  \forall \; u > 0 \;,
\end{align*}
where $\sigma^2 = \sup_{T \in B^* : \|  T \|_{op} = 1} \E \big(  \left| T(G) \right|^2  \big)  .   $
\end{lemma}
\begin{proof}[Proof of Lemma \ref{aux3}]
Since$(B, \| . \|_{B})$ is a separable Banach space, the Hahn-Banach theorem implies that there exists a countable  family  $\{ T_i   \}_{i=1}^{\infty} \subseteq  B^* $ such that $ \| T_i  \|_{op} = 1$ for every $i \in \mathbb{N}$ and $ \|   G   \|_{B} = \sup_{i \in \mathbb{N}} \left|  T_i(G)      \right|  $.
The desired bound follows from an application of  \citep[Theorem 2.1.20]{gine2021mathematical} to the separable centered Gaussian process $\{ T_i(G) : i \in \mathbb{N}  \}$.

\end{proof}

\begin{lemma}
\label{aux4}
Suppose Assumptions  \ref{data}, \ref{identif}, \ref{svd} hold. Then, given any $\epsilon > 0 $, there exists  $\delta > 0 $ such that \begin{align*} &  h \in \mathbf{H}^t(M) \; , \; \| m(W,h)  \|_{L^2(\mathbb{P})} < \delta  \implies  \| h - h_0 \|_{L^2(\mathbb{P})} < \epsilon \; \; ,  \\ & h \in \mathbf{H}^t(M) \; , \; \| D_{h_0}[h- h_0]  \|_{L^2(\mathbb{P})} < \delta  \implies  \| h - h_0 \|_{L^2(\mathbb{P})} < \epsilon  .
\end{align*}
\end{lemma}

\begin{proof}[Proof of Lemma \ref{aux4}] 
We verify the claim for $m(.)$, the argument for $D_{h_0}(.)$ is analogous. By Assumption \ref{data}, the $\| . \|_{L^2}  $ metric is equivalent to the $\| . \|_{L^2(\mathbb{P})}$ metric. It follows that  $\mathbf{H}^t(M)$ is compact under the $\| . \|_{L^2(\mathbb{P})}$ metric. By Assumption \ref{identif}(ii), the mapping  $ m(W, \: . \:) :  ( \mathbf{H}^t(M) , \| .  \|_{L^2(\mathbb{P})}  ) \rightarrow ( L^2(W,\R^{d_{\rho}}), \| . \|_{L^2(\mathbb{P})}   )   $ is continuous. For any $\epsilon > 0 $, the set $\{ h \in \mathbf{H}^t(M) : \| h- h_0  \|_{L^2(\mathbb{P})} \geq \epsilon   \}$ is a closed (and hence, compact) subset of $( \mathbf{H}^t(M) , \| .  \|_{L^2(\mathbb{P})}  ) $. By Assumption \ref{identif}(i),  $h_0$ is the unique zero of $\| m(W,h)  \|_{L^2(\mathbb{P})}$ over $\mathbf{H}^t$.  As a continuous function over a compact set achieves its infimum, it follows that there exists a $\delta > 0 $ for which \begin{align*}
\inf_{h \in \mathbf{H}^t(M) : \| h- h_0  \|_{L^2(\mathbb{P})} \geq \epsilon } \|  m(W,h)      \|_{L^2(\mathbb{P})} =  \min_{h \in \mathbf{H}^t(M) : \| h- h_0  \|_{L^2(\mathbb{P})} \geq \epsilon } \|  m(W,h)      \|_{L^2(\mathbb{P})} \geq \delta.
\end{align*}

\end{proof}

\begin{lemma}
\label{emp-b}
Suppose Assumptions \ref{fsbasis}(i), \ref{residuals} and \ref{residuals2}(i) hold. For each fixed $l \in \{ 1 ,  \dots , d_{\rho} \}$ and function $h: \mathcal{X} \rightarrow \R $, define
 \begin{align*} & R_{h,l}^K(Z) =     \big[ G_{b,K}^{-1/2}  b^K(W) \big ]   \rho_{l}(Y,h_{}(X))  .  \end{align*}
Then, given any $M > 0 $, there exists a universal constant $D = D(M) < \infty$ such that
 \begin{equation}  \sup_{l \in \{1 , \dots ,  d_{\rho} \} }  \E\bigg(     \sup_{h  \in \mathbf{H}^t(M)   }   \|   \E_n [ R_{h,l}^K(Z) ] - \E [ R_{h,l}^K(Z) ]        \|_{\ell^2}  \bigg) \leq D  \frac{\sqrt{K}}{\sqrt{n}} \label{bern1}   \end{equation}
holds for every $K $.
\end{lemma}

\begin{proof}[Proof of Lemma \ref{emp-b}]

It suffices to verify that (\ref{bern1}) holds for each  $l \in \{1 , \dots , d_{\rho} \}$.  Fix any such $l$.  For ease of notation, we suppress the dependence on $l$ and denote the associated vector by $   R_{h,l}^K(Z) = R_{h}^K(Z)$.  Denote the $j  \in \{ 1 , \dots , K \}$ element of  $R_{h}^K(Z)$ by $[R_{h}^K(Z)]_{j} =   \big[ G_{b,K}^{-1/2}  b^K(W) \big ]_{j}   \rho_{l}(Y,h_{}(X))   $. Observe that
 \begin{align*} &
\E_{} \bigg[ \sup_{h \in \mathbf{H}^t(M)  }  \|   \E_n [R_{h}^K(Z)] - \E_{}[R_{h}^K(Z)]       \|_{\ell^2}^2  \bigg] \\ & =  \frac{1}{n} \E \bigg[   \sup_{h \in \mathbf{H}^t(M)}  \sum_{j=1}^K \left|   \frac{1}{\sqrt{n}} \sum_{i=1}^n\big \{  [R_{h}^K(Z_i)]_{j} -\E_{} ( [R_{h}^K(Z)]_{j} ) \:  \big \}  \right|^2          \:  \bigg]   \\ & \leq  \frac{1}{n} \sum_{j=1}^K \E \bigg[ \sup_{h \in \mathbf{H}^t(M)}   \left|   \frac{1}{\sqrt{n}} \sum_{i=1}^n\big \{  [R_{h}^K(Z_i)]_{j} -\E_{} ( [R_{h}^K(Z)]_{j} ) \:  \big \}  \right|^2      \;  \bigg] \\ & \leq \frac{K}{n} \sup_{j \in \{1 , \dots , K \}}  \E \bigg[ \sup_{h \in \mathbf{H}^t(M)}   \left|   \frac{1}{\sqrt{n}} \sum_{i=1}^n\big \{  [R_{h}^K(Z_i)]_{j} -\E_{} ( [R_{h}^K(Z)]_{j} ) \:  \big \}  \right|^2      \;  \bigg] .
\end{align*}
It suffices to verify that the expectations are uniformly bounded.  Fix any such $j$.  We view  the expectation as a higher moment of an empirical process over the class of functions $$\mathcal{F}  = \{ [R_h^K(Z)]_{j}   : h \in \mathbf{H}^t(M)    \}    . $$
Let $F(Z) =  \sup_{f \in \mathcal{F}} \left| f(Z) \right| $ denote the envelope of $\mathcal{F}$.  Let $C_2(M)$ be as in Assumption \ref{residuals2}(i).  By Assumption \ref{residuals2}(i) and the observation that $\big[ G_{b,K}^{-1/2}  b^K(W) \big ]_{j}$ has unit $L^2(\mathbb{P})$ norm, the envelope admits the bound
 
  \begin{align*}
\|  F  \|_{L^2(\mathbb{P})}^2 & = \bigg \|    \sup_{h \in \mathbf{H}^t(M)}   \big[ G_{b,K}^{-1/2}  b^K(W) \big ]_{j}   \rho_{l}(Y,h_{}(X))       \bigg \|_{L^2(\mathbb{P})}^2 \\ & \leq  \E \bigg[  \left|   \big[ G_{b,K}^{-1/2}  b^K(W) \big ]_{j}  \right| ^2  \E \bigg[    \sup_{h \in \mathbf{H}^t(M)}                \left| \rho_{l}(Y,h_{}(X)) \right|^2  \bigg| W \bigg]   \bigg] \\ & \leq C_2^2 \E \bigg[  \left|   \big[ G_{b,K}^{-1/2}  b^K(W) \big ]_{j}  \right| ^2   \bigg] \\ & = C_2^2.
 \end{align*}
From an application of \citep[Theorem 2.14.5]{weakc},  there exists a universal constant $ L > 0 $ such that \begin{align*}
 & \E \bigg[ \sup_{h \in \mathbf{H}^t(M)}   \left|   \frac{1}{\sqrt{n}} \sum_{i=1}^n\big \{  [R_{h}^K(Z_i)]_{j} -\E_{} ( [R_{h}^K(Z)]_{j} ) \:  \big \}  \right|^2      \;  \bigg] \\ & \leq L \bigg(   \E \bigg[ \sup_{h \in \mathbf{H}^t(M)}   \left|   \frac{1}{\sqrt{n}} \sum_{i=1}^n\big \{  [R_{h}^K(Z_i)]_{j} -\E_{} ( [R_{h}^K(Z)]_{j} ) \:  \big \}  \right|      \;  \bigg]    + C_2   \bigg)^2 .
\end{align*}

 By an application of \citep[Theorem 3.5.13]{gine2021mathematical},  there exists a universal constant $ L > 0 $ such that \begin{align*}
&  \E \bigg[ \sup_{h \in \mathbf{H}^t(M)}   \left|   \frac{1}{\sqrt{n}} \sum_{i=1}^n\big \{  [R_{h}^K(Z_i)]_{j} -\E_{} ( [R_{h}^K(Z)]_{j} ) \:  \big \}  \right|      \;  \bigg]    \leq \frac{L}{\sqrt{n}}  \int_{0}^{8  \|  F  \|_{L^2(\mathbb{P})}  }  \sqrt{\log N_{[]}(\mathcal{F}, \| . \|_{L^2(\mathbb{P})},  \epsilon  ) } d \epsilon.
\end{align*}
Let $\{  h_i \}_{i=1}^T  $  denote a $\delta > 0$ covering of  $\big( \mathbf{H}^t(M) , \| . \|_{\infty}   \big)$. Define the functions $$ e_{i}(Z) = \sup_{h \in \mathbf{H}^t(M) : \|  h -h_i \|_{\infty} < \delta  } \ \big| \:  [R_{h}^K(Z)]_{j}  -  [R_{h_i}^K(Z)]_{j}   \:   \big| \; \; \; \; \; i=1,\dots,T . $$
Observe that $$ \bigg \{   [R_{h_i}^K(Z)]_{j}   - e_i \; ,  \;  [R_{h_i}^K(Z)]_{j}  + e_i   \bigg \}_{i=1}^T  $$ is a bracket covering for $\mathcal{F}$.  Let $C_1(M)$ and $\kappa \in (0,1]$ be as in Assumption \ref{residuals}.  By Assumption \ref{residuals}(i) we have that  \begin{align*}
 \| e_i \|_{L^2(\mathbb{P})}^2  &   \leq \E \bigg[  \left|   \big[ G_{b,K}^{-1/2}  b^K(W) \big ]_{j}  \right| ^2  \E \bigg[    \sup_{h \in \mathbf{H}^t(M) : \|  h -h_i \|_{\infty} < \delta  }           \left| \rho_{l}(Y,h_{}(X)) - \rho_{l}(Y,h_{i}(X))  \right|^2  \bigg| W \bigg]   \bigg] \\ & \leq C_1^2 \delta^{2 \kappa} \E \bigg[  \left|   \big[ G_{b,K}^{-1/2}  b^K(W) \big ]_{j}  \right| ^2 \bigg] \\ & = C_1^2 \delta^{2 \kappa}.
\end{align*}
It follows that \begin{align*}
 & \int_{0}^{8 \|  F  \|_{L^2(\mathbb{P})}  }  \sqrt{\log N_{[]}(\mathcal{F}, \| . \|_{L^2(\mathbb{P})},  \epsilon  ) } d \epsilon \leq  \int_{0}^{8 \|  F  \|_{L^2(\mathbb{P})}  } \sqrt{\log N \bigg(  \mathbf{H}^t(M)   ,\|.\|_{\infty} ,   \bigg( \frac{\epsilon}{ 2  C_1}   \bigg)^{1/ \kappa} \:   \bigg)   } d \epsilon .
\end{align*}
 By \citep[Proposition C.7]{ghosal2017fundamentals}, we have that $  \log N_{}( \mathbf{H}^t(M),  \| . \|_{ \infty}, \epsilon) \lessapprox  \epsilon^{-d/t}$ as $\epsilon \downarrow 0$.  It follows that there exists a universal constant $ L> 0$ such that \begin{align*}
  \int_{0}^{8 \|  F  \|_{L^2(\mathbb{P})}  } \sqrt{\log N \bigg(  \mathbf{H}^t(M)   ,\|.\|_{\infty} ,   \bigg( \frac{\epsilon}{ 2  C_1}   \bigg)^{1/ \kappa} \:   \bigg)   } d \epsilon  &  \leq L  \int_{0}^{8 \|  F  \|_{L^2(\mathbb{P})}  } \epsilon^{-d/2 \kappa t} d \epsilon  \\ & \leq L \int_{0}^{8 C_2}  \epsilon^{-d/2 \kappa t} d \epsilon .
 \end{align*}
 By Assumption \ref{residuals}(ii), $t > (2 \kappa)^{-1} d$ and so the integral  above  is convergent.  From the monotonicity of the $L^p(\mathbb{P})$ norm and combining all the preceding bounds, it follows that there exists a universal constant $ D > 0 $ such that

\begin{align*}
 \E_{} \bigg[ \sup_{h \in \mathbf{H}^t(M)  }  \|   \E_n [R_{h}^K(Z)] - \E_{}[R_{h}^K(Z)]       \|_{\ell^2}  \bigg]   & \leq  \bigg \|      \sup_{h \in \mathbf{H}^t(M)  }  \|   \E_n [R_{h}^K(Z)] - \E_{}[R_{h}^K(Z)]       \|_{\ell^2}      \bigg  \|_{L^2(\mathbb{P})} \\ & \leq D \frac{\sqrt{K}}{\sqrt{n}}.
\end{align*}
\end{proof}

\begin{lemma}
\label{emp-c}
Suppose Assumptions \ref{fsbasis}(i), \ref{residuals}  and \ref{residuals2}(i)(ii) hold.  For each fixed $l \in \{ 1 ,  \dots , d_{\rho} \}$ and function $h(X)$, define
 $$ R_{h,l}^K(Z) =     \big[ G_{b,K}^{-1/2}  b^K(W) \big] \rho_{l}(Y,h_{}(X)). $$
Let $\epsilon> 0 $ be as in Assumption \ref{residuals2}(ii) and define $\gamma = 1- 1/(2+2 \epsilon) > 1/2 $. Suppose $K_{\min}  < \bar{K}_{\max} \   $ are sequences of sieve dimensions that  both diverge to $\infty$ and satisfy $K_{\min} \log(\bar{K}_{\max})= o( n^{\gamma - 1/2}) $ and $ \log \bar{K}_{\max} =o(K_{\min})$.  Define  $\mathcal{K}_n= [ K_{\min} , \bar{K}_{\max} ] \cap \mathbb{N} $. Then, given any $M > 0 $, there exists a universal constant $D = D(M) < \infty$ such that  \begin{equation}
  \mathbb{P}     \bigg(   \sup_{l \in \{ 1 , \dots ,d_{\rho} \} }     \sup_{K \in \mathcal{K}_n} \sup_{h \in \mathbf{H}^t(M)}   K^{-1/2} 
 \| \E_n \big[ R_{h,l}^K(Z) \big]   - \E_{} \big[   R_{h,l}^K(Z)  \big]       \|_{2}  \leq  \frac{D}{\sqrt{n}}   \bigg ) \rightarrow 1 .    \label{bern2}
    \end{equation}
\end{lemma}

\begin{proof}[Proof of Lemma \ref{emp-c}]
It suffices to verify that (\ref{bern2}) holds at each fixed $l \in \{  1 , \dots, d_{\rho} \} $. Fix any such $l$. For a given sequence of deterministic constants $L_n \uparrow \infty$, define   \begin{align*}  & \xi_{1,i}^K (h_{}) =  R_{h,l}^K(Z_i)  \mathbbm{1} \bigg \{ \sup_{h \in \mathbf{H}^t(M)}  | \rho(Y_i, h(X_i)) | \leq L_n   \bigg \} \;, \\ &  \xi_{2,i}^K (h_{}) = R_{h,l}^K(Z_i)  \mathbbm{1} \bigg \{ \sup_{h \in \mathbf{H}^t(M)}  | \rho(Y_i, h(X_i)) | > L_n    \bigg \}  \; .      \end{align*}
Write the deviation as \begin{align} \label{sums} (\E_n - \E_{})[R_{h,l}^K(Z) ]   = \sum_{i=1}^n  \Xi_{1,i}^K (h_{})  + \sum_{i=1}^n \Xi_{2,i}^K (h_{})    . \end{align}
where $ \Xi_{1,i}^K (h_{}) = n^{-1} [\xi_{1,i}^K (h_{}) - \E_{} \xi_{1,i}^K (h_{})] $ and $ \Xi_{2,i}^K(h_{}) = n^{-1} [ \xi_{2,i}^K(h_{}) - \E_{} \xi_{2,i}^K(h_{}) ]$.
First, we derive a bound for $ \sum_{i=1}^n \Xi_{2,i}^{K}(h) $.  Let $\epsilon > 0 $ be as in Assumption \ref{residuals2}(ii). From the bound $ \zeta_{b,K}^{-1} \|  G_{b,K}^{-1/2} b^{K}(W_i)     \|_{\ell^2}  \leq   1 $, it follows that 
\begin{align*}   &  \mathbb{P}_{} \bigg( \sup_{h \in \mathbf{H}^t(M)}   \bigg \| \sum_{i=1}^n \Xi_{2,i}^K (h_{}) \bigg \|_{\ell^2}   >  \frac{\zeta_{b,K}}{\sqrt{n}}  \bigg) \\ & \leq       \frac{\sqrt{n}}{\zeta_{b,K}} \E_{} \bigg(     \sup_{h \in \mathbf{H}^t(M)}    \sum_{i=1}^n \| \Xi_{2,i}^K (h_{}) \|_{\ell^2} \bigg) \\ & \leq 2 \sqrt{n} \E_{} \bigg( \sup_{h \in  \mathbf{H}^t(M)}  | \rho(Y_i, h(X_i))  | \mathbbm{1} \bigg \{ \sup_{h \in \mathbf{H}^t(M)}  | \rho(Y_i, h(X_i)) | > L_n    \bigg \} \bigg) \\ & \leq \frac{2 \sqrt{n}}{L_n^{1+ \epsilon}} \E_{} \bigg( \sup_{h \in \mathbf{H}^t(M)}  | \rho(Y_i, h(X_i))  |^{2+ \epsilon} \bigg) . \end{align*}
Since $\E_{} \big( \sup_{h \in \mathbf{H}^t(M)}  | \rho(Y_i, h(X_i))  |^{2+ \epsilon} \big) < \infty $, a union bound over $K \in \mathcal{K}_n$ yields $$  \mathbb{P}_{} \bigg( \bigcup_{K \in \mathcal{K}_n} \bigg \{ \sup_{h \in \mathbf{H}^t(M)}   \bigg  \| \sum_{i=1}^n \Xi_{2,i}^K (h) \bigg \|_{\ell^2}   >  \frac{\zeta_{b,K}}{\sqrt{n}}  \bigg \} \bigg) \lessapprox \frac{\sqrt{n} \log(\bar{K}_{\max})}{L_n^{1+ \delta}} .  $$
The term on the right is $o(1)$ when $L_n^{1+ \epsilon} \asymp   \sqrt{n} (\log \bar{K}_{\max})^{1+ \epsilon} $. The desired bound then follows from observing that $\zeta_{b,K} \lessapprox \sqrt{K}$. It remains to bound the first sum in (\ref{sums}) when $L_n^{1+ \epsilon} \asymp \sqrt{n} (\log \bar{K}_{\max})^{1+ \epsilon} $. Observe that $$ \sup_{h \in \mathbf{H}^t(M)} \bigg \|  \sum_{i=1}^n \Xi_{1,i}^K(h)        \bigg \|_{\ell^2} =  \sup_{h \in \mathbf{H}^t(M)} \; \sup_{\alpha  \in \mathbb{S}^{K-1}  }  \sum_{i=1}^n \alpha' \Xi_{1,i}^K (h)  \: .$$
Let $C_2 = C_2(M) < \infty$ be as in   Assumption \ref{residuals2}(i). Define $\gamma = 1- 1/(2+2 \epsilon) > 1/2 $. For any fixed $\alpha \in  \mathbb{S}^{K-1}  $ and  $h \in \mathbf{H}^t(M) $, we have that   \begin{align*} &    \E_{}  \big[ \big(  \alpha' \Xi_{1,i}^K(h)  \big)^2 \big] \leq   n^{-2} \E_{} \bigg( \alpha' G_{b,K}^{- 1/2} b^K(W_i) b^K(W_i)' G_{b,K}^{- 1/2} \alpha  \sup_{h \in \mathbf{H}^t(M)} | \rho_l(Y,h)|^2   \bigg) \leq  C_2^2 n^{-2}  \: , \\ &  \left|  \alpha' \Xi_{1,i}^K (h)   \right| \leq 2 n^{-1}  L_n \zeta_{b,K} \lessapprox  \frac{2  \zeta_{b,K} \log \bar{K}_{\max} }{n^{\gamma} }.    \: 
\end{align*}
By Lemma \ref{emp-b}, there exists a universal constant $D  =D(M) <  \infty  $ such that
\begin{align*}
\E_{} \bigg(    \sup_{h \in \mathbf{H}^t(M)}    \bigg     \|  \sum_{i=1}^n \Xi_{1,i}^K (h) \bigg \|_{\ell^2} \bigg) \leq  D  \frac{\sqrt{K}}{\sqrt{n}} .
\end{align*}
holds for every $K$.
Talagrand's inequality \citep[Theorem 3.3.9]{gine2021mathematical} yields   \begin{align*} & \mathbb{P}_{} \bigg( \sup_{h \in \mathbf{H}^t(M)}  \bigg \|  \sum_{i=1}^n \Xi_{1,i}^K        \bigg \|_{\ell^2} \geq   \frac{ D  \sqrt{K}  }{\sqrt{n}} + \frac{\sqrt{K}}{\sqrt{n}}   \bigg) \\ &  \leq \exp \bigg(  - \frac{1}{ 2 C_2^2 K^{-1}  + (8 D  + 4/3) \big(\zeta_{b,K} \log (\bar{K}_{\max}) K^{-1/2} n^{1/2 - \gamma}   \big   )      }    \bigg) .
\end{align*}
Let $E > 0 $ be such that $\zeta_{b,K} \leq E \sqrt{K}$. From a union bound, we obtain \begin{align*}
 & \mathbb{P}_{} \bigg( \bigcup_{K \in \mathcal{K}_n} \bigg \{  \sup_{h \in \mathbf{H}^t(M)}   \bigg \|  \sum_{i=1}^n \Xi_{1,i}^K        \bigg \|_{\ell^2} \geq  \frac{ D  \sqrt{K}  }{\sqrt{n}} + \frac{\sqrt{K}}{\sqrt{n}} \bigg \}  \bigg) \\ & \lessapprox \bar{K}_{\max} \exp \bigg(   - \frac{1}{ 2 C_2^2 K_{\min}^{-1} + E(8 D + 4/3) \log(\bar{K}_{\max}) n^{1/2 - \gamma}  }         \bigg).
\end{align*}
This term is $o(1)$ when $K_{\min} \log(\bar{K}_{\max}) / n^{\gamma - 1/2} \downarrow 0$ and $ \log(\bar{K}_{\max}) K_{\min}^{-1} \downarrow 0  $.

\end{proof}

\begin{proof}[Proof of Theorem \ref{t1}]
Let $C, C'$ denote generic universal constants that may change from line to line.  First, we introduce some preliminary notation that will be used throughout the proof. Define  \begin{align}
\epsilon_{n} =  \frac{\sqrt{ K_n}}{\sqrt{n}}  .    \label{epsg}
\end{align}
For ease of notation, we surpress the dependence of $K = K_n$ on $n$ in the remainder of the proof. 

Given any fixed function $h: \mathcal{X} \rightarrow \R$, the estimator $\widehat{m}(w,h)$ can be expressed as \begin{align}
    \label{mwh-early} \widehat{m}(w,h)  =     \E_n \big(   \rho(Y,h_{}(X)) \big[ G_{b,K}^{-1/2} b^K(W) \big]'    \big)               [ \widehat{G}_{b,K}^{o}  ]^{-1} G_{b,K}^{-1/2}   b^K(w).
\end{align}
It follows that
 \begin{align}
 \label{Enm}  &  \E_n \big( \| \widehat{m}(W,h) \|_{\ell^2}^2  \big)   = \sum_{l=1}^{d_{\rho}}  [ \E_n (R_{h, l}^K)      ] ' [ \widehat G_{b,K}^o]^{-1}  [ \E_n (R_{h, l}^K)  ] \\ &  R_{h,l}^K(Z) =     \big[ G_{b,K}^{-1/2}  b^K(W) \big] \rho_{l}(Y,h_{}(X)). \label{Rh} 
\end{align}
As the vector $G_{b,K}^{-1/2} b^K(W)$ is an orthonormal (with respect to $\E$) basis of $\mathcal{V}_K$, the $L^2(\mathbb{P})$ norm of $\Pi_K m(W,h)$ can be expressed as
\begin{align}
\label{Pim}  \|  \Pi_K m(W,h)  \|_{L^2(\mathbb{P})}^2   =  \E  \big( \| \Pi_K m(W,h) \|_{\ell^2}^2 \big) = \sum_{l=1}^{d_{\rho}} \| \E[R_{h,l}^K(Z)]   \|_{\ell^2}^2. 
\end{align}
We denote the empirical analog of this projection by  \begin{align} \label{Pem}
  \|  \widehat{\Pi}_{K} m(W,h_{})  \|_{L^2(\mathbb{P}_n)}^2 =  \sum_{l=1}^{d_{\rho}}  \| \E_n (R_{h, l}^K)      \|_{\ell^2}^2.
\end{align}
Let $\hat{\lambda}_{K,\min}$ and $\hat{\lambda}_{K,\max}$ denote the minimum and maximum eigenvalues of $   [\widehat{G}_{b,K}^o]^{-1}$. By Lemma \ref{aux2}, we have that \begin{equation} \label{eig} \mathbb{P}(0.9 <\hat{\lambda}_{K,\min} \leq  \hat{\lambda}_{K,\max} < 1.1 ) \uparrow 1.
\end{equation}
The proof proceeds through several steps which we outline below.
\begin{enumerate}

\item[\textbf{$(i)$}]
 First, we derive a lower bound for the normalizing constant of the posterior measure.  Specifically, we aim to show that there exists a $C > 0 $ such that \begin{align}   \int  \exp\bigg(    - \frac{n}{2}  \E_n \big[    \widehat{m}(W,h) ' \widehat{\Sigma}(W)  \widehat{m}(W,h)            \big]       \bigg) d  \mu (h|\alpha,K)  \geq   \exp \big(   - C n  \log (n)  \epsilon_{n}^2   \big)  \label{conslb} \end{align}
holds with $\mathbb{P}$ probability approaching $1$.

By Assumption \ref{fsbasis}(ii-iii), the eigenvalues of $\widehat{\Sigma}(W)$ are bounded above by some constant $B < \infty $ with $\mathbb{P}$ probability approaching $1$. Therefore, it suffices to verify that  $$  \int  \exp\bigg(    - \frac{n B }{2}  \E_n \big( \| \widehat{m}(W,h) \|_{\ell^2}^2  \big)   \bigg) d  \mu (h|\alpha,K)  \geq   \exp \big(   - C n \log (n)  \epsilon_{n}^2   \big) . $$
Fix any $l \in \{1 , \dots , d_{\rho} \} $. On the set where (\ref{eig}) holds,  we have that \begin{align*}
 &  \E_n (R_{h, l}^K)      ] ' [ \widehat G_{b,K}^o]^{-1}  [ \E_n (R_{h, l}^K)  ]   \leq 1.1 \| \E_n (R_{h, l}^K)  \|_{\ell^2}^2.
\end{align*} 

 Fix any $M > 0 $.  By Lemma \ref{emp-c}, there exists a $C = C(M) < \infty$ such that  \begin{align*} \sum_{l=1}^{d_{\rho}}
\| \E_n(R_{h,l}^K)       \|_{\ell^2}^2 & \leq   \sum_{l=1}^{d_{\rho}} \big(  \| \E_n(R_{h,l}^K) - \E(R_{h,l}^K)      \|_{\ell^2} + \| \E(R_{h,l}^K)  \|_{\ell^2}  \big)^2                        \nonumber \\  & \leq  C  \bigg(          \frac{K}{n} + \|  \Pi_K m(W,h)  \|_{L^2(\mathbb{P})}^2     \bigg)
\end{align*}
holds for all $h \in \mathbf{H}^t(M)$ (with $\mathbb{P}$ probability approaching $1$). By Assumption \ref{identif}$(ii)$ and the observation that $\Pi_K$ is norm decreasing, the preceding bound implies that \begin{align*}
   \sum_{l=1}^{d_{\rho}} \| \E_n(R_{h,l}^K)       \|_{\ell^2}^2 \leq  C  \bigg(          \frac{K}{n} + \|  h- h_0  \|_{L^2(\mathbb{P})}^2     \bigg)
\end{align*}
holds for all $h \in \mathbf{H}^t(M)$ (with $\mathbb{P}$ probability approaching $1$). From combining the preceding bounds, it follows that \begin{align*}
 &   \int  \exp\bigg(    - \frac{n B}{2}  \E_n \big( \| \widehat{m}(W,h) \|_{\ell^2}^2  \big)   \bigg) d  \mu (h|\alpha,K)  \\ & \geq  \int \limits _{\| h  \|_{\mathbf{H}^t} \leq M, \;  \| h- h_0  \|_{L^2(\mathbb{P})} \leq \epsilon_{n}}   \exp\bigg(    - \frac{n B }{2}  \E_n \big( \| \widehat{m}(W,h) \|_{\ell^2}^2  \big)   \bigg) d  \mu (h|\alpha,K) \\ &  \geq  \exp \big(  -  C n  \epsilon_{n}^2       \big)  \int \limits _{\| h  \|_{\mathbf{H}^t} \leq M, \;  \| h- h_0  \|_{L^2(\mathbb{P})} \leq \epsilon_{n}}   d  \mu (h|\alpha,K).
\end{align*}
 By Lemma \ref{aux3}, there exists a universal constant $D > 0 $ such that for all sufficiently $n$, \begin{align}
\label{gbound}  \int \limits _{\| h  \|_{\mathbf{H}^t} > M  }  d \mu(h| \alpha,K) \leq 2 \exp \big( - D M^2  n \log(n)  \epsilon_{n}^2     \big).
\end{align}

It follows that \begin{align*}
&   \int \limits _{\| h  \|_{ \mathbf{H}^t} \leq M \;  \|  h- h_0 \|_{L^2(\mathbb{P})} \leq   \epsilon_{n} }  d \mu(h| \alpha,K) \\ & = \int \limits_{\| h-h_0  \|_{L^2(\mathbb{P})} \leq  \epsilon_{n}} d \mu(h|\alpha,K) - \int \limits_{\| h-h_0  \|_{L^2(\mathbb{P})  } \leq  C'  \epsilon_{n} , \: \| h \|_{\mathbf{H}^t} > M   } d \mu(h|\alpha,K) \\ & \geq \int \limits_{\| h-h_0  \|_{L^2(\mathbb{P})} \leq  \epsilon_{n}} d \mu(h|\alpha,K)  \; \; - \; \; 2 \exp \big( - D M^2  n \log(n) \epsilon_{n}^2    \big)  \;  ,
\end{align*}
where $D $ is as in (\ref{gbound}). It remains to lower bound the first integral above. The Reproducing Kernel Hilbert Space (RKHS) associated to the Gaussian  random element  $G_{\alpha} $    can be represented as \begin{align}
\label{RKHS-D}  \mathbb{H}_{\alpha}  =  \bigg \{  h \in L^2(\mathcal{X})  : \|h \|_{\mathbb{H}_{\alpha}}^2 =   \sum_{i=1}^{\infty}     i^{ 1 +   2\alpha /d}  \left| \langle h ,  e_i \rangle_{L^2(\mathcal{X})}      \right|^2 < \infty               \bigg \} .
\end{align}
For every $\epsilon > 0 $, denote the concentration function of $\mu(. |\alpha,K)$  at $h_0$ by

\begin{align} 
\label{conf} \varphi_{h_0}(\epsilon|\alpha,K) = \inf_{h \in \mathbb{H}_{\alpha}:\| h - h_0   \|_{L^2( \mathcal{X})} \leq \epsilon  } \bigg\{ \frac{ \log (n)  K}{2} \| h \|_{\mathbb{H}_{\alpha}}^2 - \log \mathbb{P} \bigg(  \| G_{\alpha}   \|_{L^2(\mathbb{P})} < \epsilon \sqrt{ \log (n) K}           \bigg)             \bigg \}.
\end{align}
By \citep[Proposition 11.19]{ghosal2017fundamentals}, we have that \begin{align*}
\int \limits_{\| D_{h_0}[h-h_0]  \|_{L^2(\mathbb{P})} \leq   \epsilon_{n}} d \mu(h|\alpha,K) \geq \exp \big(  - \varphi_{h_0}(0.5  \epsilon_{n}| \alpha,K)        \big).
\end{align*}
Since $h_0 \in \mathcal{H}^p$ for some $p \geq \alpha+d/2$, it follows that $h_0 \in \mathbb{H}_{\alpha}$. In particular, by choosing $h = h_0$ in the infimum defining $\varphi_{h_0}$ in (\ref{conf}), we obtain \begin{align*}
    \varphi_{h_0}(\epsilon|\alpha,K) \leq  C \bigg[ K \log n - \log \mathbb{P} \bigg(  \| G_{\alpha}   \|_{L^2(\mathbb{P})} < \epsilon \sqrt{ \log (n) K}           \bigg)           \bigg].
\end{align*}
From the representation of $G_{\alpha}$ in (\ref{g-series-4}) and an application of  \citep[Lemma 11.47]{ghosal2017fundamentals}, we obtain  \begin{align*}
\varphi_{h_0}(  0.5  \epsilon_{n}| \alpha,K)  \leq C \big[  \log (n)  K +   (   \epsilon_{n}  \sqrt{  \log(n) K})^{-d/\alpha  }        \big].
\end{align*}
Since $\epsilon_n = \sqrt{K} / \sqrt{n}$ and $K \gtrapprox n^{d/2(\alpha + d)} $, the first term on the right of the preceding inequality dominates. It follows that \begin{align*}
    \int \limits_{\|h-h_0  \|_{L^2(\mathbb{P})} \leq   \epsilon_{n}} d \mu(h|\alpha,K) \geq \exp(-E K \log n ) = \exp(-E n \log(n) \epsilon_n^2)
\end{align*}
for some universal constant $E > 0$ (that does not depend on $M$). By choosing $ M > 0$ sufficiently large, we obtain   \begin{align*}
&   \int \limits _{\| h  \|_{ \mathbf{H}^t} \leq M \;  \|  h- h_0 \|_{L^2(\mathbb{P})} \leq   \epsilon_{n} }  d \mu(h| \alpha,K)  \\ & \geq \int \limits_{\| h-h_0  \|_{L^2(\mathbb{P})} \leq  \epsilon_{n}} d \mu(h|\alpha,K)  - 2 \exp \big( - D M^2  n \log(n) \epsilon_{n}^2    \big)  \\ & \geq \exp(-E n \log(n) \epsilon_n^2)   - 2 \exp \big( - D M^2  n \log(n) \epsilon_{n}^2    \big) \\ & \geq \exp(-C n \log(n) \epsilon_n^2).
\end{align*}
The lower bound in (\ref{conslb}) follows from combining all the preceding bounds.

\item[\textbf{$(ii)$}]

We aim to show that given any $D' > 0 $, there exists  $D > 0 $ such that

\begin{align*}      \int \limits_{  \|  \widehat{\Pi}_{K} m(W,h)  \|_{L^2(\mathbb{P}_n)}     >  D  \sqrt{\log (n)} \epsilon_{n}   }    & \exp\bigg(    - \frac{n}{2}  \E_n \big[    \widehat{m}(W,h) ' \widehat{\Sigma}(W)  \widehat{m}(W,h)            \big]       \bigg) d  \mu (h|\alpha,K)  \\ &   \leq    \exp \big(   - D' n \log(n) \epsilon_{n}^2   \big)    
\end{align*}
holds with $\mathbb{P}$ probability approaching $1$. By Assumption \ref{fsbasis}$(ii-iii)$,  we can work under the setting where the eigenvalues of $\widehat{\Sigma}$ are bounded below by some constant $b > 0 $. Therefore, it suffices to verify that
\begin{align*}      \int \limits_{  \|  \widehat{\Pi}_{K} m(W,h)  \|_{L^2(\mathbb{P}_n)}     >  D  \sqrt{\log (n)}  \epsilon_{n}   }  & \exp\bigg(    - \frac{n b }{2}  \E_n \big( \| \widehat{m}(W,h) \|_{\ell^2}^2  \big)   \bigg) d  \mu (h|\alpha,K)  \\ &  \leq   \exp \big(   - D' n \log(n) \epsilon_{n}^2   \big)    .
\end{align*}
On the set where (\ref{eig}) holds,  we have that \begin{align*}
 &  \E_n (R_{h, l}^K)      ] ' [ \widehat G_{b,K}^o]^{-1}  [ \E_n (R_{h, l}^K)  ]   \geq 0.9  \| \E_n (R_{h, l}^K)  \|_{\ell^2}^2.
\end{align*} 
for every $l \in \{1 , \dots , d_{\rho} \} $.   From this bound and the representation given in (\ref{Enm}), it follows immediately that   $ \E_n \big( \| \widehat{m}(W,h) \|_{\ell^2}^2  \big) \geq 0.9  \sum_{l=1}^{d_{\rho}}  \| \E_n (R_{h, l}^K)      \|_{\ell^2}^2  $.  Hence, \begin{align*}
&  \int \limits_{  \|  \widehat{\Pi}_{K} m(W,h)  \|_{L^2(\mathbb{P}_n)}     >  D  \sqrt{\log (n)}  \epsilon_{n}   }  & \exp\bigg(    - \frac{n b }{2}  \E_n \big( \| \widehat{m}(W,h) \|_{\ell^2}^2  \big)   \bigg) d  \mu (h|\alpha,K) \\ & \leq  \int \limits_{  \|  \widehat{\Pi}_{K} m(W,h)  \|_{L^2(\mathbb{P}_n)}     >  D  \sqrt{\log (n)}  \epsilon_{n}   }  & \exp\bigg(    - \frac{n b }{2} 0.9  \|  \widehat{\Pi}_{K} m(W,h)  \|_{L^2(\mathbb{P}_n)}^2     \bigg) d  \mu (h|\alpha,K) \\ & \leq \exp \bigg(   - \frac{nb}{2} 0.9 D^2 \log(n) \epsilon_{n}^2      \bigg).
\end{align*}
For any $D' > 0 $, we can pick $D > 0$ sufficiently large so that the preceding expression is bounded above by $\exp(- D'n \log(n) \epsilon_n^2)$.

\item[\textbf{$(iii)$}]
We prove the main statement of the theorem. By Lemma \ref{aux3},  there exists a universal constant $ E > 0 $ such that  for any $ M > 0$ we have \begin{align*}
\int \limits _{\| h  \|_{\mathbf{H}^t} > M  }  d \mu(h| \alpha,K) \leq 2 \exp \big( - E M  \log(n) K       \big) =  2 \exp \big( - E M n \log(n) \epsilon_{n}^2       \big) .
\end{align*}
for all sufficiently large  $n$.  From this bound and steps $(i-ii)$, it follows that  for any $C ' > 0 $, we can pick $M > 0 $ sufficiently large such that
 \begin{align*} &
  \mu \big( \|  \widehat{\Pi}_{K}  m(W,h_{})   \|_{L^2(\mathbb{P}_n)}  \leq   M  \sqrt{\log n}  \epsilon_{n},  \| h \|_{\mathbf{H}^t} \leq M  \big|  \alpha,K,\mathcal{Z}_n     \big  )           > 1 -   \exp( - C ' n  \log(n)  \epsilon_{n}^2  ) .
\end{align*}
holds with $\mathbb{P}$ probability approaching $1$. By Lemma \ref{emp-c} and observing that $\epsilon_n  = \sqrt{K} / \sqrt{n}$, the preceding inequality further implies that  for any $C ' > 0 $, we can pick $M > 0 $ sufficiently large such that \begin{align*} &
  \mu \big( \|  {\Pi}_{K}  m(W,h_{})   \|_{L^2(\mathbb{P})}  \leq   M  \sqrt{\log n}  \epsilon_{n},  \| h \|_{\mathbf{H}^t} \leq M  \big|  \alpha,K,\mathcal{Z}_n     \big  )           > 1 -   \exp( - C ' n  \log(n)  \epsilon_{n}^2  ) 
\end{align*}
holds with $\mathbb{P}$ probability approaching $1$. By Assumption \ref{fsbasis}$(iii)$ the bias from the $\Pi_K$ projection vanishes uniformly $h \in \mathbf{H}^t(M)$. Define $$ \gamma_n(M) = \max \bigg\{   M \sqrt{\log n} \epsilon_n , \sup_{h \in \mathbf{H}^t(M)  } \|(\Pi_K - I) m(W,h)   \|_{L^2(\mathbb{P})}     \bigg  \} . $$
It follows that for any $C ' > 0 $, we can pick $M > 0 $ sufficiently large such that \begin{align*} &
  \mu \big( \|    m(W,h_{})   \|_{L^2(\mathbb{P})}  \leq   
 \gamma_n,  \| h \|_{\mathbf{H}^t} \leq M  \big|  \alpha,K,\mathcal{Z}_n     \big  )           > 1 -   \exp( - C ' n  \log(n)  \epsilon_{n}^2  ) 
\end{align*}
holds with $\mathbb{P}$ probability approaching $1$. Fix any $\epsilon > 0$. Since $\gamma_n \rightarrow 0$, Lemma \ref{aux4} implies  \begin{align*}
  \mu( \| h - h_0  \|_{L^2(\mathcal{\mathbb{P}})} \leq \epsilon \big| \alpha,K,\mathcal{Z}_n )  & \geq   \mu \big( \|    m(W,h_{})   \|_{L^2(\mathbb{P})}  \leq   
 \gamma_n,  \| h \|_{\mathbf{H}^t} \leq M  \big|  \alpha,K,\mathcal{Z}_n     \big  ) \\ &           > 1 -   \exp( - C ' n  \log(n)  \epsilon_{n}^2  )
\end{align*}
for all sufficiently large $n$. Since $n \log(n) \epsilon_n^2 \rightarrow  \infty$, the claim follows.

\end{enumerate}

\end{proof}

\begin{proof}[Proof of Theorem \ref{t1-2}]
The proof is analogous to the proof of Theorem \ref{t1}. We use the same notation as introduced there and point out the relevant modifications below. Define $$ \epsilon_n = \frac{\sqrt{K_n}}{\sqrt{n}}. $$

First,  we aim to show that there exists a $C > 0 $ such that \begin{align}   \int  \exp\bigg(    - \frac{n}{2}  \E_n \big[    \widehat{m}(W,h) ' \widehat{\Sigma}(W,h)  \widehat{m}(W,h)            \big]       \bigg) d  \mu (h|\alpha,K)  \geq   \exp \big(   - C n  \log (n)  \epsilon_{n}^2   \big) . \label{conslbtnew} \end{align}
Fix any $ M > 0$ and let $h_0 \in \Theta_0 \cap \mathcal{H}^p$ be as in the statement of the Theorem. We start with the lower bound \begin{align*}
    & \int  \exp\bigg(    - \frac{n}{2}  \E_n \big[    \widehat{m}(W,h) ' \widehat{\Sigma}(W,h)  \widehat{m}(W,h)            \big]       \bigg) d  \mu (h|\alpha,K)  \\ & \geq \int_{\| h \|_{\mathbf{H}^t } \leq M, \| h- h_0 \|_{L^2(\mathbb{P})} \leq \epsilon_n}  \exp\bigg(    - \frac{n}{2}  \E_n \big[    \widehat{m}(W,h) ' \widehat{\Sigma}(W,h)  \widehat{m}(W,h)            \big]       \bigg) d  \mu (h|\alpha,K).
\end{align*}
By Condition \ref{fsbasis2}, the eigenvalues of $ \sup_{h \in \mathbf{H}^t(M)} \widehat{\Sigma}(W,h)$ are bounded above by some constant $B < \infty$ with $\mathbb{P}$ probability approaching $1$. Therefore, it suffices to verify that  $$  \int_{\| h \|_{\mathbf{H}^t } \leq M, \| h- h_0 \|_{L^2(\mathbb{P})} \leq \epsilon_n}   \exp\bigg(    - \frac{n B }{2}  \E_n \big( \| \widehat{m}(W,h) \|_{\ell^2}^2  \big)   \bigg) d  \mu (h|\alpha,K)  \geq   \exp \big(   - C n \log (n)  \epsilon_{n}^2   \big) . $$
The argument to verify this is identical to that of part $(i)$ in Theorem \ref{t1}.

Next,  by Lemma \ref{aux3},  there exists a universal constant $ E > 0 $ such that  for any $ M > 0$ we have \begin{align*}
\int \limits _{\| h  \|_{\mathbf{H}^t} > M  }  d \mu(h| \alpha,K) \leq 2 \exp \big( - E M  \log(n) K       \big) =  2 \exp \big( - E M n \log(n) \epsilon_{n}^2       \big) .
\end{align*}
for all sufficiently large  $n$. 

Furthermore, for all $D > 0 $, we also have the trivial bound \begin{align*}      \int \limits_{  \E_n \big[    \widehat{m}(W,h) ' \widehat{\Sigma}(W,h)  \widehat{m}(W,h)            \big]    >  2D  \sqrt{\log (n)} \epsilon_{n}   }    & \exp\bigg(    - \frac{n}{2}  \E_n \big[    \widehat{m}(W,h) ' \widehat{\Sigma}(W,h)  \widehat{m}(W,h)            \big]       \bigg) d  \mu (h|\alpha,K)  \\ &   \leq    \exp \big(   - D n \log(n) \epsilon_{n}^2   \big) .
\end{align*}
From combining the preceding two bounds with (\ref{conslbtnew}), it follows that  for any $C ' > 0 $, we can pick $D,M > 0 $ sufficiently large such that \begin{align*} 
  & \mu^{CU} \big( \E_n \big[    \widehat{m}(W,h) ' \widehat{\Sigma}(W,h)  \widehat{m}(W,h)            \big]   \leq   D  \sqrt{\log n}  \epsilon_{n},  \| h \|_{\mathbf{H}^t} \leq M  \big|  \alpha,K,\mathcal{Z}_n     \big  )    \\ &        > 1 -   \exp( - C ' n  \log(n)  \epsilon_{n}^2  ) 
\end{align*}
holds with $\mathbb{P}$ probability approaching $1$. By Condition \ref{fsbasis2}, the eigenvalues of $ \inf_{h \in \mathbf{H}^t(M)} \widehat{\Sigma}(W,h)$ are bounded below by some constant $c > 0$ with $\mathbb{P}$ probability approaching $1$. From this observation and Lemma \ref{emp-c}, it follows that for any $C ' > 0 $, we can pick $D,M > 0 $ sufficiently large such that \begin{align*} &
  \mu^{CU} \big( \|  {\Pi}_{K}  m(W,h_{})   \|_{L^2(\mathbb{P})}  \leq   D  \sqrt{\log n}  \epsilon_{n},  \| h \|_{\mathbf{H}^t} \leq M  \big|  \alpha,K,\mathcal{Z}_n     \big  )          \\ &  > 1 -   \exp( - C ' n  \log(n)  \epsilon_{n}^2  ) 
\end{align*}
holds with $\mathbb{P}$ probability approaching $1$. 

By Condition \ref{fsbasis2}$(iii)$ the bias from the $\Pi_K$ projection vanishes uniformly $h \in \mathbf{H}^t(M)$. Define $$ \gamma_n(D,M) = \max \bigg\{   D \sqrt{\log n} \epsilon_n , \sup_{h \in \mathbf{H}^t(M)  } \|(\Pi_K - I) m(W,h)   \|_{L^2(\mathbb{P})}     \bigg  \} . $$
It follows that for any $C ' > 0 $, we can pick $D,M > 0 $ sufficiently large such that \begin{align*} &
  \mu^{CU} \big( \|    m(W,h_{})   \|_{L^2(\mathbb{P})}  \leq   
 \gamma_n,  \| h \|_{\mathbf{H}^t} \leq M  \big|  \alpha,K,\mathcal{Z}_n     \big  )           > 1 -   \exp( - C ' n  \log(n)  \epsilon_{n}^2  ) 
\end{align*}
holds with $\mathbb{P}$ probability approaching $1$. Fix any $\epsilon > 0$. Since $\gamma_n \rightarrow 0$, Lemma \ref{aux4-e} implies  \begin{align*}
  \mu^{CU}( d(h, \Theta_0) \leq \epsilon \big| \alpha,K,\mathcal{Z}_n )  & \geq   \mu \big( \|    m(W,h_{})   \|_{L^2(\mathbb{P})}  \leq   
 \gamma_n,  \| h \|_{\mathbf{H}^t} \leq M  \big|  \alpha,K,\mathcal{Z}_n     \big  ) \\ &           > 1 -   \exp( - C ' n  \log(n)  \epsilon_{n}^2  )
\end{align*}
for all sufficiently large $n$. Since $n \log(n) \epsilon_n^2 \rightarrow  \infty$, the claim follows.

\end{proof}

\begin{proof}[Proof of Theorem \ref{rate}]
Let $C, C'$ denote generic universal constants that may change from line to line.  First, we introduce some preliminary notation that will be used throughout the proof. Let $\lambda_K$ and $c_K$ be as in Assumption \ref{lcurv} and \ref{basis-approx}, respectively. Define     \begin{align}    
\epsilon_{n} =  \frac{\sqrt{ K_n}}{\sqrt{n}}  .   \label{epsg-2}
\end{align}
For ease of notation, we surpress the dependence of $K = K_n$ on $n$ in the remainder of the proof. Fix any function $h: \mathcal{X} \rightarrow \R$. From expanding the quadratic form in the quasi-Bayes objective function, we can write
 \begin{align}
 \label{Enmt2}  &  \E_n \big( \| \widehat{m}(W,h) \|_{\ell^2}^2  \big)   = \sum_{l=1}^{d_{\rho}}  [ \E_n (R_{h, l}^K)      ] ' [ \widehat G_{b,K}^o]^{-1}  [ \E_n (R_{h, l}^K)  ] \\ &  R_{h,l}^K(Z) =     \big[ G_{b,K}^{-1/2}  b^K(W) \big] \rho_{l}(Y,h_{}(X)). \label{Rht2} 
\end{align}
Furthermore, the $L^2(\mathbb{P})$ norm of $\Pi_K m(W,h)$ can be expressed as
\begin{align}
\label{Pimt2}  \|  \Pi_K m(W,h)  \|_{L^2(\mathbb{P})}^2   =  \E  \big( \| \Pi_K m(W,h) \|_{\ell^2}^2 \big) = \sum_{l=1}^{d_{\rho}} \| \E[R_{h,l}^K(Z)]   \|_{\ell^2}^2. 
\end{align}
We denote the empirical analog by  \begin{align} \label{Pemt2}
  \|  \widehat{\Pi}_{K} m(W,h_{})  \|_{L^2(\mathbb{P}_n)}^2 =  \sum_{l=1}^{d_{\rho}}  \| \E_n (R_{h, l}^K)      \|_{\ell^2}^2.
\end{align}
Let $\hat{\lambda}_{K,\min}$ and $\hat{\lambda}_{K,\max}$ denote the minimum and maximum eigenvalues of $   [\widehat{G}_{b,K}^o]^{-1}$. By Lemma \ref{aux2}, we have that \begin{equation} \label{eigt2} \mathbb{P}(0.9 <\hat{\lambda}_{K,\min} \leq  \hat{\lambda}_{K,\max} < 1.1 ) \uparrow 1.
\end{equation}
The proof proceeds through several steps which we outline below.
\begin{enumerate}
\item[\textbf{$(i)$}]
 First, we derive a lower bound for the normalizing constant of the posterior measure.  Specifically, we aim to show that there exists a $C > 0 $ such that \begin{align}   \int  \exp\bigg(    - \frac{n}{2}  \E_n \big[    \widehat{m}(W,h) ' \widehat{\Sigma}(W)  \widehat{m}(W,h)            \big]       \bigg) d  \mu (h|\alpha,K)  \geq   \exp \big(   - C n  \log (n)  \epsilon_{n}^2   \big)  \label{conslbt2} \end{align}
holds with $\mathbb{P}$ probability approaching $1$.

By Assumption \ref{fsbasis}(ii-iii), the eigenvalues of $\widehat{\Sigma}(W)$ are bounded above by some constant $B < \infty $ with $\mathbb{P}$ probability approaching $1$. Therefore, it suffices to verify that  $$  \int  \exp\bigg(    - \frac{n B }{2}  \E_n \big( \| \widehat{m}(W,h) \|_{\ell^2}^2  \big)   \bigg) d  \mu (h|\alpha,K)  \geq   \exp \big(   - C n \log (n)  \epsilon_{n}^2   \big) . $$
Fix any $l \in \{1 , \dots , d_{\rho} \} $. On the set where (\ref{eig}) holds,  we have that \begin{align*}
 &  \E_n (R_{h, l}^K)      ] ' [ \widehat G_{b,K}^o]^{-1}  [ \E_n (R_{h, l}^K)  ]   \leq 1.1 \| \E_n (R_{h, l}^K)  \|_{\ell^2}^2.
\end{align*} 

 Fix any $M > 0 $.  By Lemma \ref{emp-c} and the observation that $\Pi_K$ is a norm decreasing projection, there exists a $C = C(M) < \infty$ such that  \begin{align*} \sum_{l=1}^{d_{\rho}}
\| \E_n(R_{h,l}^K)       \|_{\ell^2}^2 & \leq   \sum_{l=1}^{d_{\rho}} \big(  \| \E_n(R_{h,l}^K) - \E(R_{h,l}^K)      \|_{\ell^2} + \| \E(R_{h,l}^K)  \|_{\ell^2}  \big)^2                        \nonumber \\  & \leq  C  \bigg(          \frac{K}{n} + \|  \Pi_K m(W,h)  \|_{L^2(\mathbb{P})}^2     \bigg) \\ & \leq C  \bigg(          \frac{K}{n} + \|   m(W,h)  \|_{L^2(\mathbb{P})}^2     \bigg)
\end{align*}
holds for all $h \in \mathbf{H}^t(M)$ (with $\mathbb{P}$ probability approaching $1$). Let $B=B(M)$ be as in Assumption \ref{lcurv}. By arguing analogously to the proof of Theorem \ref{t1}, there exists a sequence $\delta_n \rightarrow 0$ such that \begin{align} & \{ h \in \mathbf{H}^t(M) : \|  m(W,h)  \|_{L^2(\mathbb{P})} \leq  B \epsilon_n \:, \:  \| D_{h_0}[h- h_0] \|_{L^2(\mathbb{P})} \leq  \epsilon_n  \} \nonumber \\ & \subseteq \Omega_n = \{  h \in \mathbf{H}^t(M) : \|  m(W,h)  \|_{L^2(\mathbb{P})} \leq  B \epsilon_n \; , \|h - h_0 \|_{L^2(\mathbb{P})} \leq \delta_n   \}  \label{inclusion}
\end{align} 
holds for all sufficiently large $n$. From combining the preceding bounds, it follows that \begin{align*}
 &   \int  \exp\bigg(    - \frac{n B}{2}  \E_n \big( \| \widehat{m}(W,h) \|_{\ell^2}^2  \big)   \bigg) d  \mu (h|\alpha,K)  \\ & \geq  \int \limits _{h \in  \Omega_n }   \exp\bigg(    - \frac{n B }{2}  \E_n \big( \| \widehat{m}(W,h) \|_{\ell^2}^2  \big)   \bigg) d  \mu (h|\alpha,K) \\ &  \geq  \exp \big(  -  C n  \epsilon_{n}^2       \big)  \int \limits _{h \in  \Omega_n }   d  \mu (h|\alpha,K).
\end{align*}
Assumption \ref{lcurv} implies that  $  \|  m(W,h)    \|_{L^2(\mathbb{P})}    \leq B  \|  D_{h_0}[h- h_0]  \|_{L^2(\mathbb{P})}   $ for every $h \in \Omega_n$. The preceding bound and the inclusion in (\ref{inclusion}) implies that  \begin{align*}
    \int \limits _{h \in  \Omega_n }   d  \mu (h|\alpha,K) & \geq \int \limits_{ h \in \mathbf{H}^t(M) : \|  m(W,h)  \|_{L^2(\mathbb{P})} \leq B \epsilon_n \:, \: \| D_{h_0}[h- h_0] \|_{L^2(\mathbb{P})} \leq \epsilon_n } d  \mu (h|\alpha,K) \\ & \geq  \int \limits_{h \in \mathbf{H}^t(M) : \| D_{h_0}[h- h_0] \|_{L^2(\mathbb{P})} \leq \epsilon_n }  d  \mu (h|\alpha,K).
\end{align*}
By Assumption \ref{gp-link}, we have $ \| D_{h_0}[h- h_0] \|_{L^2(\mathbb{P})} \leq E \| h - h_0 \|_{w, \sigma} $ for every $h \in \mathbf{H}^t(M)$ where  $ E < \infty $, $ \| h  \|_{w, \sigma}^2 = \sum_{i=1}^{\infty} \sigma_i^2 \left| \langle h , e_i  \rangle  \right|^2   $ is the weak norm as in (\ref{w-norm})  and  $ \sigma = ( \sigma_i)_{i=1}^{\infty}$ is a non-negative bounded sequence of constants. By absorbing $E$ into $\sigma$, we can without loss of generality assume that $E = 1$. It follows that
 \begin{align*}
    & \int \limits_{h \in \mathbf{H}^t(M) : \| D_{h_0}[h- h_0] \|_{L^2(\mathbb{P})} \leq \epsilon_n }  d  \mu (h|\alpha,K)  \geq  \int \limits_{h \in \mathbf{H}^t(M) : \| h- h_0 \|_{w, \sigma} \leq \epsilon_n }  d  \mu (h|\alpha,K)
\end{align*}
By Lemma \ref{aux3}, there exists a universal constant $D > 0 $ such that for all sufficiently $n$, \begin{align}
\label{gboundt2}  \int \limits _{\| h  \|_{\mathbf{H}^t} > M  }  d \mu(h| \alpha,K) \leq 2 \exp \big( - D M^2  n \log(n)  \epsilon_n^2     \big).
\end{align}
From the preceding bound, it follows that \begin{align*}
&    \int \limits_{h \in \mathbf{H}^t(M) : \| h- h_0 \|_{w, \sigma} \leq \epsilon_n } d  \mu (h|\alpha,K) \\ & =  \int \limits_{ \| h- h_0 \|_{w, \sigma} \leq \epsilon_n } d  \mu (h|\alpha,K) - \int \limits_{\| h- h_0 \|_{w, \sigma} \leq \epsilon_n , \: \| h \|_{\mathbf{H}^t} > M   } d \mu(h|\alpha,K) \\ & \geq  \int \limits_{ \| h- h_0 \|_{w, \sigma} \leq \epsilon_n } d  \mu (h|\alpha,K)  \; \; - \; \; 2 \exp \big( - D M^2  n \log(n) \epsilon_{n}^2    \big)  \;  ,
\end{align*}
where $D $ is as in (\ref{gboundt2}). It remains to lower bound the first integral above. From the representation in (\ref{g-series-4}), the Reproducing Kernel Hilbert Space (RKHS) associated to the Gaussian  random element  $G_{\alpha} $    can be represented as
\begin{align}
\label{RKHS-D-t2-0}  \mathbb{H}_{\alpha}  =  \bigg \{  h \in L^2(\mathcal{X})  : \|h \|_{\mathbb{H}_{\alpha}}^2 =   \sum_{i=1}^{\infty}     i^{ 1 +   2\alpha /d}  \left| \langle h ,  e_i \rangle_{L^2(\mathcal{X})}      \right|^2 < \infty               \bigg \} .
\end{align}
For every $\epsilon > 0 $, denote the concentration function of $\mu(. |\alpha,K)$  with respect to $\| . \|_{w,\sigma}$ and $h_0$ by 
\begin{align} 
\label{conft2} \varphi_{h_0}(\epsilon|\alpha,K) = \inf_{h \in \mathbb{H}_{\alpha}:\| h - h_0   \|_{w, \sigma} \leq \epsilon } \bigg\{ \frac{ \log (n)  K}{2} \| h \|_{\mathbb{H}_{\alpha}}^2 - \log \mathbb{P} \bigg(  \| G_{\alpha}   \|_{w, \sigma} < \epsilon \sqrt{ \log (n) K}           \bigg)             \bigg \}.
\end{align}
By \citep[Proposition 11.19]{ghosal2017fundamentals}, we have that \begin{align*}
\int \limits_{\| h-h_0  \|_{w, \sigma} \leq   \epsilon_n} d \mu(h|\alpha,K) \geq \exp \big(  - \varphi_{h_0}(0.5  \epsilon_n| \alpha,K)        \big).
\end{align*}
Since $h_0 \in \mathcal{H}^p$ for some $p \geq \alpha+d/2$, it follows that $h_0 \in \mathbb{H}_{\alpha}$. In particular, by choosing $h = h_0$ in the infimum defining $\varphi_{h_0}$ in (\ref{conf}), we obtain \begin{align*}
    \varphi_{h_0}(\epsilon_n|\alpha,K) \leq  E \bigg[ K \log n - \log \mathbb{P} \bigg(  \| G_{\alpha}   \|_{w,\sigma} < \epsilon_n \sqrt{ \log (n) K}           \bigg)           \bigg].
\end{align*}
for some universal constant $E < \infty$. We claim that   \begin{align}
    \label{concbound} - \log \mathbb{P} \bigg(  \| G_{\alpha}   \|_{w,\sigma} < \epsilon_n \sqrt{ \log (n) K}           \bigg)  \leq E K \log(n).
\end{align} 
for some universal constant $E < \infty$. Consider first the case where the model is mildly ill-posed so that $\sigma_i \asymp i^{-\zeta/d}$ for some $\zeta \geq 0$ as $i \rightarrow \infty$. From the representation of $G_{\alpha}$ in (\ref{g-series-4}) and an application of  \citep[Lemma 11.47]{ghosal2017fundamentals}, we obtain  \begin{align*}
- \log \mathbb{P} \bigg(  \| G_{\alpha}   \|_{w,\sigma} < \epsilon_n \sqrt{ \log (n) K}           \bigg) \leq C     \big(  \epsilon_{n}  \sqrt{  \log(n) K} \big)^{-d/(\alpha+ \zeta)  }       .
\end{align*}
Since $\epsilon_n =   \sqrt{K}/ \sqrt{n} $ and $K = K_n  \asymp n^{\frac{d}{2[\alpha + \zeta] + d}}$, the bound in (\ref{concbound}) follows from observing that
\begin{align*}
  n^{\frac{d}{2(\alpha + \zeta)}}   \lessapprox K_n^{\frac{d}{2( \alpha + \zeta)}}  n^{\frac{d}{2(\alpha + \zeta)}} \asymp K_n^{1 + \frac{d}{\alpha + \zeta} }.
\end{align*}
Now suppose the model is severely ill-posed so that $\sigma_i \asymp \exp(-R i^{\zeta/d} ) $ for some $R, \zeta \geq 0$ as $i \rightarrow \infty$. It follows from \citep[Lemma 5.1]{ray} that    \begin{align*}
- \log \mathbb{P} \bigg(  \| G_{\alpha}   \|_{w,\sigma} < \epsilon_n \sqrt{ \log (n) K}           \bigg) \leq C     \bigg \{  \log  \bigg ( \frac{1}{  \epsilon_{n}  \sqrt{\log (n) K}   }    \bigg )     \bigg \}^{1 + \frac{d}{\zeta} }        .
\end{align*}
Since $\log(  \gamma_n  ) \asymp \log(n) $ and  $K = K_n \asymp (\log n)^{d/\zeta}$, the bound in (\ref{concbound}) follows.

From combining the bound in (\ref{concbound}) with those preceding it, we obtain \begin{align*}
    \int \limits_{\| h-h_0  \|_{w, \sigma} \leq   \epsilon_n} d \mu(h|\alpha,K) \geq \exp(-E n \log(n) \epsilon_n^2  ).
\end{align*}
for some universal constant $E > 0$ (that does not depend on $M$).  By choosing $ M > 0$ sufficiently large, we obtain   \begin{align*}
& \int \limits_{h \in \mathbf{H}^t(M) : \| h- h_0 \|_{w, \sigma} \leq \epsilon_n } d  \mu (h|\alpha,K) \\ & \geq  \int \limits_{ \| h- h_0 \|_{w, \sigma} \leq \epsilon_n } d  \mu (h|\alpha,K)  \; \; - \; \; 2 \exp \big( - D M^2  n \log(n) \epsilon_{n}^2    \big) \\ & \geq \exp(-E n \log(n) \epsilon_n^2  ) -  2 \exp \big( - D M^2  n \log(n) \epsilon_{n}^2    \big)     \\ & \geq \exp(-C n \log(n) \epsilon_n^2 ) .
\end{align*}
The lower bound in (\ref{conslbt2}) follows from combining all the preceding bounds.
\item[\textbf{$(ii)$}]

We aim to show that given any $D' > 0 $, there exists a universal constant $D > 0 $ such that 
\begin{align*}      \int \limits_{  \|  \widehat{\Pi}_{K} m(W,h)  \|_{L^2(\mathbb{P}_n)}     >  D  \sqrt{\log (n)} \epsilon_{n}   }    & \exp\bigg(    - \frac{n}{2}  \E_n \big[    \widehat{m}(W,h) ' \widehat{\Sigma}(W)  \widehat{m}(W,h)            \big]       \bigg) d  \mu (h|\alpha,K)  \\ &   \leq    \exp \big(   - D' n \log(n) \epsilon_{n}^2   \big)    
\end{align*}
holds with $\mathbb{P}$ probability approaching $1$. By Assumption \ref{fsbasis}$(ii-iii)$,  we can work under the setting where the eigenvalues of $\widehat{\Sigma}$ are bounded below by some constant $b > 0 $. Therefore, it suffices to verify that
\begin{align*}      \int \limits_{  \|  \widehat{\Pi}_{K} m(W,h)  \|_{L^2(\mathbb{P}_n)}     >  D  \sqrt{\log (n)}  \epsilon_{n}   }  & \exp\bigg(    - \frac{n b }{2}  \E_n \big( \| \widehat{m}(W,h) \|_{\ell^2}^2  \big)   \bigg) d  \mu (h|\alpha,K)  \\ &  \leq   \exp \big(   - D' n \log(n) \epsilon_{n}^2   \big)    .
\end{align*}
On the set where (\ref{eigt2}) holds,  we have that \begin{align*}
 &  \E_n (R_{h, l}^K)      ] ' [ \widehat G_{b,K}^o]^{-1}  [ \E_n (R_{h, l}^K)  ]   \geq 0.9  \| \E_n (R_{h, l}^K)  \|_{\ell^2}^2.
\end{align*} 
for every $l \in \{1 , \dots , d_{\rho} \} $.   From this bound and the representation given in (\ref{Enmt2}), it follows immediately that   $ \E_n \big( \| \widehat{m}(W,h) \|_{\ell^2}^2  \big) \geq 0.9  \sum_{l=1}^{d_{\rho}}  \| \E_n (R_{h, l}^K)      \|_{\ell^2}^2  $.  Hence, \begin{align*}
&  \int \limits_{  \|  \widehat{\Pi}_{K} m(W,h)  \|_{L^2(\mathbb{P}_n)}     >  D  \sqrt{\log (n)}  \epsilon_{n}   }  & \exp\bigg(    - \frac{n b }{2}  \E_n \big( \| \widehat{m}(W,h) \|_{\ell^2}^2  \big)   \bigg) d  \mu (h|\alpha,K) \\ & \leq  \int \limits_{  \|  \widehat{\Pi}_{K} m(W,h)  \|_{L^2(\mathbb{P}_n)}     >  D  \sqrt{\log (n)}  \epsilon_{n}   }  & \exp\bigg(    - \frac{n b }{2} 0.9  \|  \widehat{\Pi}_{K} m(W,h)  \|_{L^2(\mathbb{P}_n)}^2     \bigg) d  \mu (h|\alpha,K) \\ & \leq \exp \bigg(   - \frac{nb}{2} 0.9 D^2 \log(n) \epsilon_{n}^2      \bigg).
\end{align*}
For any $D' > 0 $, we can pick $D > 0$ sufficiently large so that the preceding expression is bounded above by $\exp(- D'n \log(n) \epsilon_n^2)$.

\item[\textbf{$(iii)$}] Define \begin{align}
r_n = \begin{cases}     (\log n)^{-1} & \text{mildly ill-posed} \; ,  \\ (\log \log n)^{-1} & \text{severely ill-posed.}     \end{cases}  \label{rn}
\end{align}  
We aim to show that given any $D' > 0 $, there exists $ M > 0 $ such that  \begin{align}
\label{gbound2}  & \int \limits _{\| h  \|_{\mathbf{H}^t} > M  }  d \mu(h| \alpha,K) \leq  \exp \big( -  D'  n  \log (n) \epsilon_{n}^2    \big)\; , \\ & \label{gbound3} \int \limits_{\| h \|_{\mathcal{H}^{\alpha - r_n}} > M  r_n^{-1/2}   } d \mu(h| \alpha,K) \leq  \exp \big(   - D' n \log (n) \epsilon_{n}^2  \big) .
\end{align}
The bound in (\ref{gbound2}) follows immediately from an application of Lemma \ref{aux3}. In particular, there exists a universal constant $ E > 0 $ such that  \begin{align*}
\int \limits _{\| h  \|_{\mathbf{H}^t} > M  }  d \mu(h| \alpha,K) \leq 2 \exp \big( - E M  \log(n) K       \big) =  2 \exp \big( - E M n \log(n) \epsilon_{n}^2       \big) 
\end{align*}
holds for all sufficiently large $n$.  

To show (\ref{gbound3}), first observe that the representation of $G_{\alpha}$ in (\ref{g-series-4}) yields      \begin{align}
\E  \big( \| G_{\alpha} \|_{\mathcal{H}^{\alpha - r_n}}^2 \big) = \sum_{i=1}^{\infty} i ^{2(\alpha - r_n)/d} \lambda_i \; \; \; \; \text{where} \; \; \; \lambda_i \asymp i^{-1 - 2 \alpha/d}. \label{gaussenorm}
\end{align}
From the definition of $r_n$, it follows that $    \E  \big( \| G_{\alpha} \|_{\mathcal{H}^{\alpha - r_n}}^2 \big)  \leq C    r_n^{-1}  $. By an application of Lemma \ref{aux3} and the preceding bound in (\ref{gaussenorm}), we obtain (\ref{gbound3}).

\item[\textbf{$(iv)$}]
We prove the main statement of the theorem. Let $r_n$ be as in (\ref{rn}). From the bounds derived in the preceding steps, it follows that  for any $C ' > 0 $, we can pick $C, M > 0 $ sufficiently large such that
 \begin{align*} &
  \mu \big( \|  \widehat{\Pi}_{K}  m(W,h_{})   \|_{L^2(\mathbb{P}_n)}  \leq   C  \sqrt{\log n}  \epsilon_{n},  \| h \|_{\mathbf{H}^t} \leq M\:, \| h \|_{\mathcal{H}^{\alpha-r_n}} \leq M r_n^{-1/2}  \big|  \alpha,K,\mathcal{Z}_n     \big  )      \\ &      > 1 -   \exp( - C ' n  \log(n)  \epsilon_{n}^2  ) .
\end{align*}
holds with $\mathbb{P}$ probability approaching $1$. By Lemma \ref{emp-c}, the preceding inequality further implies   \begin{align*} &
  \mu \big( \|  {\Pi}_{K}  m(W,h_{})   \|_{L^2(\mathbb{P})}  \leq   C  \sqrt{\log n}  \epsilon_{n},  \| h \|_{\mathbf{H}^t} \leq M\:, \| h \|_{\mathcal{H}^{\alpha-r_n}} \leq M r_n^{-1/2}  \big|  \alpha,K,\mathcal{Z}_n     \big  )      \\ &      > 1 -   \exp( - C ' n  \log(n)  \epsilon_{n}^2  ) .
\end{align*}
holds with $\mathbb{P}$ probability approaching $1$. By arguing analogously to the proof of Theorem \ref{t1}, there exists a sequence $\delta_n \rightarrow 0$ such that \begin{align} & \{ h : \| \Pi_K m(W,h)  \|_{L^2(\mathbb{P})} \leq C \sqrt{\log n} \epsilon_n, \| h \|_{\mathbf{H}^t} \leq M\:, \| h \|_{\mathcal{H}^{\alpha-r_n}} \leq M r_n^{-1/2}  \} \nonumber \\ & \subseteq \Omega_n =   \{ \| \Pi_K m(W,h)  \|_{L^2(\mathbb{P})} \leq C \sqrt{\log n} \epsilon_n, \| h \|_{\mathbf{H}^t} \leq M\:, \| h \|_{\mathcal{H}^{\alpha-r_n}} \leq M r_n^{-1/2}, \|h - h_0 \|_{L^2(\mathbb{P})} \leq \delta_n   \} . \label{inclusion2}
\end{align} 
Observe that for both mildly and severely ill-posed models, we have $  K^{r_n} \lessapprox 1 $. From this bound and  Assumption \ref{basis-approx}, it follows that there exists a sequence of constants $\lambda_K < \infty$ such that $\|  (\Pi_{K} - I) m(W,h)  \|_{L^2(\mathbb{P})}  \leq \lambda_K  B \sigma_{K+1}  K^{- \alpha /d}$ for every $h \in \Omega_n$. It follows that
 \begin{align*} &
  \mu \big( \|  m(W,h)    \|_{L^2(\mathbb{P})}  \leq   C  \big[ \sqrt{\log n}  \epsilon_{n}  +  \lambda_{K}\sigma_{K+1} K^{-\alpha/d}  r_n^{-1/2}   \big],  \| h \|_{\mathbf{H}^t} \leq M\:, \| h \|_{\mathcal{H}^{\alpha-r_n}} \leq M r_n^{-1/2}  \big|  \alpha,K,\mathcal{Z}_n     \big  )      \\ &      > 1 -   \exp( - C ' n  \log(n)  \epsilon_{n}^2  ) .
\end{align*}
holds with $\mathbb{P}$ probability approaching $1$. By Assumption \ref{lcurv}, the preceding bound further implies that \begin{align*} &
  \mu \big( \|  D_{h_0}[h-h_0]    \|_{L^2(\mathbb{P})}  \leq   C  \big[ \sqrt{\log n}  \epsilon_{n}  +  \lambda_{K}\sigma_{K+1} K^{-\alpha/d}  r_n^{-1/2}   \big],  \| h \|_{\mathbf{H}^t} \leq M\:, \| h \|_{\mathcal{H}^{\alpha-r_n}} \leq M r_n^{-1/2}  \big|  \alpha,K,\mathcal{Z}_n     \big  )      \\ &      > 1 -   \exp( - C ' n  \log(n)  \epsilon_{n}^2  ) .
\end{align*}
Fix any $h$ that satisfies the preceding requirements. Given the definition of $K = K_{n}$, we have \begin{align*}
   &  \sqrt{\log n}  \epsilon_{n}  +  \lambda_{K}\sigma_{K+1} K^{-\alpha/d}  r_n^{-1/2} \\ &  \lessapprox \xi_n =    \begin{cases}    \lambda_{K_n} n^{- \frac{ \alpha + \zeta }{2[\alpha + \zeta] + d}}  \sqrt{\log n}    & \text{mildly ill-posed} \; , \\ \max \{ \lambda_{K_n} \sqrt{\log \log n}, \sqrt{\log n}    \} (\log n)^{d/(2 \zeta)} n^{-1/2}   & \text{severely ill-posed}.    \end{cases}
\end{align*}
For $h_0 \in \mathcal{H}^p(R)$ and $h$ as above, we obtain
\begin{align*}
\|   h- h_0 \|_{L^2(\mathbb{P})}^2 = &  \sum_{i=1}^{\infty}  \left|  \langle e_i , h - h_0 \rangle_{L^2(\mathbb{P})}     \right|^2  \\ & =  \sum_{i=1}^{J}  \left|  \langle e_i , h - h_0 \rangle_{L^2(\mathbb{P})}     \right|^2 +  \sum_{i> J}^{}  \left|  \langle e_i , h - h_0 \rangle_{L^2(\mathbb{P})}   \right|^2 \\  &  \leq \big (\max_{i \leq J} \sigma_i^{-2}   \big) \sum_{i=1}^{\infty} \sigma_i^2  \left|  \langle e_i , h - h_0 \rangle_{L^2(\mathbb{P})}     \right|^2 + C  r_n^{-1}   J^{- 2 \alpha/d}
\end{align*}
for every $J \in \mathbb{N}$. By  Condition \ref{gp-link}, we further obtain \begin{align*}
    \|   h- h_0 \|_{L^2(\mathbb{P})}^2 & \leq 
 C \big (\max_{i \leq J} \sigma_i^{-2}   \big )   \|   D_{h_0}[h- h_0]   \|_{L^2(\mathbb{P})}^2   + C  r_n^{-1}   J^{- 2 \alpha/d}   \\ & \leq C \big (\max_{i \leq J} \sigma_i^{-2}   \big )  \xi_n^2 + C r_n^{-1} J^{-2 \alpha/d}.
\end{align*}
We have  $  \sigma_J \asymp J^{- \zeta/d}$ in the mildly ill-posed case  and $\sigma_J \asymp \exp( -R J^{\zeta/d})$ in the severely ill-posed case. By setting $ J \asymp n^{d/[2(\alpha + \zeta) + d]}$ in the mildly ill-posed case and $J = \lfloor (c_0 \log n)^{d/\zeta} \rfloor  $ for  a sufficiently small $c_0$ in the severely ill-posed case, we obtain \begin{align*}
     \|   h- h_0 \|_{L^2(\mathbb{P})}^2   \leq \begin{cases}    \lambda_{K_n}^2 n^{- \frac{2 \alpha}{2[\alpha + \zeta] + d}} \log n    & \text{mildly ill-posed} \; , \\   ( \log n)^{- 2 \alpha/ \zeta} \log \log n  & \text{severely ill-posed}.    \end{cases}
\end{align*}
The claim follows.
\end{enumerate}

\end{proof}

\begin{lemma}
\label{aux4-e}
Suppose Assumptions  \ref{data} and \ref{identif}$(ii)$ holds. Let $ \Theta_0 = \{ h \in L^2(\mathcal{X}):  \| m(W,h) \|_{L^2(\mathbb{P})}  = 0 \} $ denote the identified set. Then, given any $\epsilon > 0 $, there exists  $\delta > 0 $ such that \begin{align*} &  h \in \mathbf{H}^t(M) \; , \; \| m(W,h)  \|_{L^2(\mathbb{P})} < \delta  \implies  d(h, \Theta_0) < \epsilon
\end{align*}
where $d(h,\Theta_0) = \inf_{h' \in \Theta_0} \|  h-h' \|_{L^2(\mathbb{P})} $.
\end{lemma}

\begin{proof}[Proof of Lemma \ref{aux4-e}] 
By Assumption \ref{data}, the $\| . \|_{L^2}  $ metric is equivalent to the $\| . \|_{L^2(\mathbb{P})}$ metric. It follows that  $\mathbf{H}^t(M)$ is compact under the $\| . \|_{L^2(\mathbb{P})}$ metric. By Assumption \ref{identif}(ii), the mapping  $ m(W, \: . \:) :  ( \mathbf{H}^t(M) , \| .  \|_{L^2(\mathbb{P})}  ) \rightarrow ( L^2(W,\R^{d_{\rho}}), \| . \|_{L^2(\mathbb{P})}   )   $ is continuous. For any $\epsilon > 0 $, the set $\{ h \in \mathbf{H}^t(M) : d(h,\Theta_0) \geq \epsilon   \}$ is a closed (and hence, compact) subset of $( \mathbf{H}^t(M) , \| .  \|_{L^2(\mathbb{P})}  ) $. As a continuous function over a compact set achieves its infimum, it follows that there exists a $\delta > 0 $ for which \begin{align*}
\inf_{h \in \mathbf{H}^t(M) : d(h,\Theta_0) \geq \epsilon } \|  m(W,h)      \|_{L^2(\mathbb{P})} =  \min_{h \in \mathbf{H}^t(M) : d(h,\Theta_0)\geq \epsilon } \|  m(W,h)      \|_{L^2(\mathbb{P})} \geq \delta.
\end{align*}

\end{proof}

\begin{lemma}
\label{emp-bloc}
Suppose Assumptions (\ref{data}-\ref{residuals2}) hold. Given functions $h(X), h'(X): \mathcal{X} \rightarrow \R $, define
 \begin{align*}  R_{h-h',l} (Z) =\big[ G_{b,K}^{-1/2}  b^K(W) \big] \big[  \big \{  \rho_{}(Y,h_{}(X)) - \rho_{}(Y,h_{}'(X)) \big \}   \big]_{l}  \; \; \; \; \; \; \; \; \; \; \;   l \in \{1 , \dots , d_{\rho} \}.  \end{align*}
Then, given any $ M > 0 $ and a  sequence $\delta_n \downarrow 0$, there exists a universal constant $D  = D(M) < \infty$ such that 
    \begin{align*}
&  \sqrt{n} \sup_{l \in \{1 , \dots ,  d_{\rho} \} }  \E\bigg(     \sup_{h,h'  \in \mathbf{H}^t(M)   : \|  h- h'   \|_{L^2(\mathbb{P})} \leq \delta_n  }   \|   \E_n [ R_{h-h',l}^{K}(Z) ] - \E [ R_{h-h',l}^{K}(Z) ]        \|_{\ell^2}  \bigg)  \\ & \leq D \bigg[   \frac{K^{3/2} \log(K)}{\sqrt{n}}  + \frac{\sqrt{K} \delta_n^{-d/t} }{\sqrt{n}}  + \sqrt{K} \sqrt{\log(K)} \delta_n^{\kappa} + \delta_n^{\kappa - d/(2t)}     \bigg].
\end{align*}

\end{lemma}

\begin{proof}[Proof of Lemma \ref{emp-bloc}]
It suffices to verify the bound for each  $l \in \{1 , \dots , d_{\rho} \}$ individually. Fix any such $l$. For ease of notation, we suppress the dependence on $l$ and denote the  vector by $   R_{h-h',l}^K(Z) = R_{h-h'}^K(Z)$. Observe that
 \begin{align*} &
\E_{} \bigg[ \sup_{h,h' \in \mathbf{H}^t(M)  : \|  h- h'   \|_{L^2(\mathbb{P})} \leq \delta_n  }  \|   \E_n [R_{h-h'}^K(Z)] - \E_{}[R_{h-h'}^K(Z)]       \|_{\ell^2}  \bigg] \\ & =  \frac{1}{\sqrt{n}} \E \bigg[  \sup_{h,h'  \in \mathbf{H}^t(M)  : \|  h- h'   \|_{L^2(\mathbb{P})} \leq \delta_n  }  \sup_{\gamma \in  \mathbb{S}^{K-1} }  \frac{1}{\sqrt{n}} \sum_{i=1}^n  \gamma' \big(  R_{h-h'}^K(Z_i) -\E_{} [ R_{h-h'}^K(Z) ] \:  \big)           \:  \bigg]  
\end{align*} 
where $\mathbb{S}^{K-1} =  \{  v \in \R^K : \| v \|_{\ell^2} = 1  \}$. Define the class of functions $$\mathcal{F}_K  = \{ \gamma' R_{h-h'}^K(Z)   : h,h'  \in \mathbf{H}^t(M)  \: , \: \| h-h'  \|_{L^2(\mathbb{P})} \leq \delta_n  \: , \:  \gamma \in   \mathbb{S}^{K-1}   \}    . $$
Denote the associated envelope function by $F_{K}(Z_i) =   \sup_{f \in \mathcal{F}_K} \left|   f(Z_i)  \right|   $.  Let $C_4(M)$ be as in Assumption \ref{residuals2}(iii).  By Cauchy-Schwarz and $\delta_n \downarrow 0$,  it follows that \begin{align*}
F_K(Z_i) & \leq  \sup_{\gamma \in \mathbb{S}^{K-1}} \left| \gamma' G_{b,K}^{-1/2} b^K(W) \right| \sup_{h,h' \in \mathbf{H}^t(M), \|h - h'  \|_{L^2(\mathbb{P}) \leq \delta_n}  } \left|  \rho_{l}(Y,h_{}(X)) - \rho_{l}(Y,h_{}'(X)) \right| \\ & \leq C_4 \zeta_{b,K}.
\end{align*}
 Let $C_1(M)$ be as in Assumption \ref{residuals}. For any fixed $\gamma \in \mathbb{S}^{K-1}$, we have that  \begin{align*}
  &  \sup_{h,h' \in \mathbf{H}^t(M), \|h - h'  \|_{L^2(\mathbb{P}) \leq \delta_n}  }  \E \big[  \left|  \gamma' R_{h-h'}^K(Z)   \right|^2     \big]      \\& =   \sup_{h,h' \in \mathbf{H}^t(M), \|h - h'  \|_{L^2(\mathbb{P}) \leq \delta_n}  }   \E \big[   \gamma' G_{b,K}^{-1/2} b^K(W) b^K(W)'  G_{b,K}^{-1/2} \gamma   \left|  \rho_{l}(Y, h(X))  - \rho_{l}(Y, h'(X))   \right|^2     \big] \\ & \leq C_1^2 \delta_n^{2 \kappa}  \E \big[   \gamma' G_{b,K}^{-1/2} b^K(W) b^K(W)' G_{b,K}^{-1/2} \gamma \big] \\ & = C_1^2 \delta_n^{2 \kappa}  \gamma' G_{b,K}^{-1/2} \E \big[ b^K(W) b^K(W)' \big] G_{b,K}^{-1/2} \gamma     \\ &  = C_1^2 \delta_n^{2 \kappa}.
\end{align*}
For ease of exposition in the remainder of the proof, define \begin{align}
    \label{sigman} \sigma_n = \delta_n^{\kappa}.
\end{align}
From the preceding bound, it follows that  $
    \sup_{f \in \mathcal{F}_K}  \|  f \|_{L^2(\mathbb{P})} \leq C_1  \sigma_n$. By an application of \citep[Proposition 3.5.15]{gine2021mathematical},  there exists a universal constant $ L > 0 $ such that \begin{align*}
& \E_{} \bigg[\sup_{h,h' \in \mathbf{H}^t(M) : \|  h- h'   \|_{L^2(\mathbb{P})} \leq \delta_n  }  \|   \E_n [R_{h-h'}^K(Z)] - \E_{}[R_{h-h'}^K(Z)]       \|_{\ell^2}  \bigg] \\ &   \leq \frac{L}{\sqrt{n}}  \int_{0}^{ 2 \sigma_n  }  \sqrt{\log N_{[]}(\mathcal{F}_K, \| . \|_{L^2(\mathbb{P})},  \epsilon  ) } d \epsilon \bigg( 1 +         \frac{\zeta_{b,K}}{\sigma_n^2 \sqrt{n}} \int_{0}^{ 2 \sigma_n  }  \sqrt{\log N_{[]}(\mathcal{F}_K, \| . \|_{L^2(\mathbb{P})},  \epsilon  ) } d \epsilon  \bigg).
\end{align*}

Fix any $\delta > 0$. Let $\{  h_i\}_{i=1}^{T_1}  $  denote a  $\delta$ covering of  $\big( \mathbf{H}^t(M)    , \| . \|_{\infty}  \big)$ and $\{ \gamma_{m} \}_{ m=1}^{T_2} $ denote a $\delta$ covering of $(\mathbb{S}^{K-1} , \| . \|_{\ell^2})$. For $i,j \in \{1 , \dots , T_1 \}$ and  $m \in \{ 1 , \dots , T_2  \}$, define the functions   \begin{align*} e_{i,j,m}(Z) & =   \sup_{ \gamma \in \mathbb{S}^{K-1} : \| \gamma - \gamma_{m}   \|_{\ell^2} < \delta \; , \; h \in \mathbf{H}^t(M) \; , \; h' \in \mathbf{H}^t(M)  } \left| (\gamma - \gamma_m)' [  R_{h}^K(Z) - R_{h'}^K(Z)    ]   \right| \\ & + \sup_{\gamma \in \mathbb{S}^{K-1} \; , \; h \in \mathbf{H}^t(M) :  \| h- h_i   \|_{\infty} < \delta }  \left|  \gamma' [    R_{h}^K(Z) - R_{h_i}^K(Z)       ]        \right| \\ & +  \sup_{\gamma \in \mathbb{S}^{K-1} \; , \; h \in \mathbf{H}^t(M) :  \| h- h_j   \|_{\infty} < \delta }  \left|  \gamma' [    R_{h}^K(Z) - R_{h_j}^K(Z)       ]  \right| .
\end{align*}

Observe that $$ \bigg \{   \gamma_m ' [ R_{h_i}^K(Z) - R_{h_j}^K(Z)  ]    - e_{i,j,m} \; ,  \;   \gamma_m ' [ R_{h_i}^K(Z) - R_{h_j}^K(Z)  ]   + e_{i,j,m}   \bigg \}_{(i,j) \in \{1 , \dots , T_1 \}  \; , m \in \{ 1 , \dots , T_2 \}  }  $$ is a bracket covering for $\mathcal{F}_K$.  Let $C_1(M)$ and $\kappa \in (0,1]$ be as in Assumption \ref{residuals}.  Let $C_2(M)$ be as in Assumption \ref{residuals2}(i).  By Assumptions (\ref{residuals},  \ref{residuals2}(i))  and Cauchy-Schwarz,  we have that

 \begin{align*}
 \| e_{i,j,m} \|_{L^2(\mathbb{P})}  &   \leq   \bigg \|  \sup_{ \gamma \in \mathbb{S}^{K-1} : \| \gamma - \gamma_{m}   \|_{\ell^2} < \delta \; , \; h \in \mathbf{H}^t(M) \; , \; h' \in \mathbf{H}^t(M)  } \left| (\gamma - \gamma_m)' [  R_{h}^K(Z) - R_{h'}^K(Z)    ]   \right|   \bigg \|_{L^2(\mathbb{P})} \\ & + \bigg \|   \sup_{\gamma \in \mathbb{S}^{K-1} \; , \; h \in \mathbf{H}^t(M) :  \| h- h_i   \|_{\infty} < \delta }  \left|  \gamma' [    R_{h}^K(Z) - R_{h_i}^K(Z)       ]        \right|          \bigg \|_{L^2(\mathbb{P})} \\ & + \bigg \|     \sup_{\gamma \in \mathbb{S}^{K-1} \; , \; h \in \mathbf{H}^t(M) :  \| h- h_j   \|_{\infty} < \delta }  \left|  \gamma' [    R_{h}^K(Z) - R_{h_j}^K(Z)       ]  \right|       \bigg \|_{L^2(\mathbb{P})} \\ & \leq 2  \delta  \zeta_{b,K} C_2 +  \delta^{\kappa} \zeta_{b,K} C_1 + \delta^{\kappa} \zeta_{b,K} C_1. 
\end{align*}
In particular, for all $\delta \in (0,1]$, we have that $ \| e_{i,j,m} \|_{L^2(\mathbb{P})}  \leq C \delta^{\kappa} \zeta_{b,K}  $ for $C = 2 C_2 + 2 C_1$. 

By \citep[Proposition C.7]{ghosal2017fundamentals}, we have that $  \log N_{}( \mathbf{H}^t(M),  \| . \|_{ \infty}, \epsilon) \lessapprox  \epsilon^{-d/t}$ as $\epsilon \downarrow 0$.  Since $ \log N(\mathbb{S}^{K-1},\| .\|_{\ell^2}, \epsilon ) \leq K \log (3 \epsilon^{-1} )  $, it follows that there exists a universal constant $ L > 0 $ such that  \begin{align*}
& \int_{0}^{ 2 \sigma_n  }  \sqrt{\log N_{[]}(\mathcal{F}_K, \| . \|_{L^2(\mathbb{P})},  \epsilon  ) } d \epsilon \\ & \leq L \bigg( \sqrt{K} \sqrt{\log \zeta_{b,K} } \sigma_n  + \sqrt{K} \int_{0}^{2 \sigma_n}  \sqrt{\log(\epsilon^{-1})} d \epsilon  + \int_{0}^{2 \sigma_n}  \epsilon^{-d/2 \kappa t}         d \epsilon  \bigg) \\ & \leq L \bigg(   \sqrt{K} \sqrt{\log \zeta_{b,K} } \sigma_n  + \sqrt{K} \sigma_n \sqrt{\log (\sigma_n^{-1})}   + \sigma_n^{1- d/ (2 \kappa t)}   \bigg)  \\ & \leq L \bigg( \sqrt{K} \sqrt{\log K} \sigma_n +   \sigma_n^{1- d/ (2 \kappa t)}      \bigg).
\end{align*}
From the preceding bounds and $  \zeta_{b,K} \lessapprox \sqrt{K} $,  it follows that   \begin{align*} 
 & \sqrt{n} \E_{} \bigg[\sup_{h,h' \in \mathbf{H}^t(M) : \|  h- h'   \|_{L^2(\mathbb{P})} \leq \delta_n  }  \|   \E_n [R_{h-h'}^K(Z)] - \E_{}[R_{h-h'}^K(Z)]       \|_{\ell^2}  \bigg] \\ &   \leq L  \int_{0}^{ 2 \sigma_n  }  \sqrt{\log N_{[]}(\mathcal{F}_K, \| . \|_{L^2(\mathbb{P})},  \epsilon  ) } d \epsilon \bigg( 1 +         \frac{\zeta_{b,K}}{\sigma_n^2 \sqrt{n}} \int_{0}^{ 2 \sigma_n  }  \sqrt{\log N_{[]}(\mathcal{F}_K, \| . \|_{L^2(\mathbb{P})},  \epsilon  ) } d \epsilon  \bigg)  \\ &   \lessapprox \bigg( \sqrt{K} \sqrt{\log K} \sigma_n +   \sigma_n^{1- d/ (2 \kappa t)}      \bigg) + \bigg( \sqrt{K} \sqrt{\log K} \sigma_n +   \sigma_n^{1- d/ (2 \kappa t)}      \bigg)^2 \frac{\sqrt{K}}{ \sigma_n^2 \sqrt{n}}. 
 \end{align*}
By substituting $\sigma_n = \delta_n^{\kappa}$, the preceding term reduces to  
  \begin{align*}
 \bigg( \sqrt{K} \sqrt{\log K} \delta_n^{\kappa} +  \delta_n^{\kappa- d/(2t) }            \bigg)  + \bigg( \sqrt{K} \sqrt{\log K} \delta_n^{\kappa}  +  \delta_n^{\kappa- d/(2t) }        \bigg)^2  \frac{\sqrt{K}}{ \delta_n^{2 \kappa}  \sqrt{n}}.
\end{align*}
The claim follows.

\end{proof}

\begin{lemma}
\label{empsq}
Suppose Assumptions (\ref{data}-\ref{residuals2}) hold. For each realization of $W$, let $\Sigma(W)$ denote a positive definite matrix such that $\mathbb{P}(  \| \Sigma(W) \|_{op} \leq C  ) = 1$ for some $C > 0$. Given any fixed $M > 0 $ and sequences $\delta_n, \gamma_n \downarrow 0$, define the set  \begin{align*}
\Theta_n = \{  h \in \mathbf{H}^t(M) :   \E (  \|  \Pi_K m(W,h)  \|_{\ell^2}^2 ) \leq M  \gamma_n ^2 , \; \| h - h_0  \|_{L^2(\mathbb{P})} \leq M \delta_n  \}.
\end{align*}
Then, there exists a universal constants $D,R < \infty$ such that \begin{align*}
     & \E \bigg( \sup_{h \in  \Theta_n}  \left|
\sum_{i=1}^n \bigg \{ [\Pi_K m(W_i,h)]' \Sigma(W_i) [\Pi_K m(W_i,h)]  - \E  \big( [\Pi_K m(W,h)]' \Sigma(W) [\Pi_K m(W,h)] \big) \bigg \} \right| \bigg) \\ & \leq R \bigg[  \sqrt{n} \gamma_n ^2 K \mathcal{J}(K^{-1/2}) + \gamma_n ^2 K^3 \mathcal{J}^2(K^{-1/2})  \bigg]
\end{align*}
where $\mathcal{J}(.)$ is defined by \begin{align*}
   &  \mathcal{J}(c) =  \int_{0}^{c}  \sqrt{\log N(\mathcal{M}_n , \| . \|_{L^2(\mathbb{P})}, \tau  D \gamma_n   )  } d \tau \; \; \; \; \; \; \; \; \forall  \; c > 0 \\ & \mathcal{M}_n = \{ m(w,h) : h \in \Theta_n   \}.
\end{align*}

\end{lemma}

\begin{proof}[Proof of Lemma \ref{empsq}]
Define the class of functions \begin{align*}
\mathcal{F} = \{  g : g(.) = [\Pi_K m(.,h)]' \Sigma(.) [\Pi_K m(.,h)]  : h \in \Theta_n   \}.
\end{align*}
For every fixed $h \in \Theta_n$, we have that
\begin{align*}
  \Pi_{K}[ m(W,h)] &  = \sum_{i=1}^K c_{h,i} [  G_{b,K}^{-1/2} b^K(W) ]_{i} \; \;  , \; \; c_{h,i} =  \E \big[  \rho \big( Y , h(X)   \big) [  G_{b,K}^{-1/2} b^K(W) ]_{i}  \big]  \; ,  \end{align*}
where $[  G_{b,K}^{-1/2} b^K(W) ]_{i}$ denotes the $i^{th}$ element of the vector $  G_{b,K}^{-1/2} b^K(W) $.  For every $ l \in \{  1 , \dots , d_{\rho} \} $, denote by $c_{h}^l$ the coefficient vector \begin{align*}
c_{h}^l = \big \{    \E \big[  \rho_{l} \big( Y , h(X)   \big) [  G_{b,K}^{-1/2} b^K(W) ]_{i}  \big]                    \big   \}_{i=1}^K.
\end{align*}
Observe that  $ \sum_{l=1}^{d_{\rho}}  \| c_h^l \|_{\ell^2}^2 =  \E (  \|  \Pi_K m(W,h)  \|_{\ell^2}^2 ) $. Let $C > 0$ be such that $ \mathbb{P} \big( \| \Sigma(W) \|_{op} \leq C \big) = 1  $. By Cauchy-Schwarz and the definition of $\Theta_n$,  it follows that \begin{align*}
\sup_{g \in \mathcal{F}} \left|    g(W)     \right| \leq C  \sup_{h \in \Theta_n} \| \Pi_K m(W,h)    \|_{\ell^2}^2 &  \leq C \zeta_{b,K}^2  \sum_{l=1}^{d_{\rho}} \|   c_h^l \|_{\ell^2}^2  \\ & = C \zeta_{b,K}^2  \E (  \|  \Pi_K m(W,h)  \|_{\ell^2}^2 ) \\ & \leq C M \zeta_{b,K}^2 \gamma_n ^2.
\end{align*}
From the estimate $\zeta_{b,K} \lessapprox \sqrt{K}$, it follows that $\sup_{g \in \mathcal{F}} \left|  g(W)  \right|  \leq C \gamma_n ^2 K  $ for some constant $C < \infty$. It follows that we can take  $F = C \gamma_n ^2 K$ to be an envelope of $\mathcal{F}$.
From this bound and the definition of $\Theta_n$,  we also obtain \begin{align*}
\sup_{g \in \mathcal{F}} \E[ g^2(W) ] \leq F \sup_{g \in \mathcal{F}} \E[\left| g(W) \right|] 
 \lessapprox F  \sup_{h \in \Theta_n}  \E (  \|  \Pi_K m(W,h)  \|_{\ell^2}^2 ) \lessapprox \gamma_n^4 K.
\end{align*} 
From similar arguments to those employed above, we have for every fixed $h,h' \in \Theta_n$,  the bound \begin{align*}     
&   \sup_{w}  \left| [\Pi_K m(w,h)]' \Sigma(w) [\Pi_K m(w,h)] - [\Pi_K m(w,h')]' \Sigma(w) [\Pi_K m(w,h')] \right| \\ & \lessapprox    \sup_{w} \sup_{g \in \Theta_n}  \left| [\Pi_K m(w,h) - \Pi_K m(w,h')]' \Sigma(w) [\Pi_K m(w,g)]  \right|       \\ & \lessapprox  \sqrt{F}  \sup_{w} \|  \Pi_K m(w,h) - \Pi_K m(w,h')    \|_{\ell^2}         \\ & \lessapprox \gamma_n   K  \sqrt{   \E \big(  \|  \Pi_K m(W,h) - \Pi_K m(W,h')      \|_{\ell^2}^2 \big) }.\\ & \lessapprox \gamma_n   K  \sqrt{   \E \big(  \|   m(W,h) -  m(W,h')      \|_{\ell^2}^2 \big) }.
\end{align*}
In particular,  there exists a universal constant $c >0 $ such that
 \begin{align*} \sup_{Q}
 \log N(  \mathcal{F},   \| . \|_{L^2(Q)}  , \tau F  )  \leq  \log N(\mathcal{M}, \| . \|_{L^2(\mathbb{P})} , c \tau \gamma_n   ) \; \; \; \;\; \; \;\forall \; \; \tau \in (0,1) \;\; ,
\end{align*}
where the supremum is over all discrete probability measures $Q$ on $\mathcal{W}$. From an application of \citep[Theorem 3.5.4]{gine2021mathematical}, it follows that \begin{align*}  \E \bigg( \sup_{g \in \mathcal{F}}  \left|
\sum_{i=1}^n g(W_i) - \E g(W) \right| \bigg)  \lessapprox \sqrt{n} \gamma_n^2 K \mathcal{J}(K^{-1/2}) + \gamma_n^2 K^3 \mathcal{J}^2(K^{-1/2}) .
\end{align*}

\end{proof}

\begin{proof}[Proof of Theorem \ref{bvm}]
Let $C, C'$ denote generic universal constants that may change from line to line. Before proceeding with the proof, we introduce a few definitions, clarify some notation and state a few preliminary observations that will be used throughout the proof. 

Given a  positive semi-definite matrix $\Sigma \in \R^{\rho \times \rho}$, we denote the inner product and norm induced by $\Sigma$ as $ \langle v, w \rangle_{\Sigma} = v' \Sigma w $ and $\| v \|_{\Sigma}^2 = v' \Sigma v$, respectively. With this notation, the quasi-Bayes posterior can be expressed as  \begin{align}  \mu(.|\alpha,K,\mathcal{Z}_n) = \frac{   \exp\big(    - \frac{n}{2}  \E_n \big( \|  \widehat{m}(W,.)  \|_{\widehat{\Sigma}(W)}^2    \big)      \big) d \mu(.|\alpha,K)  }{\int_{} \exp\big(    - \frac{n}{2} \E_n \big( \|  \widehat{m}(W,h)  \|_{\widehat{\Sigma}(W)}^2    \big)      \big) d \mu(h|\alpha,K) } . 
 \label{qb-matrix} \end{align}

Given functions $h,g : \mathcal{X} \rightarrow \R $, for ease of notation, we denote the differenced empirical estimate and projection at $(h,g)$ by
 \begin{align}
  \label{h-g} &  \widehat{m}(W,h-g)  =     \widehat{m}(W,h) - \widehat{m}(W,g) \; ,  \\ & \label{h-g2} \Pi_K m(W,h-g) = \Pi_K m(W,h) - \Pi_K m(W,g).
\end{align}
For close variations of  $\widehat{m}$ introduced below such as $\widetilde{m}$ in (\ref{mwh2}), the notation $\widetilde{m}(W,h-g)$ is to be interpreted similar to (\ref{h-g}). Given a function $h : \mathcal{X} \rightarrow \R$ and $t \in \R$, we denote by $h_t$ the quantity \begin{align}
 \label{ht}   h_t = h - \frac{ t}{\sqrt{n}} \tilde{\Phi}.
\end{align}
Given a vector $v \in \R^n$, we denote the least squares projection of $v$ onto the subspace spanned by $ \{  b_1(W_i) , \dots , b_K(W_i)  \}_{i=1}^n $ by $\widehat\Pi_{K}[v]$. In particular, for every $h: \mathcal{X} \rightarrow \R$ and $l \in \{1 , \dots , d_{\rho}  \}$, we have
\begin{align}
    \label{emp-proj-op} \widehat{\Pi}_K [  \{ \rho_{l}(Y_i,h(X_i)  \}_{i=1}^n ] =   \{ \widehat{m}(W_i,h) \}_{i=1}^n.
\end{align}
From the representation in (\ref{g-series-4}), the Reproducing Kernel Hilbert Space (RKHS) associated to the Gaussian  random element  $G_{\alpha} $    can be represented as
\begin{align}
\label{RKHS-D-t2}  \mathbb{H}_{}  =  \bigg \{  h \in L^2(\mathcal{X})  : \|h \|_{\mathbb{H}_{}}^2 =   \sum_{i=1}^{\infty}     i^{ 1 +   2\alpha /d}  \left| \langle h ,  e_i \rangle_{L^2(\mathcal{X})}      \right|^2 < \infty               \bigg \} .
\end{align}
From this representation, it is straightforward to verify that the RKHS $(\mathbb{H}_n, \| . \|_{\mathbb{H}_n})$ of the scaled measure $ \mu (.|\alpha,K) \sim  G_{\alpha} / \sqrt{K \log n}$ is given by  \begin{align}
\label{RKHS-D2-t2}  \mathbb{H}_{n}  =  \bigg \{  h \in L^2(\mathcal{X})  : \|h \|_{\mathbb{H}_{n}}^2 = K \log(n)   \sum_{i=1}^{\infty}     i^{ 1 +   2\alpha /d}  \left| \langle h ,  e_i \rangle_{L^2(\mathcal{X})}      \right|^2 < \infty               \bigg \} .
\end{align}
Since $\tilde{\Phi} \in \mathbb{H}$, it follows immediately that $\tilde{\Phi} \in \mathbb{H}_n $. By definition of the RKHS, we have  $$ h \sim \mu(.|\alpha,K_n) \implies  \langle h , \tilde{\Phi}    \rangle_{\mathbb{H}_n} \sim N(0,\|  \tilde{\Phi}  \|_{\mathbb{H}_n}^2 ) .  $$
Define the sequences \begin{align}
\label{seqs-0}  & \epsilon_{n}^{} =    \frac{  \sqrt{ K}}{\sqrt{n}}     \; \; , \; \; \delta_{n} =  \begin{cases}     n^{- \frac{\alpha}{2[\alpha + \zeta] + d}} \sqrt{\log n}     & \text{mildly ill-posed}  \\  ( \log n)^{- \alpha/ \zeta} \sqrt{\log \log n}   & \text{severely ill-posed}.    \end{cases}  \\ & r_n = \begin{cases}     (\log n)^{-1} & \text{mildly ill-posed}   \\ (\log \log n)^{-1} & \text{severely ill-posed.}     \end{cases} \nonumber
\end{align}

Given any $D,M > 0$, define the set $\Theta_n = \Theta_n(M,D)$ by  \begin{align}
\label{thetan} \Theta_n =  \bigg \{ h \in \mathbf{H}^t(M)   :&  \|   m(W,h_{})   \|_{L^2(\mathbb{P})}  \leq   D \sqrt{\log n}  \epsilon_{n},  \;  \E_n \big( \| \widehat{m}(W,h) \|_{\ell^2}^2  \big)  \leq   D^2 \log(n)   \epsilon_{n}^2, \\ \nonumber &       \| h- h_0 \|_{L^2(\mathbb{P})} \leq D \delta_n \;, \big|  \langle h ,  \tilde{\Phi}   \rangle_{\mathbb{H}_n}    \big| \leq M \sqrt{n} \sqrt{\log n} \epsilon_n \|  \tilde{\Phi}  \|_{\mathbb{H}_n}, \\ \nonumber   &  \|h \|_{\mathcal{H}^{\alpha-r_n}} \leq M r_n^{-1/2} , \| D_{h_0}[h-h_0]  \|_{L^2(\mathbb{P})} \leq D \sqrt{\log n} \epsilon_n    \bigg \}.
\end{align}
The proof proceeds through several steps which we outline below.
\begin{enumerate}

\item[$(i)$] From the proof of Theorem \ref{rate} and an application of Lemma \ref{aux3} to the Gaussian random variable $ Z_n = \langle h ,  \tilde{\Phi}   \rangle_{\mathbb{H}_n}$, we can choose $D,M > 0 $ large enough such that \begin{align}  \label{theta-n-bound} \mu(\Theta_n^c|\alpha,K,\mathcal{Z}_n)  \leq e^{-R n \epsilon_n^2 \log n}    \end{align}
holds with $\mathbb{P}$ probability approaching $1$, where $R > 0$ is a universal constant (that depends on $D,M$).

Define the localized posterior measure generated by restricting $\mu(.|\alpha,K,\mathcal{Z}_n)$ to $\Theta_n$ by \begin{align} \label{mu-star} \mu^{\star}(A|\alpha,K,\mathcal{Z}_n)  = \frac{  \int_{A \cap \Theta_n} \exp\big(    - \frac{n}{2}  \E_n \big[ \|  \widehat{m}(W,h)  \|_{\widehat{\Sigma}(W)}^2    \big]      \big) d \mu(h|\alpha,K)  }{\int_{\Theta_n} \exp\big(    - \frac{n}{2} \E_n \big[ \|  \widehat{m}(W,h)  \|_{\widehat{\Sigma}(W)}^2    \big]      \big) d \mu(h|\alpha,K) }  
\end{align}
for every Borel set $A$. If $\| . \|_{TV}$ denotes the total variation metric, it follows that
  \begin{align}
\|  \mu(.|\alpha,K,\mathcal{Z}_n)  - \mu^{\star} (.|\alpha,K,\mathcal{Z}_n)    \|_{TV} \leq 2  \mu(\Theta_n^c|\alpha,K,\mathcal{Z}_n) \xrightarrow{\mathbb{P}} 0. \label{tvconv}
\end{align}
Therefore, it suffices to verify the weak convergence under the localized measure $\mu^{\star}(.|\alpha,K,\mathcal{Z}_n)$. 
\item[$(ii)$]
We verify that \begin{align*}
\sup_{h \in \Theta_n} \left| \E_n  \big(   \|   \widehat{m}(W,h-h_0)   \|_{\widehat{\Sigma}(W)}^2    \big) - \E \big \{  \Pi_K m(W,h) ' \Sigma(W) \Pi_K m(W,h)       \big \}                  \right| =o_{\mathbb{P}}(n^{-1}).
\end{align*}
This proceeds in several steps. From the definition of $\Theta_n$, we have that \begin{align*}
&  \sup_{h \in \Theta_n}   \left| \E_n  \big(  \|   \widehat{m}(W,h- h_0)   \|_{\widehat{\Sigma}(W)}^2    \big)  - \E_n  \big(  \|   \widehat{m}(W,h- h_0)   \|_{\Sigma(W)}^2    \big)     \right|  \\ & \leq    \E_n  \big(  \|   \widehat{m}(W,h- h_0)   \|_{\ell^2}^2  \|  \widehat{\Sigma}(W)  - \Sigma(W)    \|_{op}    \big)   \\  & \leq   \sup_{w \in \mathcal{W}}  \|  \widehat{\Sigma}(w)  - \Sigma(w)    \|_{op}     \E_n  \big(  \|   \widehat{m}(W,h- h_0)   \|_{\ell^2}^2 \big) \\ & =   \epsilon_{n}^2 \log(n)  O_{\mathbb{P}} \bigg( \sup_{w \in \mathcal{W}}  \|  \widehat{\Sigma}(w)  - \Sigma(w)    \|_{op}   \bigg)    \\ & =  n^{-1}  O_{\mathbb{P}} \big(  \gamma_n \log(n) K_n   \big) \\ & = n^{-1} o_{\mathbb{P}}(1) .
\end{align*}
For any fixed $h: \mathcal{X} \rightarrow \R $, the estimator $\widehat{m}(w,h)$ can be expressed as \begin{align}
    \label{mwh} \widehat{m}(w,h)  =     \E_n \big(   \rho(Y,h_{}(X)) \big[ G_{b,K}^{-1/2} b^K(W) \big]'    \big)               [ \widehat{G}_{b,K}^{o}  ]^{-1} G_{b,K}^{-1/2}   b^K(w).
\end{align}
It follows that \begin{align*}
   &  \E_n \big( \| \widehat{m}(W,h) \|_{\ell^2}^2  \big)   = \sum_{l=1}^{d_{\rho}}  [ \E_n (R_{h, l}^K)      ] ' [ \widehat G_{b,K}^o]^{-1}  [ \E_n (R_{h, l}^K)  ] \\ &  R_{h,l}^K(Z) =     \big[ G_{b,K}^{-1/2}  b^K(W) \big] \rho_{l}(Y,h_{}(X)). 
\end{align*}
By replacing $\widehat{G}_{b,K}^{o}$ with its population analog $I_K$, we define \begin{align}
    \label{mwh2} \widetilde{m}(w,h)  =     \E_n \big(   \rho(Y,h_{}(X)) \big[ G_{b,K}^{-1/2} b^K(W) \big]'    \big)             G_{b,K}^{-1/2}   b^K(w).
\end{align}
Observe that 
\begin{align*}
& \E_n \big( \| \widehat{m}(W,h) - \widetilde{m}(W,h)   \|_{\ell^2}^2  \big)  \leq  \bigg( \sum_{l=1}^{d_{\rho}}  [ \E_n (R_{h, l}^K)      ] '   [ \E_n (R_{h, l}^K)  ] \bigg) \| \big( [ \widehat{G}_{b,K}^{o}  ]^{-1} - I   \big)   \|_{op}^2 \| \widehat{G}_{b,K}^{o}   \|_{op}.
\end{align*}
With $\mathbb{P}$ probability approaching $1$, an application of Lemma \ref{aux2} implies that the first term on the right is bounded above (up to a constant) by $  \E_n \big( \| \widehat{m}(W,h) \|_{\ell^2}^2  \big)$. Similarly, the second term has asymptotic rate $  \sqrt{K} \sqrt{\log K} / \sqrt{n} $ and the third term is bounded above by a constant. By Assumption \ref{fsbasis}$(iii)$, the eigenvalues of $\Sigma(W)$ are bounded above with probability $1$. By Cauchy-Schwarz and the definition of $\Theta_n$, it follows that
\begin{align*}
& \sup_{h \in \Theta_n} \left| \E_n  \big[  \widehat{m}(W,h-h_0) \Sigma(W) \widehat{m}(W,h-h_0) \big] - \E_n \big[  \widetilde{m}(W,h-h_0) \Sigma(W) \widetilde{m}(W,h-h_0) \big] \right| \\ & = O_{\mathbb{P}} \bigg(  \sup_{h \in \Theta_n}  \sqrt{\E_n \|   \widehat{m}(W,h) - \widetilde{m}(W,h)   \|_{\ell^2}^2  } \sqrt{\E_n \|  \widehat{m}(W,h)  \|_{\ell^2}^2} \bigg)  \\ & =  O_{\mathbb{P}} \big(  \|   [ \widehat{G}_{b,K}^{o}  ]^{-1}  - I_K  \|_{op} \big) \log(n) \epsilon_n^2 \\ & = n^{-1} O_{\mathbb{P}}  \bigg(\frac{ \log(n) K \sqrt{K \log K}  }{\sqrt{n}} \bigg).
\end{align*}
Since $ \log(n) K \sqrt{K \log K} / \sqrt{n} = o(1) $, the preceding term is $o_{\mathbb{P}}(n^{-1}).$

Observe that $\Pi_K m(w,h)$ can be expressed as  \begin{align}
\label{pikmwh2} \Pi_K m(w,h) =     \E \big(   \rho(Y,h_{}(X)) \big[ G_{b,K}^{-1/2}  b^K(W) \big]'    \big)               G_{b,K}^{-1/2}  b^K(w).
\end{align}
By Lemma \ref{aux2}, \ref{emp-bloc} and Condition $\ref{misc1}(ii)$, there exists a sequence $r_n$ satisfying $r_n  \sqrt{K_n} \sqrt{\log n} \downarrow 0$ such that \begin{align*}
   \sup_{h \in \Theta_n} \E_n \| \widetilde{m}(W,h-h_0) - \Pi_K m(W,h-h_0)      \|_{\ell^2}^2  &  \leq  \sup_{h \in \Theta_n} \bigg( \sum_{l=1}^{d_{\rho}}  \|  \E_n  (R_{h, l}^K) - \E  (R_{h, l}^K)  \|_{\ell^2}^2       \bigg)  \| \widehat{G}_{b,K}^{o}   \|_{op} \\ & = O_{\mathbb{P}} \big( n^{-1} r_n^2     \big).
\end{align*}
By Cauchy-Schwarz, it follows that \begin{align*}
& \sup_{h \in \Theta_n} \left| \E_n \big[  \widetilde{m}(W,h-h_0) \Sigma(W) \widetilde{m}(W,h-h_0) \big] - \E_n \big[  \Pi_K m(W,h-h_0) \Sigma(W) \Pi_K m(W,h-h_0) \big] \right| \\ & = O_{\mathbb{P}} \bigg(  \sup_{h \in \Theta_n}  \sqrt{\E_n \|   \widetilde{m}(W,h) - \Pi_K m(W,h)   \|_{\ell^2}^2  } \sqrt{\E_n \|  \widetilde{m}(W,h)  \|_{\ell^2}^2} \bigg)  \\ & =  O_{\mathbb{P}} \big( n^{-1/2} r_n \sqrt{\log n} \epsilon_n    \big) \\ & = n^{-1} O_{\mathbb{P}} \big(  r_n \sqrt{\log n} \sqrt{K_n}   \big) \\ & = n^{-1} o_{\mathbb{P}}(1).
\end{align*}
Finally, by Lemma \ref{empsq} and Condition \ref{misc1}$(i)$, we obtain \begin{align*}
 & \sup_{h \in \Theta_n}  \left| \E_n \big \{   \Pi_K m(W,h)' \Sigma(W) \Pi_K m(W,h) \big \}   - \E \big \{  \Pi_K m(W,h)' \Sigma(W) \Pi_K m(W,h)        \big \} \right| \\ &  =  o_{\mathbb{P}}(n^{-1}).
\end{align*}
\item[\textbf{$(ii)$}] We verify that   \begin{align*}
 &  \sup_{h \in \Theta_n}  \left|     \E \big \{  \Pi_K m(W,h)' \Sigma(W) \Pi_K m(W,h)        \big \}  - \E \big (  \Pi_K D_{h_0}[h-h_0] ' \Sigma(W) \Pi_K  D_{h_0}[h-h_0]        \big )                         \right|  = o(n^{-1}).
\end{align*}
Denote the remainder obtained from linearizing the map at $h$ by \begin{align} \label{rem-h-proof}
R_{h_0}(h,W) = m(W,h) - m(W,h_0) -  D_{h_0}[h- h_0] .
\end{align}
Observe that \begin{align*}
   & \E \big \{  \Pi_K m(W,h)' \Sigma(W) \Pi_K m(W,h) \big \}   - \E \big (  \Pi_K D_{h_0}[h-h_0] ' \Sigma(W) \Pi_K  D_{h_0}[h-h_0]        \big )   \\ & = \E \big [  \Pi_K  R_{h_0}(h,W) ' \Sigma(W) \Pi_K  R_{h_0}(h,W) \big ] + 2 \E  \big [  \Pi_K R_{h_0}(h,W)  ' \Sigma(W)  \Pi_K D_{h_0}[h- h_0]    \big ] .
 \end{align*}
Since the eigenvalues of $\Sigma(.)$ are uniformly  bounded above, Cauchy-Schwarz yields \begin{align*}
   & n \sup_{h \in \Theta_n}  \left| \E \big \{  \Pi_K m(W,h)' \Sigma(W) \Pi_K m(W,h) \big \}   - \E \big (  \Pi_K D_{h_0}[h-h_0] ' \Sigma(W) \Pi_K  D_{h_0}[h-h_0]        \big )    \right| \\  & \lessapprox  n  \sup_{h \in \Theta_n} \bigg[ \| \Pi_K R_{h_0}(h,W)  \|_{L^2(\mathbb{P})}^2  + \| \Pi_K R_{h_0}(h,W)  \|_{L^2(\mathbb{P})} \| \Pi_K D_{h_0}[h-h_0]   \|_{L^2(\mathbb{P})}  \bigg] \\ & \lessapprox n  \sup_{h \in \Theta_n}\bigg[   \| \Pi_K R_{h_0}(h,W)  \|_{L^2(\mathbb{P})}^2 + \| \Pi_K R_{h_0}(h,W)  \|_{L^2(\mathbb{P})} \sqrt{\log n} \epsilon_n    \bigg] \\ & = n \sup_{h \in \Theta_n} \bigg[   \| \Pi_K R_{h_0}(h,W)  \|_{L^2(\mathbb{P})}^2 + \| \Pi_K R_{h_0}(h,W)  \|_{L^2(\mathbb{P})} \sqrt{\log n} \sqrt{K}n^{-1/2}   \bigg] .    
\end{align*}
The preceding quantity is $o(1)$ by Condition \ref{misc1}$(iii)$.

\item[\textbf{$(iii)$}]
By repeating the argument from parts $(i-ii)$, we  similarly obtain for every fixed $t \in \R$, the estimate
  \begin{align*}
& \sup_{h \in \Theta_n} \left| \E_n  \big(   \|   \widehat{m}(W,h_t-h_0)   \|_{\widehat{\Sigma}(W)}^2    \big) - \E \big (  \Pi_K D_{h_0}[h_t - h_0]  ' \Sigma(W) \Pi_K D_{h_0}[h_t - h_0]  \big )                  \right| \\ & =o_{\mathbb{P}}(n^{-1}).
\end{align*}

\item[\textbf{$(iv)$}]

Define  \begin{align} \label{smean}  S_n =    \E_n \big[  \langle  \rho(Y,h_0(X)), D_{h_0}[\tilde{\Phi} ](W)        \rangle_{\Sigma(W)}            \big].
\end{align}
For any fixed $t \in \R$, we aim to verify that \begin{align}
 \label{sn-verify} \sup_{h \in \Theta_n} \left| \E_n \big(   \langle  \widehat{m}(W,h_0)   ,  \widehat{m}(W,h-h_t)     \rangle_{\widehat{\Sigma}(W)}     \big) - \frac{t}{\sqrt{n}} S_n \right|   = o_{\mathbb{P}}(n^{-1}).
\end{align}
By a similar argument to parts $(i-ii)$, it is straightforward to verify that \begin{align*}
 & \sup_{h \in \Theta_n} \left| \E_n \big[   \langle  \widehat{m}(W,h_0)   ,  \widehat{m}(W,h-h_t)     \rangle_{\widehat{\Sigma}(W)}     \big] - \E_n \big[   \langle  \widehat{m}(W,h_0)   ,  \widehat{m}(W,h-h_t)     \rangle_{\Sigma(W)}     \big] \right|   = o_{\mathbb{P}}(n^{-1}) \\ & \sup_{h \in \Theta_n} \left| \E_n \big[   \langle  \widehat{m}(W,h_0)   ,  \widehat{m}(W,h-h_t)     \rangle_{\Sigma(W)}     \big] - \E_n \big[   \langle  \widehat{m}(W,h_0)   ,  \Pi_K m(W,h-h_t)     \rangle_{\Sigma(W)}     \big]            \right|  = o_{\mathbb{P}}(n^{-1}).
\end{align*}
By orthogonality of the least squares projection, we can write \begin{align*}
 \E_n \big[   \langle  \widehat{m}(W,h_0)   ,  \Pi_K m(W,h-h_t)     \rangle_{\Sigma(W)}     \big]  &  =  \E_n \big[   \langle  \widehat{m}(W,h_0)   ,  \Sigma(W) \Pi_K m(W,h-h_t)     \rangle_{}     \big] \\ & =   \E_n \big[   \langle  \rho(Y,h_0(X))  ,  \widehat{\Pi}_K \big[ \Sigma(W) \Pi_K m(W,h-h_t)  \big]   \rangle_{}     \big] \; ,
\end{align*}
where $\widehat{\Pi}_K$ is the operator as defined in (\ref{emp-proj-op}). By interchanging $\E_n$ and the inner product, the preceding term can be expressed as an inner product of two vectors in $\R^{d_{\rho}}$. In particular, $\E_n \big[   \langle  \widehat{m}(W,h_0)   ,  \Pi_K m(W,h-h_t)     \rangle_{\Sigma(W)}     \big] = \sum_{i=1}^{d_{\rho}}  V_i $ where
\begin{align*}
V_l =  \E_n \big(   [ \Sigma(W) \Pi_K m(W,h-h_t)  ]_{l} \big[ G_{b,K}^{-1/2}b^K  (W) \big]'      \big) [ \widehat{G}_{b,K}^{o}  ]^{-1}    \frac{1}{n}  \sum_{i=1}^n   G_{b,K}^{-1/2}   b^K(W_i)                 \rho_{l}(Y_i,h_0(X_i)).
\end{align*}
Similarly, we can express $\E_n \big[   \langle  \rho(Y,h_0(X))  , \Pi_K \big[ \Sigma(W) \Pi_K m(W,h-h_t)  \big]   \rangle_{}     \big] $ as $\sum_{i=1}^{d_{\rho}} \widetilde{V}_i $ where
\begin{align*}
  \widetilde{V}_l =   \E \big(   [ \Sigma(W) \Pi_K m(W,h-h_t)  ]_{l} \big[ G_{b,K}^{-1/2}b^K  (W) \big]'      \big)    \frac{1}{n}  \sum_{i=1}^n   G_{b,K}^{-1/2}   b^K(W_i)                 \rho_{l}(Y_i,h_0(X_i)).
\end{align*}
The $\|. \|_{\ell^2}$ norm of the sample average on the right is of order $ \sqrt{K} /  \sqrt{n}  $ (by Lemma \ref{emp-b}). As the eigenvalues of $\Sigma(.)$ are uniformly bounded above, a straightforward application of Lemma \ref{emp-bloc} and Condition \ref{misc1}$(ii)$  implies that \begin{align*}
  \E  \bigg [  \sup_{h \in \Theta_n} \|  (\E_n - \E) \big(   [ \Sigma(W) \Pi_K m(W,h-h_t)  ]_{l} \big[ G_{b,K}^{-1/2}b^K  (W) \big]'      \big)      \|_{\ell^2} \bigg ]  \leq \frac{r_n}{\sqrt{n}}
\end{align*}
for some sequence $r_n $ satisfying $r_n \sqrt{K} \sqrt{\log n} \downarrow 0$. Furthermore, by Lemma \ref{aux2}, we have $ \| [ \widehat{G}_{b,K}^{o}  ]^{-1}  -I_K  \|_{op} \leq C  \sqrt{K \log(K)} /\sqrt{n}       $ with $\mathbb{P}$ probability approaching $1$. From combining the preceding bounds and an application of Cauchy-Schwarz, we obtain
\begin{align*}
& \sup_{h \in \Theta_n}  \left| \E_n \big[   \langle  \widehat{m}(W,h_0)   ,  \Pi_K m(W,h-h_t)     \rangle_{\Sigma(W)}     \big]  - \E_n \big[   \langle  \rho(Y,h_0(X))  , \Pi_K \big[ \Sigma(W) \Pi_K m(W,h-h_t)  \big]   \rangle_{}     \big]  \right| \\ & = o_{\mathbb{P}}(n^{-1}).
\end{align*}
Next, write $m(W,h) = R_{h_0}(h,W) +D_{h_0}[h-h_0]$ to obtain \begin{align*}
& \E_n \big[   \langle  \rho(Y,h_0(X))  , \Pi_K \big[ \Sigma(W) \Pi_K m(W,h-h_t)  \big]   \rangle_{}     \big] \\ & = \E_n \big[   \langle  \rho(Y,h_0(X))  , \Pi_K \big[ \Sigma(W) \Pi_K R_{h_0}(h,W)  \big]   \rangle_{}     \big] - \E_n \big[   \langle  \rho(Y,h_0(X))  , \Pi_K \big[ \Sigma(W) \Pi_K R_{h_0}(h_t,W)  \big]   \rangle_{}     \big] \\ & + \E_n \big[   \langle  \rho(Y,h_0(X))  , \Pi_K \big[ \Sigma(W) \Pi_K D_{h_0}[h- h_t]   \big]   \rangle_{}     \big].
\end{align*}
By interchanging $\E_n$ and the inner product as above, the first two terms on the right side of the equality can be analyzed through the terms \begin{align*}
&  Q_{i,1} =  \E \big(   [ \Sigma(W) \Pi_K R_{h_0} (h,W)  ]_{l} \big[ G_{b,K}^{-1/2}b^K  (W) \big]'      \big)    \frac{1}{n}  \sum_{i=1}^n   G_{b,K}^{-1/2}   b^K(W_i)                 \rho_{l}(Y_i,h_0(X_i)) \; , \\ & Q_{i,2} =      -  \E \big(   [ \Sigma(W) \Pi_K R_{h_0} (h_t,W)  ]_{l} \big[ G_{b,K}^{-1/2}b^K  (W) \big]'      \big)    \frac{1}{n}  \sum_{i=1}^n   G_{b,K}^{-1/2}   b^K(W_i)                 \rho_{l}(Y_i,h_0(X_i)).
\end{align*}
The $\| . \|_{\ell^2}$ norm of the sample average on the right of both the preceding terms is of order $ \sqrt{K} /  \sqrt{n}  $ (by Lemma \ref{emp-b}). Furthermore, by the Bessel inequality, we obtain
\begin{align*}
  & \|  \E \big(   [ \Sigma(W) \Pi_K R_{h_0} (h,W)  ]_{l} \big[ G_{b,K}^{-1/2}b^K  (W) \big] \|_{\ell^2}^2 \leq  \|  [ \Sigma(W) \Pi_K R_{h_0} (h,W)  ]_{l}   \|_{L^2(\mathbb{P})}^2 \;, \\ & \|  \E \big(   [ \Sigma(W) \Pi_K R_{h_0} (h_t,W)  ]_{l} \big[ G_{b,K}^{-1/2}b^K  (W) \big] \|_{\ell^2}^2 \leq  \|  [ \Sigma(W) \Pi_K R_{h_0} (h_t,W)  ]_{l}   \|_{L^2(\mathbb{P})}^2.
\end{align*}
As the eigenvalues of $\Sigma(.)$ are uniformly bounded above, the preceding bounds imply the expansion \begin{align*}
   & \E_n \big[   \langle  \rho(Y,h_0(X))  , \Pi_K \big[ \Sigma(W) \Pi_K m(W,h-h_t)  \big]   \rangle_{}     \big] \\ & =   \E_n \big[   \langle  \rho(Y,h_0(X))  , \Pi_K \big[ \Sigma(W) \Pi_K D_{h_0}[h- h_t]   \big]   \rangle_{}     \big]    \\ &   + \frac{\sqrt{K}}{\sqrt{n}} O_{\mathbb{P}} \bigg(  \sup_{h \in \Theta_n}   \| \Pi_K R_{h_0}(h,W)  \|_{L^2(\mathbb{P})}  + \sup_{h \in \Theta_n}  \| \Pi_K R_{h_0}(h_t,W)  \|_{L^2(\mathbb{P})}     \bigg)
\end{align*}
uniformly over $h \in \Theta_n$. Hence, by Condition \ref{misc1}$(iii)$, it follows that \begin{align*}
    & \E_n \big[   \langle  \rho(Y,h_0(X))  , \Pi_K \big[ \Sigma(W) \Pi_K m(W,h-h_t)  \big]   \rangle_{}     \big] \\ &  =  \E_n \big[   \langle  \rho(Y,h_0(X))  , \Pi_K \big[ \Sigma(W) \Pi_K D_{h_0}[h- h_t]   \big]   \rangle_{}     \big] + o_{\mathbb{P}}(n^{-1})
\end{align*}
uniformly over $h \in \Theta_n$. By construction $h - h_t =  t \tilde{\Phi} / \sqrt{n}  $ and $D_{h_0}(.)$ is a linear operator. It follows that the preceding term can be expressed as \begin{align*}
\E_n \big[   \langle  \rho(Y,h_0(X))  , \Pi_K \big[ \Sigma(W) \Pi_K D_{h_0}[h- h_t]   \big]   \rangle_{}     \big] = \ \frac{t}{\sqrt{n}} \E_n \big[   \langle  \rho(Y,h_0(X))  , \Pi_K \big[ \Sigma(W) \Pi_K D_{h_0}[ \tilde{\Phi}  ]   \big]   \rangle_{}     \big] .
\end{align*}
Hence, to show (\ref{sn-verify}), it suffices to verify that \begin{align} \label{sn-verify-new}
     \E_n \big[   \langle  \rho(Y,h_0(X))  , \Pi_K \big[ \Sigma(W) \Pi_K D_{h_0}[ \tilde{\Phi}  ]   \big]   \rangle_{}     \big] =    \E_n \big[   \langle  \rho(Y,h_0(X))  ,   \Sigma(W)  D_{h_0}[ \tilde{\Phi}  ]     \rangle_{}     \big] + o_{\mathbb{P}}(n^{-1/2}).
\end{align}
Observe that the sample mean is over a mean zero random variable, since $\E[\rho(Y,h_0(X))|W] = m(W,h_0) =  0$. Furthermore, since $\E\big( \| \rho(Y,h_0(X)) \|_{\ell^2}^2|W )$ is bounded above (with $\mathbb{P}$ probability $1$), we have that \begin{align*}
  & n \E  \left|   \E_n \big[ \langle  \rho(Y,h_0(X))  , (\Pi_K -I)\big[ \Sigma(W) \Pi_K D_{h_0}[\tilde{\Phi}] (W)  \big]   \rangle  ]  \right|^2 \\ & = \E \bigg( \left|  \langle  \rho(Y,h_0(X))  , (\Pi_K -I)\big[ \Sigma(W) \Pi_K D_{h_0}[\tilde{\Phi}] (W)  \big]   \rangle       \right|^2    \bigg) \\ & \rightarrow 0.
\end{align*}
Similarly, we obtain \begin{align*}
  & n \E  \left|   \E_n \big[  \langle  \rho(Y,h_0(X))  , \big[ \Sigma(W) (\Pi_K-I) D_{h_0}[\tilde{\Phi}] (W)  \big]   \rangle  ]  \right|^2 \\ & = \E \bigg( \left|   \langle  \rho(Y,h_0(X))  , \big[ \Sigma(W) (\Pi_K-I) D_{h_0}[\tilde{\Phi}] (W)  \big]   \rangle    \right|^2    \bigg) \\  & \rightarrow 0.
\end{align*}
The expression in $(\ref{sn-verify-new})$ follows by Markov's inequality.

\item[\textbf{$(v)$}]

The preceding steps $(i-iv)$ show that  \begin{align*}
& \E_n  \big(   \|   \widehat{m}(W,h)   \|_{\widehat{\Sigma}(W)}^2    \big) - \E_n  \big(    \|   \widehat{m}(W,h_t)   \|_{\widehat{\Sigma}(W)}^2    \big) \\ & =   \E_n \big(  \| \widehat{m}(W,h-h_0)     \|_{\widehat{\Sigma}(W)}^2         \big) - \E_n \big(  \| \widehat{m}(W,h_t-h_0)     \|_{\widehat{\Sigma}(W)}^2         \big) + 2 \E_n \big[   \langle  \widehat{m}(W,h_0)   ,  \widehat{m}(W,h-h_t)     \rangle_{\widehat{\Sigma}(W)}      \big] \\ & = \E \big( \| \Pi_K D_{h_0}[h-h_0]   \|_{\Sigma(W)}^2      \big) - \E \big( \| \Pi_K D_{h_0}[h_t-h_0]   \|_{\Sigma(W)}^2      \big) +2 \frac{t}{\sqrt{n}} S_n + o_{\mathbb{P}}(n^{-1})
\end{align*}
uniformly over $h \in \Theta_n$, where $S_n$ is as in (\ref{smean}). Furthermore, since $D_{h_0}(.)$ is a linear operator, we obtain \begin{align*}
& \frac{n}{2} \bigg[ \E \big( \| \Pi_K D_{h_0}[h-h_0]   \|_{\Sigma(W)}^2      \big) - \E \big( \| \Pi_K D_{h_0}[h_t-h_0]   \|_{\Sigma(W)}^2      \big) \bigg] \\ & = - \frac{t^2}{2 } \E \big( \|   \Pi_K  D_{h_0}[\tilde{\Phi}]     \|_{\Sigma(W)}^2   \big)  + t \sqrt{n} \E \big[  \langle  \Pi_K  D_{h_0}[  h - h_0   ],  \Pi_K D_{h_0}[\tilde{\Phi}]           \rangle_{\Sigma(W)}        \big] .
\end{align*}
For the first term, continuity yields $$  - \frac{t^2}{2 } \E \big( \|   \Pi_K  D_{h_0}[\tilde{\Phi}]     \|_{\Sigma(W)}^2   \big) = - \frac{t^2}{2 } \E \big( \|    D_{h_0}[\tilde{\Phi}]     \|_{\Sigma(W)}^2   \big) + o(1). $$
For the second term, we expand it as \begin{align*}
& \E \big[  \langle  \Pi_K  D_{h_0}[  h - h_0   ],  \Pi_K D_{h_0}[\tilde{\Phi}]           \rangle_{\Sigma(W)}        \big]  \\ & = \E \big[  \langle  \Pi_K  D_{h_0}[  h - h_0   ],   \Pi_K \big\{ \Sigma(W)  \Pi_K D_{h_0}[\tilde{\Phi}] \big \}          \rangle        \big]  \\ & = \E \big[  \langle  \Pi_K  D_{h_0}[  h - h_0   ],   \Pi_K \big\{ \Sigma(W)  (\Pi_K-I) D_{h_0}[\tilde{\Phi}] \big \}          \rangle        \big]  +  \E \big[  \langle  \Pi_K  D_{h_0}[  h - h_0   ],   \Pi_K \big\{ \Sigma(W)   D_{h_0}[\tilde{\Phi}] \big \}          \rangle        \big] .
\end{align*}
Since the eigenvalues of $\Sigma(.)$ are uniformly bounded above, Cauchy-Schwarz yields \begin{align*}
   &  \sup_{h \in \Theta_n}  \sqrt{n} \left| \E \big[  \langle  \Pi_K  D_{h_0}[  h - h_0   ],   \Pi_K \big\{ \Sigma(W)  (\Pi_K - I) D_{h_0}[\tilde{\Phi}] \big \}          \rangle        \big]  \right| \\ & \lessapprox \sqrt{n} \epsilon_n \sqrt{\log n} \| (\Pi_K - I) D_{h_0}[\tilde{\Phi}]   \|_{L^2(\mathbb{P})} \\ & = \sqrt{K} \sqrt{\log n}  \| (\Pi_K - I) D_{h_0}[\tilde{\Phi}]   \|_{L^2(\mathbb{P})} \\ & = o(1).
\end{align*}
Next, by orthogonality we have that \begin{align*}
    & \E \big[  \langle  \Pi_K  D_{h_0}[  h - h_0   ],   \Pi_K \big\{ \Sigma(W)   D_{h_0}[\tilde{\Phi}] \big \}          \rangle        \big] \\ & = \E \big[  \langle    D_{h_0}[  h - h_0   ],     \Sigma(W)   D_{h_0}[\tilde{\Phi}]           \rangle        \big] + \E \big[  \langle  (\Pi_K-I)  D_{h_0}[  h - h_0   ],   (\Pi_K-I) \big\{ \Sigma(W)   D_{h_0}[\tilde{\Phi}] \big \}          \rangle        \big].
\end{align*}
By Cauchy-Schwarz, we obtain \begin{align*}
  & \sup_{h \in \Theta_n}  \sqrt{n} \left|   \E \big[  \langle  (\Pi_K-I)  D_{h_0}[  h - h_0   ],   (\Pi_K-I) \big\{ \Sigma(W)   D_{h_0}[\tilde{\Phi}] \big \}          \rangle        \big]    \right| \\ & \lessapprox  \sqrt{n} \epsilon_n \sqrt{\log n} \| (\Pi_K - I) \Sigma(W) D_{h_0}[\tilde{\Phi}]    \|_{L^2(\mathbb{P})}  \\ & = \sqrt{K} \sqrt{\log n}  \| (\Pi_K - I) \Sigma(W) D_{h_0}[\tilde{\Phi}]    \|_{L^2(\mathbb{P})}  \\ & = o(1).
\end{align*}
From combining the preceding bounds, we obtain the expansion
\begin{align*}
&  \frac{-n}{2} \bigg[ \E_n  \big(   \|   \widehat{m}(W,h)   \|_{\widehat{\Sigma}(W)}^2    \big) - \E_n  \big(    \|   \widehat{m}(W,h_t)   \|_{\widehat{\Sigma}(W)}^2    \big) \bigg] \\ & =    \frac{t^2}{2 } \E \big( \|    D_{h_0}[\tilde{\Phi}]     \|_{\Sigma(W)}^2   \big) - t \sqrt{n} \E \big[  \langle    D_{h_0}[  h - h_0   ],        D_{h_0}[\tilde{\Phi}]           \rangle_{\Sigma(W)}        \big] - t \sqrt{n} S_n +  o_{\mathbb{P}}(1)
\end{align*}
uniformly over $h \in \Theta_n$. Furthermore, by definition of the adjoint $D_{h_0}^*$ and Condition \ref{misc2}$(i)$, we can write \begin{align*}
    t \sqrt{n} \E  \big[ \langle    D_{h_0}[  h - h_0   ],        D_{h_0}[\tilde{\Phi}]           \rangle_{\Sigma(W)}        \big] &= t \sqrt{n}  \langle      h - h_0   ,       D_{h_0}^*  D_{h_0}[\tilde{\Phi}]           \rangle_{L^2(\mathbb{P})}     \\ & = t \sqrt{n}   \langle      h - h_0   ,       \Phi          \rangle_{L^2(\mathbb{P})} .
\end{align*}

\item[\textbf{$(vi)$}]
We compute the Laplace transform of the random variable $\sqrt{n} \big[ \langle h  - h_0 , \Phi \rangle_{L^2(\mathbb{P})} + S_n \big]$ where $h \sim \mu^*(.|\alpha,K,\mathcal{Z}_n)$. Fix any $t \in \R$. From the conclusion of part $(v)$, the Laplace transform is \begin{align*}
&  \E^*  \bigg[ \exp \bigg \{t  \sqrt{n} \big[ \langle  h - h_0 , \Phi       \rangle_{L^2(\mathbb{P})} + S_n  \big] \bigg \}  \bigg| \alpha,K,\mathcal{Z}_n   \bigg] \\ &  = \frac{ \int_{\Theta_n} \exp \bigg \{t  \sqrt{n} \big[ \langle  h - h_0 , \Phi       \rangle_{L^2(\mathbb{P})} + S_n   \big] \bigg \}     \exp \bigg \{   -\frac{n}{2} \bigg[ \E_n  \big(   \|   \widehat{m}(W,h)   \|_{\widehat{\Sigma}(W)}^2    \big)  - \E_n  \big(   \|   \widehat{m}(W,h_t)   \|_{\widehat{\Sigma}(W)}^2    \big)      \bigg]    \bigg \} }{\int_{\Theta_n}  \exp \big(  -\frac{n}{2} \E_n  \big(   \|   \widehat{m}(W,h)   \|_{\widehat{\Sigma}(W)}^2    \big)    \big) d \mu (h|\alpha,K) } \\ & \times \exp \bigg\{   -\frac{n}{2}    \E_n  \big(   \|   \widehat{m}(W,h_t)   \|_{\widehat{\Sigma}(W)}^2    \big)            \bigg\} d \mu(h|\alpha,K) \\ & = \exp \bigg[ \frac{t^2}{2}         \E \big[ (D_{h_0} \tilde{\Phi})' \Sigma(W) (D_{h_0} \tilde{\Phi})      \big]     + o_{\mathbb{P}}(1)     \bigg] \times \frac{\int_{\Theta_n}  \exp \big(  -\frac{n}{2} \E_n  \big(   \|   \widehat{m}(W,h_t)   \|_{\widehat{\Sigma}(W)}^2    \big)    \big) d \mu (h|\alpha,K)}{\int_{\Theta_n}  \exp \big(  -\frac{n}{2} \E_n  \big(   \|   \widehat{m}(W,h)   \|_{\widehat{\Sigma}(W)}^2    \big)    \big) d \mu (h|\alpha,K)}.
\end{align*}
Next, we verify that  \begin{align*}
    \frac{\int_{\Theta_n}  \exp \big(  -\frac{n}{2} \E_n  \big(   \|   \widehat{m}(W,h_t)   \|_{\widehat{\Sigma}(W)}^2    \big)    \big) d \mu (h|\alpha,K)}{\int_{\Theta_n}  \exp \big(  -\frac{n}{2} \E_n  \big(   \|   \widehat{m}(W,h)   \|_{\widehat{\Sigma}(W)}^2    \big)    \big) d \mu (h|\alpha,K)} \xrightarrow{\mathbb{P}} 1.
\end{align*}
Let $ \mu_{t,\tilde{\Phi}}(h|\alpha,K) $ denote the measure obtained from  translating $  \mu(.|\alpha,K) $ around $t \tilde{\Phi} / \sqrt{n}$. To be specific, it is the measure obtained via $$ \mu_{t,\tilde{\Phi}}(h|\alpha,K) \sim  \frac{G_{\alpha}}{\sqrt{K} \sqrt{\log n}} - \frac{t}{\sqrt{n}} \tilde{\Phi}.   $$

Since $\tilde{\Phi} \in \mathbb{H}$, it follows from \citep[Proposition I.20]{ghosal2017fundamentals} that $\mu_{t,\tilde{\Phi}}(h|\alpha,K)$ is absolutely continuous with respect to $\mu(.|\alpha,K)$ and admits a density function \begin{align} \label{gauss-cov}  \frac{d \mu_{t,\tilde{\Phi}}(h|\alpha,K)}{d \mu(h|\alpha,K)}  = \exp \bigg \{ \frac{t}{\sqrt{n}}  \langle  h , \tilde{\Phi}   \rangle_{\mathbb{H}_n}   - \frac{t^2}{2n} \|  \tilde{\Phi}   \|_{\mathbb{H}_n}^2    \bigg \} . \end{align}
From the definition of $\Theta_n$, we have   \begin{align*}
    \sup_{h \in \Theta_n} \left| \frac{t}{\sqrt{n}}  \langle  h , \tilde{\Phi}   \rangle_{\mathbb{H}_n} \right| \leq M \frac{t}{\sqrt{n}} \sqrt{n} \sqrt{\log n} \epsilon_n \|  \tilde{\Phi}  \|_{\mathbb{H}_n}  & = M t \sqrt{\log n} \epsilon_n \|  \tilde{\Phi}  \|_{\mathbb{H}_n} = M t \epsilon_n \log(n) \sqrt{K} \|  \tilde{\Phi}  \|_{\mathbb{H}} \; , 
\end{align*}
where we used the fact that $  \|  \tilde{\Phi}  \|_{\mathbb{H}_n} = \sqrt{K \log(n)}  \|  \tilde{\Phi}  \|_{\mathbb{H}}$. It follows that  \begin{align*}
   &  \sup_{h \in \Theta_n} \left| \frac{t}{\sqrt{n}}  \langle  h , \tilde{\Phi}   \rangle_{\mathbb{H}_n} \right| \lessapprox \frac{K \log(n)}{\sqrt{n}} = o(1) \; , \\ &  \frac{t^2}{2n} \|  \tilde{\Phi}  \|_{\mathbb{H}_n}^2 \lessapprox  \frac{K \log(n)}{\sqrt{n}} = o(1).
\end{align*}
Define$$ \Theta_{n,\tilde{\Phi}} = \Theta_n - \frac{t}{\sqrt{n}} \tilde{\Phi} = \bigg \{ g : g= h - \frac{t}{\sqrt{n}} \tilde{\Phi} \; , h \in \Theta_n      \bigg \} .  $$
By   Gaussian change of variables in (\ref{gauss-cov}) and the preceding bounds, we obtain
 \begin{align*}
    &   \frac{\int_{\Theta_n}  \exp \big(  -\frac{n}{2} \E_n  \big(   \|   \widehat{m}(W,h_t)   \|_{\widehat{\Sigma}(W)}^2    \big)    \big) d \mu (h|\alpha,K)}{\int_{\Theta_n}  \exp \big(  -\frac{n}{2} \E_n  \big(   \|   \widehat{m}(W,h)   \|_{\widehat{\Sigma}(W)}^2    \big)    \big) d \mu (h|\alpha,K)}   = e^{o(1)} \frac{\mu(\Theta_{n,\tilde{\Phi}}|\alpha,K,\mathcal{Z}_n)}{\mu(\Theta_n|\alpha,K,\mathcal{Z}_n)}.
\end{align*}
Since $\mu(\Theta_n^c|\alpha,K,\mathcal{Z}_n) \xrightarrow{\mathbb{P}} 0$, the preceding expression reduces to \begin{align*}
    &   \frac{\int_{\Theta_n}  \exp \big(  -\frac{n}{2} \E_n  \big(   \|   \widehat{m}(W,h_t)   \|_{\widehat{\Sigma}(W)}^2    \big)    \big) d \mu (h|\alpha,K)}{\int_{\Theta_n}  \exp \big(  -\frac{n}{2} \E_n  \big(   \|   \widehat{m}(W,h)   \|_{\widehat{\Sigma}(W)}^2    \big)    \big) d \mu (h|\alpha,K)}   = e^{o(1)} \frac{\mu(\Theta_{n,\tilde{\Phi}}|\alpha,K,\mathcal{Z}_n)}{1+ o_{\mathbb{P}}(1)}.
\end{align*}
By replacing $D,M$ in the definition of $\Theta_n$ in (\ref{thetan}) with a larger $D',M'$ if necessary, it is straightforward to verify that $\mu(\Theta_{n,\tilde{\Phi}}|\alpha,K,\mathcal{Z}_n) \xrightarrow{\mathbb{P}} 1$. From combining the preceding bounds, we obtain \begin{align}
    \label{laplace-bound} & \E^*  \bigg[ \exp \bigg \{t  \sqrt{n} \big[ \langle  h - h_0 , \Phi       \rangle_{L^2(\mathbb{P})} + S_n  \big] \bigg \}  \bigg| \alpha,K,\mathcal{Z}_n   \bigg] \\ & \nonumber  =  \exp \bigg[ \frac{t^2}{2}         \E \big[ (D_{h_0} \tilde{\Phi})' \Sigma(W) (D_{h_0} \tilde{\Phi})      \big]          \bigg] [1 + o_{\mathbb{P}}(1)].
\end{align}
Since this is true for every $t \in \R$, it follows from \citep[Lemma 1]{castillo2015bernstein}  and (\ref{tvconv}) that \begin{align} \label{weak-conv-1}
     \sqrt{n} \big( \langle  h - h_0 , \Phi       \rangle_{L^2(\mathbb{P})} + S_n  \big)\big| \mathcal{Z}_n \overset{\mathbb{P}}{\rightsquigarrow}  N \big( 0 ,  \E \big[ (D_{h_0} \tilde{\Phi})' \Sigma (D_{h_0} \tilde{\Phi})      \big]             \big) . 
\end{align}

\item[\textbf{$(vii)$}]
Recall that \begin{align} \label{smean2}  S_n =    \E_n \big[  \langle  \rho(Y,h_0(X)), D_{h_0}[\tilde{\Phi} ](W)        \rangle_{\Sigma(W)}            \big].
\end{align}
Since $S_n$ is the sample mean of a mean zero  random variable with finite variance, note that $ n \E[ S_n^2] = O(1)$. From (\ref{weak-conv-1}) and Lemma \ref{posmean-conv}, it follows from a uniform integrability in probability argument (see e.g. \cite{monard2021statistical}) that \begin{align}
    \label{weak-conv-mean} \langle \E \big[  h|\alpha,K,\mathcal{Z}_n  \big] , \Phi \rangle_{L^2(\mathbb{P})} =  \langle h_0 , \Phi \rangle_{L^2(\mathbb{P})} - S_n + o_{\mathbb{P}}(n^{-1/2}).
\end{align}
The first implication of this is that by substituting this identity back into (\ref{weak-conv-1}), we obtain \begin{align*}
\sqrt{n} \langle  h - \E \big[  h|\alpha,K,\mathcal{Z}_n  \big] , \Phi       \rangle \big| \mathcal{Z}_n  \overset{\mathbb{P}}{\rightsquigarrow} N \big( 0 ,  \E \big[ (D_{h_0} \tilde{\Phi})' \Sigma (D_{h_0} \tilde{\Phi})      \big]             \big) .  
\end{align*}
The second implication is that  $  \sqrt{n}\langle \E \big[  h|\alpha,K,\mathcal{Z}_n  \big] - h_0 , \Phi \rangle_{L^2(\mathbb{P})} $ is asymptotically equivalent to $ - \sqrt{n} S_n $. Hence, by the central limit theorem, we obtain \begin{align*}
    \sqrt{n} \langle  h_0 - \E \big[  h|\alpha,K,\mathcal{Z}_n  \big] , \Phi       \rangle  =   \sqrt{n} S_n + o_{\mathbb{P}}(1) \rightsquigarrow &   N(0, \E \big[ (D_{h_0} \tilde{\Phi})' \Sigma  \rho \rho'  \Sigma  (D_{h_0} \tilde{\Phi}) \big]   ) .
\end{align*}
The claim follows.

\end{enumerate}

\end{proof}

\begin{lemma}
\label{posmean-conv} Suppose the hypothesis of Theorem \ref{bvm} holds. Then \begin{align*} 
    n \E  \bigg[ \left|  \langle h - h_0 , \Phi    \rangle_{L^2(\mathbb{P})} \right|^2 \bigg| \alpha,K,\mathcal{Z}_n  \bigg]  = O_{\mathbb{P}}(1).
\end{align*}
\end{lemma}

\begin{proof}[Proof of Lemma \ref{posmean-conv}] 
Let $C$ denote a generic universal constant that may change from line to line. Define the sequences \begin{align}
\label{seqs}  & \epsilon_{n}^{} =    \frac{  \sqrt{ K_n}}{\sqrt{n}}     \; \; , \; \; \delta_{n} =  \begin{cases}     n^{- \frac{\alpha}{2[\alpha + \zeta] + d}} \sqrt{\log n}     & \text{mildly ill-posed}  \\  ( \log n)^{- \alpha/ \zeta} \sqrt{\log \log n}   & \text{severely ill-posed}.    \end{cases}  
\end{align}
First, we state a few preliminary observations from the proof of Theorem \ref{rate}. There exists a universal constant $c > 0$ such that \begin{align} \label{lb-newproof}   \int  \exp\bigg(    - \frac{n}{2}  \E_n \big[    \widehat{m}(W,h) ' \widehat{\Sigma}(W)  \widehat{m}(W,h)            \big]       \bigg) d  \mu (h|\alpha,K)  \geq   \exp \big(   - c n  \log (n)  \epsilon_{n}^2   \big)   \end{align}
holds with $\mathbb{P}$ probability approaching $1$. Furthermore, for every $E' > 0$, there exists a sufficiently large $E$ (which depends on $E'$) such that \begin{align}
    \label{elarge} \mu \big( \| h- h_0 \|_{L^2(\mathbb{P})} \leq E \delta_n   \big| \alpha,K,\mathcal{Z}_n  \big) \geq 1- \exp(-E' n \log(n) \epsilon_n^2)
\end{align}
holds with $\mathbb{P}$ probability approaching $1$.

Fix any $E ' > c$ and let $E$ be as specified above. Write  \begin{align*}
  & \E  \bigg[ \left|  \langle h - h_0 , \Phi    \rangle_{L^2(\mathbb{P})} \right|^2 \bigg| \mathcal{Z}_n  \bigg]  \\ & = \E  \bigg[ \left|  \langle h - h_0 , \Phi    \rangle_{L^2(\mathbb{P})} \right|^2 \mathbbm{1} \{  \| h - h_0  \|_{L^2(\mathbb{P})} \leq E \delta_n  \} \bigg| \mathcal{Z}_n  \bigg]      \\ & + \E  \bigg[ \left|  \langle h - h_0 , \Phi    \rangle_{L^2(\mathbb{P})} \right|^2 \mathbbm{1} \{  \| h - h_0  \|_{L^2(\mathbb{P})} > E \delta_n  \} \bigg| \mathcal{Z}_n  \bigg]   \\ & = A_1  + A_2.
\end{align*}
For $A_2$, Cauchy-Schwarz yields \begin{align*}
    A_2^2 \leq \bigg ( \E  \bigg[ \left|  \langle h - h_0 , \Phi    \rangle_{L^2(\mathbb{P})} \right|^4  \bigg| \alpha,K,\mathcal{Z}_n  \bigg]   \bigg )   \mu \big( \| h- h_0 \|_{L^2(\mathbb{P})} > E \delta_n   \big| \alpha,K, \mathcal{Z}_n  \big)   .
\end{align*}
From $(\ref{lb-newproof})$, we obtain \begin{align*}
    & \E  \bigg[ \left|  \langle h - h_0 , \Phi    \rangle_{L^2(\mathbb{P})} \right|^4  \bigg| \mathcal{Z}_n  \bigg]  \\ &  = \frac{\int  \left|  \langle h - h_0 , \Phi    \rangle_{L^2(\mathbb{P})} \right|^4    \exp\bigg(    - \frac{n}{2}  \E_n \big[    \widehat{m}(W,h) ' \widehat{\Sigma}(W)  \widehat{m}(W,h)            \big]       \bigg) d  \mu (h|\alpha,K) }{\int \exp\bigg(    - \frac{n}{2}  \E_n \big[    \widehat{m}(W,h) ' \widehat{\Sigma}(W)  \widehat{m}(W,h)            \big]       \bigg) d  \mu (h|\alpha,K)  } \\ & \leq \exp(c n \log(n) \epsilon_n^2) \int  \left|  \langle h - h_0 , \Phi    \rangle_{L^2(\mathbb{P})} \right|^4    \exp\bigg(    - \frac{n}{2}  \E_n \big[    \widehat{m}(W,h) ' \widehat{\Sigma}(W)  \widehat{m}(W,h)            \big]       \bigg) d  \mu (h|\alpha,K) \\ & \leq \exp(c n \log(n) \epsilon_n^2) \|  \Phi \|_{L^2(\mathbb{P})}^4  \int  \| h -h_0 \|_{L^2(\mathbb{P})}^4  d  \mu (h|\alpha,K) \\ & \leq C \exp(cn \log(n) \epsilon_n^2).
\end{align*}
Hence, by (\ref{elarge}) it follows that \begin{align*}
    A_2^2 \leq C \exp(  (c- E')     n \log(n) \epsilon_n^2  ) .
\end{align*}
Since $E' > c$, it follows that $n A_2 = o_{\mathbb{P}}(1) $.

Let $\Theta_n$ be defined as (\ref{thetan}). In particular, by $(\ref{theta-n-bound})$, we have $\mu(\Theta_n^c|\alpha,K,\mathcal{Z}_n)  \leq e^{-R n \epsilon_n^2 \log n} $ for some universal constant $R > 0$. We denote by $\E^*(.|\alpha,K,\mathcal{Z}_n)$, the expectation with respect to the localized (to $\Theta_n$) posterior measure $\mu^*(.|\alpha,K,\mathcal{Z}_n)$  in $(\ref{mu-star})$. It follows that $A_1$ can be expressed as \begin{align*}
    A_1 & =  \E^*  \bigg[ \left|  \langle h - h_0 , \Phi    \rangle_{L^2(\mathbb{P})} \right|^2 \mathbbm{1} \{  \| h - h_0  \|_{L^2(\mathbb{P})} \leq E \delta_n  \} \bigg| \mathcal{Z}_n  \bigg]  \\ & +  \int  \left|  \langle h - h_0 , \Phi    \rangle_{L^2(\mathbb{P})} \right|^2 \mathbbm{1} \{  \| h - h_0  
 \|_{L^2(\mathbb{P})} \leq E \delta_n  \} d \big[ \mu(h|\alpha,K,\mathcal{Z}_n) - \mu^*(h|\alpha,K,\mathcal{Z}_n)   ] \\ & = A_{1,1} + A_{1,2}.
\end{align*}
From the general bound $x^2 \leq 2 \big( e^x + e^{-x} ) $ for every $x \in \R$, it follows from (\ref{laplace-bound}) with $t= \pm 1$ that \begin{align*}
    n A_{1,1} \leq C \big( e^{\sqrt{n} S_n} + e^{- \sqrt{n} S_n}      \big) \;,
\end{align*}
with $\mathbb{P}$ probability approaching $1$, where $S_n$ is defined as in (\ref{smean}). Since $S_n$ is a sample mean of a mean zero random variable with finite variance, the central limit theorem implies $nA_{1,1} = O_{\mathbb{P}}(1)$.

For $A_{1,2}$, if $\| . \|_{TV}$ denotes the total variation metric, we have that \begin{align*}
    A_{1,2} \leq E^2 \delta_n^2 \| \Phi  \|_{L^2(\mathbb{P})}^2 \| \mu - \mu^*   \|_{TV}  \leq E^2 \delta_n^2 2 \mu(\Theta_n^c|\alpha,K,\mathcal{Z}_n) \leq  C \delta_n^2 e^{-R n \epsilon_n^2 \log(n)}.
\end{align*}
It follows that $n A_{1,2} = o_{\mathbb{P}}(1)$.
\end{proof}

\begin{proof}[Proof of Corollary \ref{posgp}]
Let $\delta_n$ denote the stated contraction rate and $\epsilon_n = \sqrt{K_n} / \sqrt{n}$. From the proof of Theorem \ref{rate}, there exists a universal constant $D > 0$ such that for all sufficiently large $L > 0$, we have
\begin{align*}
    \mu \big(  \| h -h_0 \|_{L^2} > L \delta_n  \: \big| \: \mathcal{Z}_n     \big) \leq \exp(-D L \log(n) n \epsilon_n^2).
\end{align*}
with $\mathbb{P}$ probability approaching $1$. Suppose the preceding bound holds for all $ L \geq \overline{L} > 0$. We have that \begin{align*}
     & \|  h_0 -  \E \big[ h|\mathcal{Z}_n \big]  \|_{L^2}^2 \\ &   \leq \E \big(  \| h - h_0  \|_{L^2}^2 \big| \mathcal{Z}_n        ) \\ & = \int_{\|  h - h_0  \|_{L^2} < \overline{L} \delta_n }  \| h - h_0  \|_{L^2}^2 d \mu(h|\mathcal{Z}_n)   + \sum_{j=1}^{\infty}  \int \limits_{ j \overline{L} \delta_n \leq  \|  h - h_0  \|_{L^2}  < (j+1) \overline{L} \delta_n }  \| h - h_0  \|_{L^2}^2 d \mu(h|\mathcal{Z}_n)  \\ & \leq \overline{L}^2 \delta_n^2 +   \overline{L}^2 \delta_n^2 \sum_{j=1}^{\infty} (j+1)^2 \exp(-D j \overline{L} \log(n) n \epsilon_n^2   ) .
\end{align*}
Since the preceding sum is finite, the claim follows.
\end{proof}

\begin{proof}[Proof of Corollary \ref{bvm-col}] 
The set $C_n(\gamma)$ can equivalently be expressed as \begin{align*}
   & C_n(\gamma) = \{  t \in \R : \sqrt{n} \left|  t - \mathbf{L}\big( \E [h| \mathcal{Z}_n  ] \big)    \right|  \leq c_{1- \gamma}   \}\:, \\ & c_{1- \gamma} = (1- \gamma) \; \:  \text{quantile of} \:   \;  \sqrt{n} \left| \mathbf{L}(h) - \mathbf{L}\big( \E [h| \mathcal{Z}_n ]   \big)  \right| \;  ,  \; h \sim \mu(.|\alpha,K,\mathcal{Z}_n).
\end{align*}
Define \begin{align*}
    \sigma_{\Phi}^2 =  \E \big[ (D_{h_0} \tilde{\Phi} )' \{ \E[ \rho(Y,h_0(X)) \rho(Y,h_0(X))'|W  ] \}^{-1} (D_{h_0} \tilde{\Phi} )     \big]      .
\end{align*}

By Theorem \ref{bvm}$(i)$, we have \begin{align}
 \label{limvarc}   c_{1-\gamma} \xrightarrow{\mathbb{P}} (1- \gamma) \; \:  \text{quantile of} \:   \;  \left|Z \right|  \; \;, \; \; Z \sim N(0, \sigma_{\Phi}^2 ).
\end{align}
By Theorem \ref{bvm}$(ii)$, the distribution of $ \sqrt{n} \big( \mathbf{L}(h_0) - \mathbf{L}\big( \E [h| \mathcal{Z}_n ] \big) \big)$ is asymptotically Gaussian with variance $\sigma_{\Phi}^2$. From this observation and (\ref{limvarc}), it follows that the frequentist coverage of $C_n(\gamma)$ is given by
 \begin{align*}
    \mathbb{P} \big(  \sqrt{n} \left|  \mathbf{L}(h_0) - \mathbf{L}\big( \E [h| \mathcal{Z}_n ]  \big)   \right|  \leq c_{1- \gamma}      \big) = 1- \gamma + o_{\mathbb{P}}(1).
\end{align*}

\end{proof}

\bibliographystyle{style}

\bibliography{main.bib}

\begin{thebibliography}{44}
\newcommand{\enquote}[1]{``#1''}
\expandafter\ifx\csname natexlab\endcsname\relax\def\natexlab#1{#1}\fi

\bibitem[\protect\citeauthoryear{Agapiou, Larsson, and Stuart}{Agapiou
  et~al.}{2013}]{agapiou2013posterior}
\textsc{Agapiou, Sergios, Stig Larsson, and Andrew~M Stuart} (2013):
  \enquote{Posterior contraction rates for the Bayesian approach to linear
  ill-posed inverse problems,} \emph{Stochastic Processes and their
  Applications}, 123 (10), 3828--3860.

\bibitem[\protect\citeauthoryear{Belloni, Chernozhukov, Chetverikov, and
  Kato}{Belloni et~al.}{2015}]{belloni2015some}
\textsc{Belloni, Alexandre, Victor Chernozhukov, Denis Chetverikov, and Kengo
  Kato} (2015): \enquote{Some new asymptotic theory for least squares series:
  Pointwise and uniform results,} \emph{Journal of Econometrics}, 186 (2),
  345--366.

\bibitem[\protect\citeauthoryear{Blundell, Chen, and Kristensen}{Blundell
  et~al.}{2007}]{blundell2007semi}
\textsc{Blundell, Richard, Xiaohong Chen, and Dennis Kristensen} (2007):
  \enquote{Semi-nonparametric IV estimation of shape-invariant Engel curves,}
  \emph{Econometrica}, 75 (6), 1613--1669.

\bibitem[\protect\citeauthoryear{Borovitskiy, Terenin, Mostowsky
  et~al.}{Borovitskiy et~al.}{2020}]{borovitskiy2020matern}
\textsc{Borovitskiy, Viacheslav, Alexander Terenin, Peter Mostowsky, et~al.}
  (2020): \enquote{Mat{\'e}rn Gaussian processes on Riemannian manifolds,}
  \emph{Advances in Neural Information Processing Systems}, 33, 12426--12437.

\bibitem[\protect\citeauthoryear{Breunig}{Breunig}{2020}]{breunig2020specification}
\textsc{Breunig, Christoph} (2020): \enquote{Specification testing in
  nonparametric instrumental quantile regression,} \emph{Econometric Theory},
  36 (4), 583--625.

\bibitem[\protect\citeauthoryear{Castillo and Rousseau}{Castillo and
  Rousseau}{2015}]{castillo2015bernstein}
\textsc{Castillo, Isma{\"e}l and Judith Rousseau} (2015): \enquote{A
  Bernstein--von Mises theorem for smooth functionals in semiparametric
  models,} \emph{The Annals of Statistics}, 43 (6), 2353--2383.

\bibitem[\protect\citeauthoryear{Chen, Chernozhukov, Lee, and Newey}{Chen
  et~al.}{2014}]{chen2014local}
\textsc{Chen, Xiaohong, Victor Chernozhukov, Sokbae Lee, and Whitney~K Newey}
  (2014): \enquote{Local identification of nonparametric and semiparametric
  models,} \emph{Econometrica}, 82 (2), 785--809.

\bibitem[\protect\citeauthoryear{Chen and Christensen}{Chen and
  Christensen}{2015}]{chen2015optimal}
\textsc{Chen, Xiaohong and Timothy~M Christensen} (2015): \enquote{Optimal
  uniform convergence rates and asymptotic normality for series estimators
  under weak dependence and weak conditions,} \emph{Journal of Econometrics},
  188 (2), 447--465.

\bibitem[\protect\citeauthoryear{Chen, Linton, and Van~Keilegom}{Chen
  et~al.}{2003}]{chen2003estimation}
\textsc{Chen, Xiaohong, Oliver Linton, and Ingrid Van~Keilegom} (2003):
  \enquote{Estimation of semiparametric models when the criterion function is
  not smooth,} \emph{Econometrica}, 71 (5), 1591--1608.

\bibitem[\protect\citeauthoryear{Chen and Pouzo}{Chen and
  Pouzo}{2009}]{chen2009efficient}
\textsc{Chen, Xiaohong and Demian Pouzo} (2009): \enquote{Efficient estimation
  of semiparametric conditional moment models with possibly nonsmooth
  residuals,} \emph{Journal of Econometrics}, 152 (1), 46--60.

\bibitem[\protect\citeauthoryear{Chen and Pouzo}{Chen and
  Pouzo}{2012}]{chen2012estimation}
---\hspace{-.1pt}---\hspace{-.1pt}--- (2012): \enquote{Estimation of
  nonparametric conditional moment models with possibly nonsmooth generalized
  residuals,} \emph{Econometrica}, 80 (1), 277--321.

\bibitem[\protect\citeauthoryear{Chen and Pouzo}{Chen and
  Pouzo}{2015}]{chen2015sieve}
---\hspace{-.1pt}---\hspace{-.1pt}--- (2015): \enquote{Sieve Wald and QLR
  inferences on semi/nonparametric conditional moment models,}
  \emph{Econometrica}, 83 (3), 1013--1079.

\bibitem[\protect\citeauthoryear{Chen, Pouzo, and Powell}{Chen
  et~al.}{2019}]{chen2019penalized}
\textsc{Chen, Xiaohong, Demian Pouzo, and James~L Powell} (2019):
  \enquote{Penalized sieve GEL for weighted average derivatives of
  nonparametric quantile IV regressions,} \emph{Journal of Econometrics}, 213
  (1), 30--53.

\bibitem[\protect\citeauthoryear{Chernozhukov and Hong}{Chernozhukov and
  Hong}{2003}]{chernozhukov2003mcmc}
\textsc{Chernozhukov, Victor and Han Hong} (2003): \enquote{An MCMC approach to
  classical estimation,} \emph{Journal of econometrics}, 115 (2), 293--346.

\bibitem[\protect\citeauthoryear{Chernozhukov, Imbens, and Newey}{Chernozhukov
  et~al.}{2007}]{chernozhukov2007instrumental}
\textsc{Chernozhukov, Victor, Guido~W Imbens, and Whitney~K Newey} (2007):
  \enquote{Instrumental variable estimation of nonseparable models,}
  \emph{Journal of Econometrics}, 139 (1), 4--14.

\bibitem[\protect\citeauthoryear{Chernozhukov, Newey, and Santos}{Chernozhukov
  et~al.}{2023}]{chernozhukov2023constrained}
\textsc{Chernozhukov, Victor, Whitney~K Newey, and Andres Santos} (2023):
  \enquote{Constrained conditional moment restriction models,}
  \emph{Econometrica}, 91 (2), 709--736.

\bibitem[\protect\citeauthoryear{De~Hoop, Qiu, and Scherzer}{De~Hoop
  et~al.}{2012}]{de2012local}
\textsc{De~Hoop, Maarten~V, Lingyun Qiu, and Otmar Scherzer} (2012):
  \enquote{Local analysis of inverse problems: H{\"o}lder stability and
  iterative reconstruction,} \emph{Inverse Problems}, 28 (4), 045001.

\bibitem[\protect\citeauthoryear{Dunker, Florens, Hohage, Johannes, and
  Mammen}{Dunker et~al.}{2014}]{dunker2014iterative}
\textsc{Dunker, Fabian, Jean-Pierre Florens, Thorsten Hohage, Jan Johannes, and
  Enno Mammen} (2014): \enquote{Iterative estimation of solutions to noisy
  nonlinear operator equations in nonparametric instrumental regression,}
  \emph{Journal of Econometrics}, 178, 444--455.

\bibitem[\protect\citeauthoryear{Evans}{Evans}{2022}]{evans2022partial}
\textsc{Evans, Lawrence~C} (2022): \emph{Partial differential equations},
  vol.~19, American Mathematical Society.

\bibitem[\protect\citeauthoryear{Florens and Simoni}{Florens and
  Simoni}{2016}]{florens2016regularizing}
\textsc{Florens, Jean-Pierre and Anna Simoni} (2016): \enquote{Regularizing
  priors for linear inverse problems,} \emph{Econometric Theory}, 32 (1),
  71--121.

\bibitem[\protect\citeauthoryear{Ghosal and Van~der Vaart}{Ghosal and Van~der
  Vaart}{2017}]{ghosal2017fundamentals}
\textsc{Ghosal, Subhashis and Aad Van~der Vaart} (2017): \emph{Fundamentals of
  nonparametric Bayesian Inference}, Cambridge University Press.

\bibitem[\protect\citeauthoryear{Gin{\'e} and Nickl}{Gin{\'e} and
  Nickl}{2021}]{gine2021mathematical}
\textsc{Gin{\'e}, Evarist and Richard Nickl} (2021): \emph{Mathematical
  foundations of infinite-dimensional statistical models}, Cambridge university
  press.

\bibitem[\protect\citeauthoryear{Giordano and Nickl}{Giordano and
  Nickl}{2020}]{giordano2020consistency}
\textsc{Giordano, Matteo and Richard Nickl} (2020): \enquote{Consistency of
  Bayesian inference with Gaussian process priors in an elliptic inverse
  problem,} \emph{Inverse Problems}, 36 (8), 085001.

\bibitem[\protect\citeauthoryear{Gugushvili, van~der Vaart, and Yan}{Gugushvili
  et~al.}{2020}]{gugushvili2020bayesian}
\textsc{Gugushvili, Shota, Aad van~der Vaart, and Dong Yan} (2020):
  \enquote{Bayesian linear inverse problems in regularity scales,} in
  \emph{Annales de l’Institut Henri Poincar{\'e}-Probabilit{\'e}s et
  Statistiques}, vol.~56, 2081--2107.

\bibitem[\protect\citeauthoryear{Hall and Horowitz}{Hall and
  Horowitz}{2005}]{hall2005nonparametric}
\textsc{Hall, Peter and Joel~L Horowitz} (2005): \enquote{Nonparametric methods
  for inference in the presence of instrumental variables,} \emph{Annals of
  Statistics}, 33 (6), 2904--2929.

\bibitem[\protect\citeauthoryear{Hanke, Neubauer, and Scherzer}{Hanke
  et~al.}{1995}]{hanke1995convergence}
\textsc{Hanke, Martin, Andreas Neubauer, and Otmar Scherzer} (1995): \enquote{A
  convergence analysis of the Landweber iteration for nonlinear ill-posed
  problems,} \emph{Numerische Mathematik}, 72 (1), 21--37.

\bibitem[\protect\citeauthoryear{Horowitz and Lee}{Horowitz and
  Lee}{2007}]{horowitz2007nonparametric}
\textsc{Horowitz, Joel~L and Sokbae Lee} (2007): \enquote{Nonparametric
  instrumental variables estimation of a quantile regression model,}
  \emph{Econometrica}, 75 (4), 1191--1208.

\bibitem[\protect\citeauthoryear{Kaltenbacher, Nguyen, and
  Scherzer}{Kaltenbacher et~al.}{2021}]{kaltenbacher2021tangential}
\textsc{Kaltenbacher, Barbara, Tram Thi~Ngoc Nguyen, and Otmar Scherzer}
  (2021): \enquote{The tangential cone condition for some coefficient
  identification model problems in parabolic PDEs,} \emph{Time-dependent
  Problems in Imaging and Parameter Identification}, 121--163.

\bibitem[\protect\citeauthoryear{Kaltenbacher, Sch{\"o}pfer, and
  Schuster}{Kaltenbacher et~al.}{2009}]{kaltenbacher2009iterative}
\textsc{Kaltenbacher, Barbara, Frank Sch{\"o}pfer, and Thomas Schuster} (2009):
  \enquote{Iterative methods for nonlinear ill-posed problems in Banach spaces:
  convergence and applications to parameter identification problems,}
  \emph{Inverse Problems}, 25 (6), 065003.

\bibitem[\protect\citeauthoryear{Kato}{Kato}{2013}]{kato2013quasi}
\textsc{Kato, Kengo} (2013): \enquote{Quasi-Bayesian analysis of nonparametric
  instrumental variables models,} \emph{The Annals of Statistics}, 41 (5),
  2359--2390.

\bibitem[\protect\citeauthoryear{Knapik and Salomond}{Knapik and
  Salomond}{2018}]{knapik2018general}
\textsc{Knapik, Bartek and Jean-Bernard Salomond} (2018): \enquote{A general
  approach to posterior contraction in nonparametric inverse problems,}
  \emph{Bernoulli}, 24 (3), 2091--2121.

\bibitem[\protect\citeauthoryear{Knapik, van~der Vaart, and van Zanten}{Knapik
  et~al.}{2011}]{knapik2011bayesian}
\textsc{Knapik, BT, AW~van~der Vaart, and JH~van Zanten} (2011):
  \enquote{Bayesian inverse problems with Gaussian priors,} \emph{The Annals of
  Statistics}, 39 (5), 2626--2657.

\bibitem[\protect\citeauthoryear{Liao and Jiang}{Liao and
  Jiang}{2011}]{liao2011posterior}
\textsc{Liao, Yuan and Wenxin Jiang} (2011): \enquote{Posterior consistency of
  nonparametric conditional moment restricted models,} \emph{The Annals of
  Statistics}, 39 (6), 3003--3031.

\bibitem[\protect\citeauthoryear{Mair and Ruymgaart}{Mair and
  Ruymgaart}{1996}]{mair1996statistical}
\textsc{Mair, Bernard~A and Frits~H Ruymgaart} (1996): \enquote{Statistical
  inverse estimation in Hilbert scales,} \emph{SIAM Journal on Applied
  Mathematics}, 56 (5), 1424--1444.

\bibitem[\protect\citeauthoryear{Math{\'e} and Pereverzev}{Math{\'e} and
  Pereverzev}{2001}]{mathe2001optimal}
\textsc{Math{\'e}, Peter and Sergei~V Pereverzev} (2001): \enquote{Optimal
  discretization of inverse problems in Hilbert scales. Regularization and
  self-regularization of projection methods,} \emph{SIAM Journal on Numerical
  Analysis}, 38 (6), 1999--2021.

\bibitem[\protect\citeauthoryear{Monard, Nickl, and Paternain}{Monard
  et~al.}{2021{\natexlab{a}}}]{monard2021consistent}
\textsc{Monard, Fran{\c{c}}ois, Richard Nickl, and Gabriel~P Paternain}
  (2021{\natexlab{a}}): \enquote{Consistent Inversion of Noisy Non-Abelian
  X-Ray Transforms,} \emph{Communications on Pure and Applied Mathematics}, 74
  (5), 1045--1099.

\bibitem[\protect\citeauthoryear{Monard, Nickl, and Paternain}{Monard
  et~al.}{2021{\natexlab{b}}}]{monard2021statistical}
---\hspace{-.1pt}---\hspace{-.1pt}--- (2021{\natexlab{b}}):
  \enquote{Statistical guarantees for Bayesian uncertainty quantification in
  nonlinear inverse problems with Gaussian process priors,} \emph{The Annals of
  Statistics}, 49 (6), 3255--3298.

\bibitem[\protect\citeauthoryear{Newey and Powell}{Newey and
  Powell}{2003}]{newey2003instrumental}
\textsc{Newey, Whitney~K and James~L Powell} (2003): \enquote{Instrumental
  variable estimation of nonparametric models,} \emph{Econometrica}, 71 (5),
  1565--1578.

\bibitem[\protect\citeauthoryear{Ray}{Ray}{2013}]{ray}
\textsc{Ray, Kolyan} (2013): \enquote{Bayesian inverse problems with
  non-conjugate priors,} \emph{Electronic Journal of Statistics}, 7,
  2516--2549.

\bibitem[\protect\citeauthoryear{Reed}{Reed}{2012}]{reed2012methods}
\textsc{Reed, Michael} (2012): \emph{Methods of modern mathematical physics:
  Functional analysis}, Elsevier.

\bibitem[\protect\citeauthoryear{Severini and Tripathi}{Severini and
  Tripathi}{2012}]{severini2012efficiency}
\textsc{Severini, Thomas~A and Gautam Tripathi} (2012): \enquote{Efficiency
  bounds for estimating linear functionals of nonparametric regression models
  with endogenous regressors,} \emph{Journal of Econometrics}, 170 (2),
  491--498.

\bibitem[\protect\citeauthoryear{Tropp}{Tropp}{2012}]{tropp}
\textsc{Tropp, Joel~A} (2012): \enquote{User-friendly tail bounds for sums of
  random matrices,} \emph{Foundations of computational mathematics}, 12,
  389--434.

\bibitem[\protect\citeauthoryear{Van Der~Vaart}{Van
  Der~Vaart}{1991}]{van1991differentiable}
\textsc{Van Der~Vaart, Aad} (1991): \enquote{On differentiable functionals,}
  \emph{The Annals of Statistics}, 178--204.

\bibitem[\protect\citeauthoryear{Van Der~Vaart and Wellner}{Van Der~Vaart and
  Wellner}{1996}]{weakc}
\textsc{Van Der~Vaart, Aad~W and Jon~A Wellner} (1996): \emph{Weak convergence
  and empirical processes: with applications to statistics}, vol.~3, Springer.

\end{thebibliography}

\end{document}